\documentclass[12pt]{article}
\usepackage{amsmath,amssymb,amsfonts,bbm,bm,cases,verbatim,multirow,color,xcolor}
\usepackage[ruled,vlined]{algorithm2e}
\usepackage{epic,eepic,psfrag,epsfig}
\usepackage{graphicx}
\usepackage{amsthm}
\usepackage{natbib}
\usepackage{enumerate}
\usepackage{array}
\usepackage[a4paper,top=1in,bottom=1in,left=1in,right=1in]{geometry}
\RequirePackage[colorlinks,citecolor={blue!60!black},urlcolor={blue!70!black},linkcolor={red!60!black},breaklinks,hypertexnames=false]{hyperref}

\renewcommand{\epsilon}{\varepsilon}
\newcommand{\ep}{\varepsilon}
\definecolor{ed}{RGB}{225,0,0}

\DeclareMathOperator*{\argmin}{argmin}
\DeclareMathOperator*{\argmax}{argmax}

\newtheorem{thm}{Theorem}
\newtheorem*{thm*}{Theorem}
\newtheorem{prop}[thm]{Proposition}
\newtheorem{lemma}[thm]{Lemma}
\newtheorem{cor}[thm]{Corollary}

\newcommand{\beq}{\begin{equation}}
\newcommand{\eeq}{\end{equation}}
\newcommand{\RR}{\mathbb{R}}

\DeclareMathOperator*{\sargmax}{sargmax}

\graphicspath{
{../code/Experiment 1 - plot EM obj/}
{../code/Experiment 2 - EM conv rate/}
}

\title{Sharp-SSL: Selective high-dimensional axis-aligned random projections for semi-supervised learning}

\author{Tengyao Wang$^*$, Edgar Dobriban$^\dagger$, Milana Gataric$^\ddagger$ \\
and Richard J. Samworth$^\ddagger$ \medskip \\{$^*$Department of Statistics, London School of Economics}\\{$^\dagger$Department of Statistics and Data Science, University of Pennsylvania}\\{$^\ddagger$Statistical Laboratory, University of Cambridge}}

\begin{document}
\maketitle 


\begin{abstract}
We propose a new method for high-dimensional semi-supervised learning problems based on the careful aggregation of the results of a low-dimensional procedure applied to many axis-aligned random projections of the data.  Our primary goal is to identify important variables for distinguishing between the classes; existing low-dimensional methods can then be applied for final class assignment.  Motivated by a generalized Rayleigh quotient, we score projections according to the traces of the estimated whitened between-class covariance matrices on the projected data.  This enables us to assign an importance weight to each variable for a given projection, and to select our signal variables by aggregating these weights over high-scoring projections.  Our theory shows that the resulting \texttt{Sharp-SSL} algorithm is able to recover the signal coordinates with high probability when we aggregate over sufficiently many random projections and when the base procedure estimates the whitened between-class covariance matrix sufficiently well.  The Gaussian EM algorithm is a natural choice as a base procedure, and we provide a new analysis of its performance in semi-supervised settings that controls the parameter estimation error in terms of the proportion of labeled data in the sample.  Numerical results on both simulated data and a real colon tumor dataset support the excellent empirical performance of the method.        
\end{abstract}

\section{Introduction}

Semi-supervised learning, where we attempt to assign observations to one of finitely many groups based on partially-labeled training data, 
represents a core modern statistical challenge.  It is sufficiently general to incorporate, at either extreme, the unsupervised case of no labeled training data (clustering) and the supervised setting of fully-labeled training data (classification).  Such tasks abound in many application areas, including genomics \citep[e.g.,][]{eisen1998cluster}, image processing \citep{jain1996image,cheplygina2019not}, natural language processing \citep{liang2005semi,turian2010word} and anomaly detection \citep{akcay2019ganomaly,wang2019semi}.  Entry points to the literature on semi-supervised learning include \citet{zhu2005semi}, \citet{zhu2009introduction}, \citet{chapelle2006semi} and \citet{van2020survey}.  For introductions to clustering, see \citet{xu2005survey}, \citet{kaufman2009finding} and \citet{xu2015comprehensive}, and for classification, see \citet{devroye2013probabilistic} and \citet{hastie2009elements}.

\sloppy A common feature of contemporary semi-supervised learning problems is high-dimensionality, since we may record many covariates having a possible association with the labels corresponding to different observations.  This represents a significant challenge, as can be seen by considering a simple two-class problem with more covariates than observations.  For any given assignment of class labels, if no subset of $n_0$ observations lies in an $(n_0-2)$-dimensional affine space, then we can find hyperplanes with orthogonal normal vectors, each of which achieves zero training error (in other words, they perfectly separate the classes).  Nevertheless, even in the simple setting where the true Bayes decision boundary is linear, many such hyperplanes may be little better than a random guess on test data.

An appealing approach to tackling high-dimensionality is via random projections into lower-dimensional spaces.  Such projections may almost preserve the pairwise distances between observations, as seen from the Johnson--Lindenstrauss lemma \citep{lindenstrauss1984extensions,dasgupta2003elementary}.  Moreover, in cases where we have reason to believe that only a relatively small proportion of the variables recorded are relevant for the learning task, we can choose our random projections to be axis-aligned in order to preserve this structure.  A third benefit is the possibility of aggregating results over multiple random projections, though this must be done with care so as to avoid noise accumulation.  These attractions have meant that random projections have now been employed in many high-dimensional statistical problems, including precision matrix estimation \citep{marzetta2011random}, two-sample mean testing \citep{lopes2011more}, classification \citep{durrant2015random,cannings2017random}, (sparse) principal component analysis \citep{yang2021reduce,GWS2020}, linear regression \citep{thanei2017random,slawski2018principal,dobriban2019asymptotics,ahfock2021statistical}, clustering \citep{dasgupta1999learning,fern2003random,han2015hidden,yellamraju2018clusterability,anderlucci2022high} and dimensionality reduction \citep{bingham2001random,reeve2022heterogeneous}.  See \citet{cannings2021random} for a review of recent developments in the area.

In this paper, we propose a new method, called \texttt{Sharp-SSL} (short for \textbf{S}elective \textbf{h}igh-dimensional, \textbf{a}xis-aligned \textbf{r}andom \textbf{p}rojections for \textbf{S}emi-\textbf{S}upervised \textbf{L}earning).  Our primary goal is to identify a small subset of variables that are particularly helpful for label assignment; existing low-dimensional methods can then be used to complete the learning task.  To this end, we generate a large number of axis-aligned random projections, and apply a base learning procedure such as a semi-supervised version of the Gaussian Expectation--Maximization (EM) algorithm to our projected data.  Motivated by the notion of a generalized Rayleigh quotient (see~\eqref{Eq:GRQ} below for a formal definition), and to avoid the noise accumulation issue mentioned above, we score the projections by computing the trace of the corresponding estimated whitened between-class covariance matrices.  This enables us to assign an importance weight to each variable for a given projection, and we select our signal variables by aggregating these importance weights over the high-scoring projections.  See Section~\ref{Sec:Methodology} for a more detailed description of our methodology.

Section~\ref{Sec:Theory} is devoted to a theoretical analysis of our \texttt{Sharp-SSL} algorithm.  We first show in Theorem~\ref{Thm:Meta} that provided the low-dimensional base learning procedure satisfies a guarantee on the proximity of the estimated whitened between-class covariance matrix to its population analogue, the corresponding high-dimensional semi-supervised learning algorithm can recover the signal coordinates with high probability when we aggregate over sufficiently many random projections.  It turns out that both Linear Discriminant Analysis and an EM algorithm are examples of low-dimensional learning procedures that satisfy this proximity guarantee, as we prove in Theorems~\ref{Thm:LDA} and~\ref{Thm:LowDimEM} respectively.  The latter is particularly challenging, and one of the main novel contributions of our analysis is to provide a guarantee on the performance of a $d$-dimensional Gaussian EM algorithm in a semi-supervised setting.  In particular, we control the parameter estimation error in terms of the proportion of labeled data in the sample, showing that with a sample size of $n$ it smoothly interpolates between the $(d/n)^{1/4}$ rate for unsupervised learning and the $(d/n)^{1/2}$ rate for fully-labeled data, up to logarithmic factors.  An advantage of the modular approach to our analysis is that it illustrates the way in which the \texttt{Sharp-SSL} algorithm can be combined with different base learning algorithms to adapt to different problem settings and reflect the preferences of the practitioner. 

In Section~\ref{Sec:Numerics}, we study the numerical performance of the \texttt{Sharp-SSL} algorithm.  Our first goal, in Section~\ref{SubSec:TuningParameters}, is to study the effect of the choices of input parameters to our method, which allows us to recommend sensible default choices for application in our subsequent comparisons.  Section~\ref{SubSec:Comparison} presents the results of a simulation study involving the \texttt{Sharp-SSL} method, as well as five alternative approaches, on high-dimensional clustering tasks (since not all of the competing methods are able to leverage partial label information).  We find that the \texttt{Sharp-SSL} algorithm is able to attain a misclustering rate very close to that of the optimal Bayes classifier, even with only around 50 observations per cluster, in settings where these alternative techniques may perform poorly.  In Section~\ref{SubSec:gammafraction}, we investigate the extent to which the different versions of the \texttt{Sharp-SSL} method are able to leverage partial label information.  The results here are consistent with the phase transition phenomenon articulated by our theory.  Finally, in Section~\ref{SubSec:EmpiricalData}, we apply the \texttt{Sharp-SSL} algorithm, as well as the other methods from our simulation study, on a colon tumor dataset, where we withhold the true labels from the algorithms in order to assess performance.  Our analysis supports the ability of the \texttt{Sharp-SSL} algorithm to identify signal coordinates (genes) that are useful for identifying patients with and without tumors.

In the broader literature on high-dimensional learning problems, a large number of methods have been developed to leverage sparse low-dimensional structures for both clustering \citep{WittenTibshirani2010,azizyan2013minimax,wasserman2014feature,azizyan2015efficient,JinWang2016,verzelen2017detection,loffler2020computationally,loffler2021optimality} and classification \citep{cai2011direct,witten2011penalized,mai2012direct,tony2019high}.  These methods are not designed for partially-labeled (semi-supervised) settings.  
Another common approach is to project the data into the span of the top few principal components, and run a standard low-dimensional method such as $k$-means clustering or the EM algorithm \citep{butler2018integrating}.  This approach can fail if the directions of largest variation in the data are not aligned with the directions separating the clusters.  Finally, recent developments in other aspects of semi-supervised learning include self-training \citep{oymak2020statistical}, mean estimation \citep{zhang2019semi}, choice of $k$ in $k$-nearest neighbour classification \citep{cannings2020local} and linear regression \citep{chakrabortty2018efficient}.

Proofs of all of our results are provided in Section~\ref{Sec:Proofs}, and we conclude this introduction with some notation used throughout the paper.  We write $\mathbb{S}^{d\times d}$ for the set of $d$-dimensional symmetric matrices, write $\mathbb{S}_+^{d \times d}$ for the subset that are invertible, and write $\mathbb{S}^{d\times d}_{K-1}$ the subset of matrices in $\mathbb{S}^{d\times d}$ of rank at most $K-1$. 
We write $\mathbb{R}^{d\times d}$ for the set of $d$-dimensional matrices.
For $p \geq d$, let $\mathbb{O}^{p\times d}$ denote the set of $p\times d$ matrices with orthonormal columns.  The Euclidean norm is denoted by $\|\cdot\|$, and the operator norm of a matrix is denoted by $\|\cdot\|_{\mathrm{op}}$, so that $\|A\|_{\mathrm{op}} := \sup_{\{x:\|x\| = 1\}} \|Ax\|$.  
Given two sequences $(a_n)$ and $(b_n)$, we write $a_n \lesssim b_n$ when there exists a universal constant $C > 0$ such that $a_n \leq Cb_n$, and, given an additional problem parameter $R$, we write $a_n \lesssim_R b_n$ when there exists $C > 0$, depending only on $R$, such that $a_n \leq Cb_n$. 

For any set $S \subseteq \mathbb{R}^d$ and $d\leq |S|$, we write $\binom{S}{d}:=\{A\subseteq S: |A|=d\}$.  If $S \subseteq \mathbb{R}^d$, we define $\sargmax S$ to be the smallest element in the $\argmax$ in the lexicographic order. 
For a positive integer $k$, we define $[k]:=\{1,\ldots, k\}$.
For a vector $v = (v_1,\ldots, v_k)^\top\in \mathbb{R}^k$, and $j\in [k]$, we define $v_{-j} = (v_1,\ldots, v_{j-1}, v_{j+1}, \ldots, v_k)^\top \in \mathbb{R}^{k-1}$.

\section{The \texttt{Sharp-SSL} algorithm}
\label{Sec:Methodology}

In this section, we describe in detail the \texttt{Sharp-SSL} algorithm for $K$-class semi-supervised learning, with $K \geq 2$.  We aim to provide a unified treatment of clustering, semi-supervised learning and classification.  To this end, we assume that for $i \in [n]$, the observation $x_i \in \mathbb{R}^p$ has a true label $y_i^* \in [K]$, but it may be the case that we do not observe $y_i^*$.  
Instead, we assume that our observed label $y_i$ takes values in $[K] \cup \{0\}$, where $y_i := y_i^*$ when the true class label is observed, and $y_i := 0$ otherwise.  
Thus, our data can be regarded as $(x_1,y_1),\ldots,(x_n,y_n) \in \mathbb{R}^p \times ([K] \cup \{0\})$, and our goal is to construct a \emph{data-dependent classifier}\footnote{It is convenient to use the term `classifier' here, even though some or all of the labels may be unobserved.}, i.e.~a Borel measurable function $C:\mathbb{R}^p \times \bigl(\mathbb{R}^p \times ([K] \cup \{0\})\bigr)^n \rightarrow [K]$, with the interpretation that $C\bigl(x;(x_1,y_1),\ldots,(x_n,y_n)\bigr)$ is the predicted class of $x \in \mathbb{R}^p$.  

To motivate our \texttt{Sharp-SSL} algorithm, it is instructive first to consider a canonical Gaussian classification problem, where our data can be regarded as $n$ independent realizations of a pair $(X,Y)$ taking values in $\mathbb{R}^p \times [K]$, with prior probability $\pi_k := \mathbb{P}(Y = k)$ for the $k$th class and $X \mid Y=k \sim \mathcal{N}_p(\nu_k,\Sigma_{\mathrm{w}})$, for class means $\nu_1,\ldots,\nu_K \in \mathbb{R}^p$ and \emph{within-class covariance matrix} $\Sigma_{\mathrm{w}} \in \mathbb{S}_+^{p \times p}$.  Let
$\nu := \sum_{k=1}^K \pi_k\nu_k \in \mathbb{R}^p$
denote the grand population mean, let 
\beq\label{sb}
\Sigma_{\mathrm{b}} := \sum_{k=1}^K \pi_k (\nu_k - \nu)(\nu_k - \nu)^\top \in \mathbb{S}_{K-1}^{p \times p}
\eeq
denote the \emph{between-class covariance matrix}, 
and consider $D \in \mathbb{O}^{p \times (K-1)}$ with a column space spanned by $\Sigma_{\mathrm{w}}^{-1}(\nu_1-\nu),\ldots,\Sigma_{\mathrm{w}}^{-1}(\nu_K-\nu)$. Observe that for $k \neq \ell$,
\[
\log\biggl\{\frac{\mathbb{P}(Y = k\mid X=x)}{\mathbb{P}(Y = \ell \mid X=x)} \biggr\} = \log \Bigl(\frac{\pi_k}{\pi_\ell}\Bigr) - \frac{1}{2}(\nu_k + \nu_\ell)^\top \Sigma_{\mathrm{w}}^{-1}(\nu_k - \nu_\ell) + x^\top \Sigma_{\mathrm{w}}^{-1} (\nu_k - \nu_\ell),
\]
from which we deduce that this likelihood ratio, and hence the Bayes classifier $x \mapsto \argmax_{k \in [K]} \mathbb{P}(Y = k\mid X=x)$, only depends on $x$ through $D^\top x$.  Thus, for the purposes of classification, no signal would be lost (and the noise would be reduced) if $X$ were replaced with $D^\top X$.  
 
In high-dimensional settings with $p \gg n$, the matrix $\Sigma_{\mathrm{w}}^{-1}$ is not consistently estimable in general, but we can nevertheless make progress if the vectors $\Sigma_{\mathrm{w}}^{-1}(\nu_1 - \nu),\ldots,\Sigma_{\mathrm{w}}^{-1}(\nu_K - \nu)$ are sparse. In other words, writing $S_0$ for the union of the set of coordinates for which these vectors are non-zero, we suppose that $|S_0| \ll p$; this is a very common assumption in high-dimensional LDA \citep[e.g.][]{cai2011direct,witten2011penalized,mai2012direct,tony2019high}.

In such a setting, the column space of $D$ has a sparse basis, so it is natural to consider projecting the data onto a small subset of its coordinates.  For $d \in [p]$, define the set of axis-aligned projection matrices $\mathcal{P}_d := \bigl\{P \in \{0,1\}^{d \times p}: PP^\top = I_d\bigr\}$, i.e.~the set of binary $d \times p$ matrices with orthonormal rows.  By the argument above, if $d \geq |S_0|$ then there exists $P^* \in \mathcal{P}_d$ such that the error of the Bayes classifier is unchanged by projecting the data along $P^*$. 
In practice, it would typically be computationally too expensive to enumerate through all $p(p-1)\cdots(p-d+1)$ 
axis-aligned projections. 
Instead, we consider a randomly chosen subset of projections within $\mathcal{P}_d$.  An axis-aligned projection chosen uniformly at random is unlikely to capture all the signal coordinates $S_0$, but by aggregating over a carefully-chosen subset of these random projections, we can nevertheless recover the set of signal coordinates under suitable conditions; see Theorem~\ref{Thm:Meta} below.  To describe our method for choosing good projections, for $V \in \mathbb{O}^{p\times d}$, we define the \emph{generalized Rayleigh quotient} along $V$ by
\begin{equation}
\label{Eq:GRQ}
J(V;\Sigma_{\mathrm{b}},\Sigma_{\mathrm{w}}) := \mathrm{tr}\{(V^\top\Sigma_{\mathrm{w}} V)^{-1}(V^\top \Sigma_{\mathrm{b}} V)\}.
\end{equation}
Proposition~\ref{Prop:Motivation} below motivates seeking to choose projections to maximize the generalized Rayleigh quotient by showing that the column span of any maximizer $J(V;\Sigma_{\mathrm{b}},\Sigma_{\mathrm{w}})$ over $V\in\mathbb{O}^{p\times d}$ must contain the column space of $D$.
\begin{prop}
\label{Prop:Motivation}
Let $K \geq 2$ and $d\geq K-1$. Assume that the convex hull of $\nu_1,\ldots,\nu_K$ is $(K-1)$-dimensional, and let $V^* \in \argmax_{V\in\mathbb{O}^{p\times d}} J(V;\Sigma_{\mathrm{b}},\Sigma_{\mathrm{w}})$.  Then the column space of $V^*$ contains the eigenspace corresponding to the $K-1$ non-zero eigenvalues\footnote{Even though $\Sigma_{\mathrm{w}}^{-1}\Sigma_{\mathrm{b}}$ is not guaranteed to be symmetric, it is similar (i.e.~conjugate) to the symmetric matrix $\Sigma_{\mathrm{w}}^{-1/2}\Sigma_{\mathrm{b}} \Sigma_{\mathrm{w}}^{-1/2}$, so has real eigenvalues and eigenvectors.} of $\Sigma_{\mathrm{w}}^{-1}\Sigma_{\mathrm{b}}$, which is equal to the space spanned by $\bigl(\Sigma_{\mathrm{w}}^{-1}(\nu_k-\nu):k \in [K]\bigr)$.
\end{prop}
Based on Proposition~\ref{Prop:Motivation}, a natural conceptual approach to maximizing the generalized Rayleigh quotient is to compute the leading $(K-1)$-dimensional eigenspace of $\Sigma_{\mathrm{w}}^{-1}\Sigma_{\mathrm{b}}$.  This strategy, however, runs into difficulties when we replace these population quantities with their sample versions in the setting of the opening paragraph of this section.  More precisely, writing $n_k:=\sum_{i=1}^n \mathbbm{1}_{\{ y_i = k \}}$ for $k \in [K]$, as well as 
\begin{align*}
\tilde\Sigma_{\mathrm{w}} &:= \frac{1}{n}\sum_{k=1}^K\sum_{i=1}^n (x_i-\hat{\nu}_k)(x_i-\hat{\nu}_k)^{\top}\mathbbm{1}_{\{y_i=k\}} \in \mathbb{R}^{p \times p} \\
\tilde\Sigma_{\mathrm{b}} &:= \sum_{k=1}^K \frac{n_k}{n} (\hat{\nu}_k - \hat{\nu})(\hat{\nu}_k - \hat{\nu})^{\top} \in \mathbb{R}^{p \times p},
\end{align*}
for the sample versions of the within-class and between-class covariance matrices respectively, the matrix $\tilde\Sigma_{\mathrm{w}}$ is not invertible whenever $p > n$.  Fortunately, though, this issue can be resolved by working with the projected data, as long as we choose $d \leq n-K$: the projected data $\{PX_i:i \in [n]\}$ has within-class covariance matrix $P\Sigma_{\mathrm{w}}P^\top \in \mathbb{R}^{d \times d}$ and between-class covariance matrix $P\Sigma_{\mathrm{b}}P^\top \in \mathbb{R}^{d \times d}$, so with probability one, the sample version $P\tilde{\Sigma}_{\mathrm{w}}P^\top$ is invertible.

Returning to the general setting of the opening paragraph of this section, then, we seek projections $P$ with large $J(P^\top;\tilde\Sigma_{\mathrm{b}},\tilde\Sigma_{\mathrm{w}}) = \mathrm{tr}\bigl\{(P\tilde \Sigma_{\mathrm{w}}P^\top)^{-1}(P\tilde \Sigma_{\mathrm{b}}P^\top)\bigr\}$. 
To this end, for fixed $A,B\in\mathbb{N}$, we sample a set of projections  $\{P^{a,b}: a \in [A], b \in [B]\}$ uniformly at random from $\mathcal{P}_d$.
For each $a$ and $b$, we apply a low-dimensional base algorithm $\psi:\bigl(\mathbb{R}^d \times ([K] \cup \{0\})\bigr)^n \rightarrow \mathbb{S}_+^{d\times d}$ to the projected data $(P^{a,b}x_1,y_1),\ldots, (P^{a,b}x_n,y_n)$ to obtain an estimator $\hat{Q}^{a,b}$ of $(P^{a,b} \Sigma_{\mathrm{w}} P^{a,b,\top})^{-1} (P^{a,b} \Sigma_{\mathrm{b}} P^{a,b,\top})$, the \emph{whitened between-class covariance matrix} of the projected data.  We assume throughout for convenience that $\psi$ is \emph{permutation equivariant} in the sense that $\psi\bigl( (\Pi z_1,y_1), \ldots, (\Pi z_n,y_n)\bigr) 
= \Pi \psi\bigl( (z_1,y_1), \ldots, (z_n,y_n)\bigr) \Pi^\top$ for every permutation matrix $\Pi \in \mathbb{R}^{d \times d}$. 
One choice for the base algorithm is to set $\hat{Q}^{a,b} = (\hat\Sigma^{a,b}_{\mathrm{w}})^{-1}\hat\Sigma^{a,b}_{\mathrm{b}}$, where $\hat\Sigma^{a,b}_\mathrm{w}$  and $\hat\Sigma^{a,b}_{\mathrm{b}}$ are estimated projected within- and between-class (or cluster) covariance matrices. 

To select projections, for each $a \in [A]$, we define
\[
b^*(a) := \sargmax_{b\in [B]}  \mathrm{tr}(\hat Q^{a,b}),
\]
and select $P^{a,b^*(a)}$.  The main rationale for dividing the projections into $A$ groups and selecting one within each group---as opposed to selecting the $A$ projections with the largest values of $\mathrm{tr}(\hat{Q}^{a,b})$---is that, conditional on the original data, the selected projections are independent and identically distributed.  This facilitates our theoretical analysis by enabling the application of concentration inequalities in the proof of Theorem~\ref{Thm:Meta}.

The diagonal entries of $\{\hat Q^{a,b^*(a)}:a \in [A]\}$ measure the importance of the projected variables for the semi-supervised learning task.  These can then be converted into importance scores for the original variables by `back-projecting' into the higher-dimensional space, i.e.~by forming the vector $\hat w = (\hat w_1,\ldots,\hat w_p)^\top \in \mathbb{R}^p$ given by
\begin{equation*}\label{eq:aggregation}
\hat w_{j} :=  \frac{1}{A} \sum_{a=1}^{A} \left[P^{a,b^*(a),\top}\hat Q^{a,b^*(a)} P^{a,b^*(a)}\right]_{j,j}, \quad j \in [p].
\end{equation*}
Finally, we rank the variables by their importance scores, and our estimate $\hat{S}$ of the set of signal coordinates is given by the largest $\ell$ entries in $\hat w$, breaking ties arbitrarily if necessary, where $\ell \in [p]$ is specified by the practitioner.  
Pseudocode for the \texttt{Sharp-SSL} procedure is given in Algorithm~\ref{Algo:Sharp-SSL}.

\begin{algorithm}
\SetAlgoLined
\DontPrintSemicolon
\KwIn{Data $(x_1,y_1), \ldots, (x_n,y_n)$ $\in$ $\mathbb{R}^p\times$ $([K]\cup\{0\})$ (where $y_i=0$ denotes a missing label); \newline
Projected dimension $d \in [\min(p,n-K)]$, number of selected signal coordinates $\ell \in [p]$; \newline
Number $A \in \mathbb{N}$ of groups of projections, number $B \in \mathbb{N}$ of projections in each group; \newline 
Permutation equivariant base algorithm $\psi: \bigl(\mathbb{R}^d\times ([K]\cup\{0\})\bigr)^n$ $\to \mathbb{S}_+^{d\times d}$.}

Generate axis-aligned random projections $\{P^{a,b}: a \in [A], b \in [B]\}$ independently and uniformly from $\mathcal{P}_d$.\;
\For{$a \in [A]$}{
      \For{$b \in [B]$}{
      	  Let $ \hat{Q}^{a,b} := \psi\bigl((P^{a,b}x_1,y_1),\ldots,(P^{a,b}x_n,y_n)\bigr)$.
      }   
      Set $b^*(a) := \sargmax_{b\in[B]}  \mathrm{tr}(\hat Q^{a,b})$.
} 
Let $\hat w=(\hat w_1,\ldots,\hat w_p)^\top$, where
$\hat w_j :=  \frac{1}{A} \sum_{a=1}^{A} [P^{a,b^*(a),\top}\hat Q^{a,b^*(a)} P^{a,b^*(a)}]_{j,j}$.\;
\KwOut{$\hat S \subseteq [p]$, defined as the index set of the $\ell$ largest components of $\hat w$, breaking ties randomly.}
\caption{\texttt{Sharp-SSL}: Clustering and semi-supervised learning via ensembles of axis-aligned random projections.}
\label{Algo:Sharp-SSL}
\end{algorithm}

After applying Algorithm~\ref{Algo:Sharp-SSL} to obtain an estimated set $\hat{S}$ of signal variables, we can then apply any existing semi-supervised learning method for low-dimensional data with input $(P_{\hat S} x_{i,}, y_i)_{i\in[n]}$, 
where $P_{\hat S}$ denotes the projection onto the coordinates in $\hat S$.  

\subsection{Base learning methods}

Algorithm~\ref{Algo:Sharp-SSL} relies on a  base learning method for low-dimensional data to estimate the projected whitened between-class covariance matrix from the projected data.  When all or almost all of the input data are labeled, we can use the procedure outlined in Algorithm~\ref{Algo:LDA}, which ignores any unlabeled data, for this purpose.  
On the other hand, when we have a substantial amount of unlabeled data, Algorithm~\ref{Algo:LDA} may be inaccurate.  
In such circumstances, it may be preferable to use Algorithm~\ref{Algo:EM}, which runs an \emph{Expectation--Maximization} (EM) procedure to predict the unobserved labels and subsequently estimate the whitened between-class covariance matrix.  More precisely, from $M$ random initializations of the cluster means and the within-class covariance matrix, Algorithm~\ref{Algo:EM} uses the EM algorithm to update these quantities, and thereby compute the whitened between-cluster sample covariance matrix estimators $\bigl\{\hat Q^{[m]}=(\hat\Sigma_{\mathrm{w}}^{[m]})^{-1}(\hat\Sigma_{\mathrm{b}}^{[m]}): m \in [M]\bigr\}$.  
We select $\hat{m} \in [M]$ such that $\hat{Q}^{[\hat{m}]}$ is in best agreement with results from the other EM runs; this is made precise in~\eqref{Eq:BestInitialization}.  

The algorithm also allows the practitioner to incorporate prior knowledge about the true cluster means and within-cluster covariance matrices, both through optimizing over a restricted constraint set $\mathcal{C}$ in the M step of the EM algorithm, and through the choice of a distribution supported on $\mathcal{C}$ for the initialization of these quantities.  
An alternative to the EM algorithm for unsupervised learning would be to apply $k$-means clustering as a base procedure.  Previous studies have suggested that these approaches have comparable empirical performance \citep[e.g.,][and references therein]{de2008clustering,rodriguez2019clustering}, but the EM algorithm is more amenable to theoretical analysis in our setting.
\begin{algorithm}
\SetAlgoLined
\IncMargin{1em}
\DontPrintSemicolon
\KwIn{$(z_1,y_1),\ldots,(z_n,y_n) \in \mathbb{R}^d \times ([K] \cup \{0\})$, where the convex hull of $(z_i)_{i:y_i=k}$ is $d$-dimensional for at least one $k \in [K]$.}

Compute $\hat{\mu}:= n^{-1}\sum_{i=1}^n z_i$\;

\For{$k\in[K]$}{
Set $n_k:=|\{i: y_i = k\}|$ and $n' := \sum_{k=1}^K n_k$\;
Compute $\hat{\mu}_k:=n_k^{-1}\sum_{i: y_i=k} z_i$ (with the convention that $\hat{\mu}_k:=0$ if $n_k=0$).\;
}

Compute the within-class and between-class covariance matrices as 
\begin{equation}
\label{Eq:withinbetween}
\hat\Sigma_{\mathrm{w}} := \frac{1}{n'}\sum_{i=1}^n (z_i - \hat{\mu}_{y_i})(z_i - \hat{\mu}_{y_i})^\top \quad\text{and}\quad \hat\Sigma_{\mathrm{b}} := \sum_{k=1}^K \frac{n_k}{n'}(\hat{\mu}_k - \hat{\mu})(\hat{\mu}_k - \hat{\mu})^\top.
\end{equation}

\KwOut{$\hat Q := \hat\Sigma_{\mathrm{w}}^{-1}\hat\Sigma_{\mathrm{b}}$.}
\caption{Base learning using only labeled data}
\label{Algo:LDA}
\end{algorithm}

\begin{algorithm}[htbp]
\SetAlgoLined
\IncMargin{1em}
\DontPrintSemicolon
\KwIn{Data $(z_1,y_1),\ldots,(z_n,y_n) \in \mathbb{R}^d\times ([K]\cup\{0\})$. A constraint set $\mathcal{C} \subseteq (\mathbb{R}^d)^K \times \mathbb{S}_+^{d\times d}$ and a probability distribution $\pi_{\mathcal{C}}$ supported on $\mathcal{C}$. Number of random initializations $M$. Number of iterations $T$.}

\For{$m \in [M]$}{

Randomly sample $(\hat{\mu}_{1},\ldots,\hat{\mu}_K,\hat\Sigma_{\mathrm{w}}) \sim \pi_{\mathcal{C}}$.\;

\For{$t \in [T]$}{
  (E step) Compute the soft-label matrix $(L_{i,k})_{i\in[n], k\in[K]}$
  \begin{equation}
\label{Eq:EStep}
L_{i,k} := \biggl(\frac{e^{-\frac{1}{2}(z_i-\hat{\mu}_k)^\top \hat\Sigma_{\mathrm{w}}^{-1}(z_i-\hat{\mu}_k)}}{\sum_{\ell=1}^K e^{-\frac{1}{2}(z_i-\hat{\mu}_\ell)^\top \hat\Sigma_{\mathrm{w}}^{-1}(z_i-\hat{\mu}_\ell)}}\biggr)\mathbbm{1}_{\{y_i=0\}} + \mathbbm{1}_{\{y_i = k\}}.
\end{equation}

(M step) Update parameter estimates by
\begin{align}
\label{Eq:MStep}
(\hat{\mu}_1&,\ldots,\hat{\mu}_K,\hat\Sigma_{\mathrm{w}}) \nonumber \\
&:= \argmin_{(\mu_1,\ldots,\mu_K, \Sigma)\in\mathcal{C}} \biggl\{\frac{1}{n}\sum_{i=1}^n\sum_{k=1}^K L_{i,k} (z_i - \mu_k)^\top \Sigma^{-1}(z_i - \mu_k) + \log\det\Sigma\biggr\}.
\end{align}
}
Compute $(L_{i,k})_{i\in[n],k\in[K]}$ using the final values of $(\hat\mu_1,\ldots,\hat\mu_K,\hat\Sigma_{\mathrm{w}})$ as in~\eqref{Eq:EStep}. 

Compute $\hat{\mu}_{\mathrm{tot}} := n^{-1}\sum_{i=1}^n\sum_{k=1}^K L_{i,k}\hat\mu_{k}$ and the between-class covariance matrix 
$$
\hat\Sigma_{\mathrm{b}} := \frac{1}{n}\sum_{i=1}^n \sum_{k=1}^K L_{i,k}(\hat{\mu}_k - \hat{\mu}_{\mathrm{tot}})(\hat{\mu}_k - \hat{\mu}_{\mathrm{tot}})^\top.
$$
Set $\hat Q^{[m]} := \hat\Sigma_{\mathrm{w}}^{-1}\hat\Sigma_{\mathrm{b}}$.
}
Set 
\begin{equation}
\label{Eq:BestInitialization}
\hat m\in \argmin_{m\in[M]} \mathrm{median}\bigl(\|\hat Q^{[m]} - \hat Q^{[m']}\|_{\mathrm{op}}:m'\in [M] \setminus \{m\}\bigr).
\end{equation}

\KwOut{$\hat Q := \hat Q^{[\hat m]}$.}
\caption{Base learning using partially labeled data via an EM algorithm}
\label{Algo:EM}
\end{algorithm}

\section{Theoretical guarantees}
\label{Sec:Theory}
\subsection{Results for the high-level algorithm}

In this subsection, we consider independent triples $(X_1,Y_1,Y_1^*),\ldots,(X_n,Y_n,Y_n^*)$ taking values in $\mathbb{R}^p \times ([K] \cup \{0\}) \times [K]$.  We recall that $Y_i^*$ denotes the true label of the $i$th observation, and that $Y_i := Y_i^*$ if the $i$th label is observed, and $Y_i := 0$ otherwise.  For $k \in [K]$, let $\pi_k:=\mathbb{P}(Y_1^*=k)$ and $ \nu_k^* := \mathbb{E}(X_1\mid Y_1^*=k)$ denote the prior probability and the cluster mean of the $k$th cluster respectively, let $\nu^* := \sum_{k=1}^K \pi_k\nu_k^*$ denote the weighted cluster mean and let $\Sigma_{\mathrm{w}} := \mathrm{Cov}(X_1 \mid Y_1^*=k)$ denote the common within-cluster covariance matrix. 
With the between-cluster covariance matrix $\Sigma_{\mathrm{b}}$ from~\eqref{sb},
our goal is to estimate the set of signal coordinates,
\[
S_0 := \bigl\{j\in[p]: (\Sigma_{\mathrm{w}}^{-1}\Sigma_{\mathrm{b}})_{j,j}\neq 0\bigr\},
\] 
and we write $s_0:=|S_0|$.

Our first main theoretical result shows that if the base algorithm is accurate on each low-dimensional projection and $A$ is large, then with high probability, all signal coordinates are selected.  
\begin{thm}
\label{Thm:Meta}
Define $\gamma_{\min} := \min_{j\in S_0} (\Sigma_{\mathrm{w}}^{-1}\Sigma_{\mathrm{b}})_{j,j}$ and $\gamma_{\max} := \max_{j\in S_0} (\Sigma_{\mathrm{w}}^{-1}\Sigma_{\mathrm{b}})_{j,j}$. Let $\hat{S}$ be the output of Algorithm~\ref{Algo:Sharp-SSL} with input $K$, $p$, $(X_1,Y_1),\ldots,(X_n,Y_n)$, $A$, $B$, $d\geq s_0$, $\ell\geq s_0$ and permutation equivariant base procedure $\psi$. Write 
\begin{equation}
\label{Eq:AssumptionUniformConsistency}
\epsilon:=\mathbb{P}\biggl(\max_{P\in\mathcal{P}_d}\bigl\|\psi\bigl((PX_i, Y_i)_{i\in[n]}\bigr) -
P\Sigma_{\mathrm{w}}^{-1}\Sigma_{\mathrm{b}}P^\top\bigr\|_{\mathrm{op}}\geq \frac{\gamma_{\min}}{4(K-1)} \biggr).
\end{equation}
Then
\[
\mathbb{P}(S_0\subseteq \hat{S}) \geq  1 - \epsilon - pe^{-A\gamma_{\min}^2/(50p^2\gamma_{\max}^2)}.
\]
\end{thm}
In fact, we can see from the proof of Theorem~\ref{Thm:Meta} that the following stronger conclusion holds: for any realization $(x_i,y_i)_{i\in[n]}$ of the data satisfying 
\begin{equation}
\label{Eq:Omega}
\max_{P\in\mathcal{P}_d}\bigl\|\psi\bigl((Px_i, y_i)_{i\in[n]}\bigr) - P\Sigma_{\mathrm{w}}^{-1}\Sigma_{\mathrm{b}}P^\top\bigr\|_{\mathrm{op}} < \frac{\gamma_{\min}}{4(K-1)},
\end{equation}
we have $\mathbb{P}\bigl(S_0\subseteq \hat{S} \mid (X_i,Y_i)_{i\in[n]} = (x_i,y_i)_{i\in[n]}\bigr) \geq 1 - pe^{-A\gamma_{\min}^2/(50p^2\gamma_{\max}^2)}$.  Note here that, after conditioning on the data, the probability is taken over the randomness in the projections.  An attraction of Theorem~\ref{Thm:Meta} is its generality, and in particular the fact that we do not impose strong distributional assumptions --- we simply require control of $\varepsilon$ in~\eqref{Eq:AssumptionUniformConsistency}.  The price we pay for this generality is that the probability bound may be loose in particular cases; for example, the bound holds even with $B=1$, though in practice we would expect it to improve as $B$ increases.


\subsection{Theory for base learning using labeled data}

In this subsection, we demonstrate how the high-level result in Theorem~\ref{Thm:Meta} can be used to derive performance guarantees for a high-dimensional classification algorithm that uses the \texttt{Sharp-SSL} procedure in Algorithm~\ref{Algo:Sharp-SSL} in conjunction with the low-dimensional base method described in Algorithm~\ref{Algo:LDA} for estimating the projected whitened between-class covariance matrix. The following theorem provides uniform control of the output of Algorithm~\ref{Algo:LDA} for all axis-aligned $d$-dimensional projected datasets.  
 \begin{thm}
\label{Thm:LDA}
Fix $\epsilon \in (0,1]$, $K \in \{2,3,\ldots,\}$, $p,d \in \mathbb{N}$ with $p \geq d$ and $n \in \mathbb{N}$ with $n \geq Kd+1$. Suppose that $(X_1,Y_1),\ldots,(X_n,Y_n)$ are independent and identically distributed pairs, with $\mathbb{P}(Y_1=k)=\pi_k$ and $X_1\mid Y_1 = k\sim \mathcal{N}_p(\nu_k^*,\Sigma_{\mathrm{w}})$ for $k\in[K]$,
and let $\psi\bigl((PX_i, Y_i)_{i\in[n]}\bigr)$ be the output of Algorithm~\ref{Algo:LDA} with input $(X_i,Y_i)_{i\in[n]}$, 
for $P\in\mathcal{P}_d$. 
Suppose that $\|\nu_k^* - \nu_\ell^*\|\leq R_1$ for all $k,\ell \in[K]$ and some $R_1 > 0$, and that $\Sigma_{\mathrm{w}}$ is diagonal and well-conditioned in the sense that $\max\{\|\Sigma_{\mathrm{w}}\|_{\mathrm{op}},\|\Sigma_{\mathrm{w}}^{-1}\|_{\mathrm{op}}\} \leq R_2$ for some $R_2 \geq 1$.  If 
\begin{equation}
    \label{Eq:SampleSize}
\frac{16R_2^2K}{n} \leq 1 \quad \text{and} \quad \frac{32R_2^2}{n^{1/2}}\log^{1/2}\biggl(\frac{8 \cdot 9^d \binom{p}{d}}{\epsilon}\biggr) \leq 1,  
\end{equation}
then with probability at least $1-\epsilon$, we have
\begin{align*}
\max_{P\in\mathcal{P}_d} \bigl\|\psi\bigl((PX_i,Y_i)_{i\in[n]}\bigr) - (P\Sigma_{\mathrm{w}}P^\top)^{-1}P\Sigma_{\mathrm{b}}P^\top & \bigr\|_{\mathrm{op}} \\
& \lesssim_{R_1,R_2} \frac{K}{n} +  \sqrt\frac{d\log(ep/d) + \log(1/\epsilon)}{n}.
\end{align*}
\end{thm}
The sample size condition~\eqref{Eq:SampleSize} can be restated as $n\gtrsim_{R_1,R_2} d\log p + \log(1/\epsilon) + K$, so may be regarded as mild.  Regarding $K$ as a constant, Theorem~\ref{Thm:LDA} confirms that the uniform control of Algorithm~\ref{Algo:LDA} is at the parametric rate, up to a logarithmic factor.  The following corollary then follows immediately by combining Theorems~\ref{Thm:Meta} and~\ref{Thm:LDA}.
\begin{cor}
\label{Cor:LDA}
Fix $\epsilon \in (0,1]$.  Suppose that the conditions of Theorem~\ref{Thm:LDA} hold, 
and moreover that $\lambda_{\min}(\Sigma_{\mathrm{b}})\geq 1/R_3$ for some $R_3>0$. 
Then there exist $C_1, C_2 >0$, depending only on $R_1, R_2$ and $R_3$, such that if 
\[
C_1 \biggl(\frac{K}{n} + \sqrt\frac{d\log(ep/d) + \log(1/\epsilon)}{n}\biggr)  \leq \frac{1}{K},
\]
then the output $\hat S$ of Algorithm~\ref{Algo:Sharp-SSL} with input $K$, $p$, $d\geq s_0$, $\ell\geq s_0$, $(X_1,Y_1),\ldots,(X_n,Y_n)$, $A$, $B$, and base procedure $\psi$ from Algorithm~\ref{Algo:LDA} satisfies
\[
\mathbb{P}(S_0\subseteq \hat S)\geq 1 - \epsilon - p\exp\biggl(-\frac{A}{C_2p^2}\biggr).
\]
\end{cor}
Thus, under the conditions of Corollary~\ref{Cor:LDA}, the \texttt{Sharp-SSL} algorithm can, with high probability, 
select the signal variables in the top $s_0$  output variables, provided that the number $A$ of groups of random projections is large by comparison with $p^2$.  
In other words, the algorithm reduces the problem to a low-dimensional one, for which standard learning techniques can be applied.  The guarantees for these methods 
\citep[e.g.][Theorem~6.6.1]{anderson2003introduction}
can then be combined on the high-probability event of Corollary~\ref{Cor:LDA} to establish theoretical results for the full procedure. 

\subsection{Theory for semi-supervised based learning}
\label{emguar}

When the proportion of labeled data is low,  Algorithm~\ref{Algo:LDA} may be inaccurate when used 
as the base procedure in Algorithm~\ref{Algo:Sharp-SSL}. 
The aim of this subsection, therefore, is to study the base procedure of Algorithm~\ref{Algo:EM}, which is able to leverage both the labeled and unlabeled data via an EM algorithm to estimate the whitened between-class covariance matrix for each projected data set.  
Our analysis builds on several recent breakthroughs in our understanding of the EM algorithm.  This line of work includes \citet{balakrishnan2017statistical}, \citet{daskalakis2017ten}, \citet{yan2017convergence}, \citet{dwivedi2020sharp}, \citet{dwivedi2020singularity}, \citet{davis2021clustering}, \citet{ho2020instability}, \citet{ndaoud2022sharp}, \citet{WuZhou2019} and \citet{doss2020optimal}, all of which focus on the unsupervised case.

For simplicity, we will focus on the setting where independent and identically distributed $(X_1,Y_1^*),\ldots,(X_n,Y_n^*)$ are generated from a mixture of two Gaussians with opposite means and identity covariance matrix:
\begin{equation}
Y_i^* \sim \mathrm{Unif}(\{1,2\}),
\, X_i\mid Y_i^* \sim \mathcal{N}_p\bigl((-1)^{Y_i^*} \nu^*, I_p\bigr), \text{ and } Y_i = Y_i^*\mathbbm{1}_{\{i\leq n_{\mathrm{L}}\}} \text{ for all $i\in[n]$}.\label{Eq:MixtureModel}
\end{equation}
\sloppy We assume that we observe $(X_1,Y_1),\ldots,(X_{n_{\mathrm{L}}},Y_{n_{\mathrm{L}}}), X_{n_{\mathrm{L}}+1},\ldots,X_n$ for some $n_{\mathrm{L}}\in \{0,\ldots,n\}$. In other words, we are given $n_{\mathrm{L}}$ labeled observations and  $n_{\mathrm{U}}:=n-n_{\mathrm{L}}$ unlabeled ones.  Thus, $n_{\mathrm{L}}=0$ corresponds to the fully unsupervised case, i.e.,~clustering, while $n_{\mathrm{L}} = n$ corresponds to the supervised case, i.e.,~classification.
We define $Y_i = Y_i^*$ for $i \in [n_{\mathrm{L}}]$, and  $Y_i = 0$ for $i \in \{n_{\mathrm{L}}+1, \ldots, n\}$. 

We first study the performance of the EM procedure after the covariates have been projected into a lower-dimensional space.  In other words, for some fixed $P \in \mathcal{P}_{d}$, define $Z_i:=PX_i$ for $i\in [n]$ and $\mu^*:=P\nu^* \in \mathbb{R}^d$, so that $Z_i\mid Y_i^*\sim \mathcal{N}_d((-1)^{Y_i^*}\mu^*, I_d)$. In this setting, we have a single unknown parameter $\mu^*$ to estimate, and this can be achieved by applying Algorithm~\ref{Algo:EM} to $(Z_i, Y_i)_{i\in[n]}$ with $K=2$ and the constraint set
\begin{equation}
\label{Eq:ConstraintSet}
\mathcal{C}:=\bigl\{(-\mu,\mu,I_d):\mu\in\mathbb{R}^d\bigr\}.
\end{equation}
After initializing the EM algorithm at some fixed $(-\hat{\mu}^{(0)},\hat{\mu}^{(0)},I_d) \in \mathcal{C}$, for $t \in \mathbb{N}$, the $t$th iterate of the EM iteration described in~\eqref{Eq:EStep} and~\eqref{Eq:MStep} is $(-\hat{\mu}^{(t)}, \hat{\mu}^{(t)}, I_d)$, where 
\begin{equation}
\label{Eq:EMUpdate}
\hat{\mu}^{(t)} := \frac{1}{n}\biggl\{\sum_{i:Y_i\neq 0} (-1)^{Y_i}Z_i  + \sum_{i:Y_i=0} Z_i\tanh\big\langle Z_i,\hat{\mu}^{(t-1)}\big\rangle\biggr\};
\end{equation}
see Lemma~\ref{Lemma:EMderiv}.  
Since we allow $n_{\mathrm{L}}=0$, where $\mu$ is only identifiable up to sign, and since the between-class sample covariance matrix $\hat\Sigma_{\mathrm{b}}$ computed in Algorithm~\ref{Algo:EM} is  equal to $\hat\Sigma_{\mathrm{b}} = \hat\mu_1\hat\mu_1^\top - \hat\mu_{\mathrm{tot}}\hat\mu_{\mathrm{tot}}^\top$, which is invariant to flipping the signs of $\hat\mu_1$ and $\hat\mu_2$ simultaneously, it is natural to consider the loss function $L:\RR^d\times \RR^d \to [0,\infty)$ given by
\[
  L(\mu,\mu') := \|\mu-\mu'\|\wedge \|\mu+\mu'\|.
\]

Proposition~\ref{Prop:OneGoodInitialization} below provides a theoretical guarantee for this semi-supervised EM algorithm.  For notational simplicity, 
we define $\gamma:= n_{\mathrm{L}}/n$,  
$\omega_0:=\sqrt{\{d \log n + \log(1/\delta)\}/{n_{\mathrm{U}}}}$ and $\zeta_0 :=\min\{\omega_0\gamma^{-1/2}, \omega_0^{1/2}\}$ throughout this section.  Thus, treating $d$ as a constant and ignoring polylogarithmic terms, $\omega_0$ is of order $n_{\mathrm{U}}^{-1/2}$ and $\zeta_0$ is of order $\min\{n_{\mathrm{L}}^{-1/2}, n_{\mathrm{U}}^{-1/4}\}$ when $\gamma < 1/2$.  We remark that $n_{\mathrm{L}}^{-1/2}$ is the critical $\ell_2$-testing radius for distinguishing the means of two labeled Gaussian distributions with identity covariance using $n_{\mathrm{L}}$ observations.  On the other hand, as we show in Lemma~\ref{Lemma:TestingRadius}, no test of the null hypothesis $H_0: \mathcal{N}_d(0, I_d)$ against the two-component mixture alternative $H_1: \frac{1}{2}\mathcal{N}_d(\mu^*,I_d) + \frac{1}{2}\mathcal{N}_d(-\mu^*,I_d)$ based on $n_{\mathrm{U}}$ observations can have large power unless the signal strength $\|\mu^*\|$ is at least of order $n_{\mathrm{U}}^{-1/4}$.
\begin{prop}
\label{Prop:OneGoodInitialization}
Fix $\delta\in(2e^{-n}, 1]$ and $r\geq 1$, and suppose that $\|\mu^*\|\leq r$ and $\gamma < 1/2$.  There exists $c>0$, depending only on $r$, such that if $\omega_0 \leq c$ and $n \geq 3$, then the following statements hold:
\begin{enumerate}[(i)]
    \item For any $\hat{\mu}^{(0)} \in \mathbb{R}^d$ with $\|\hat{\mu}^{(0)}\| \leq r+3$, we have with probability at least $1-2\delta$ that
\[
\limsup_{t\to\infty} L(\hat{\mu}^{(t)},\mu^*) \lesssim_r  \zeta_0 \vee \|\mu^*\|.
\]
\item There exists $C>0$, depending only on $r$, such that if $\|\mu^*\|\geq C \zeta_0 \sqrt{d\log n}$ and $\hat\mu^{(0)} = (\zeta_0\vee r\omega_0) \eta_0$ with $\eta_0\sim \mathrm{Unif}(\mathbb{S}^{d-1})$, then with probability at least $1-2\delta-\sqrt
{2/(\pi \log n_{\mathrm{U}})}$, we have
\[
\limsup_{t\to\infty} L(\hat\mu^{(t)}, \mu^*) \lesssim_r \frac{\omega_0}{\|\mu^*\|}\wedge \frac{\omega_0}{\gamma^{1/2}}.
\]
\end{enumerate}
\end{prop}
In order to interpret Proposition~\ref{Prop:OneGoodInitialization}\emph{(i)}, consider the regime where $\|\mu^*\| \leq \zeta_0$.  In this case, as discussed above, the two mixture components are essentially indistinguishable, and the bound reveals that the EM algorithm performs no worse than the trivial zero estimator, up to constant factors.  
On the other hand, part~\emph{(ii)} studies the more interesting regime where the two mixture components are distinguishable, and we establish a faster convergence rate for the EM algorithm in this strong signal regime.  

The following theorem combines the two convergence regimes in Proposition~\ref{Prop:OneGoodInitialization} to derive a convergence guarantee for the estimated whitened between-class covariance matrix output by Algorithm~\ref{Algo:EM}. To state the result, recall the definition of $\mathcal{C}$ from~\eqref{Eq:ConstraintSet}.  For any $\zeta > 0$, we write $U(\zeta)$ for the pushforward measure on $\mathcal{C}$ induced by $\mathrm{Unif}(\zeta\mathbb{S}^{d-1})$ under the map $\mu\mapsto(-\mu,\mu, I_d)$.
\begin{thm}
\label{Thm:LowDimEM}
Fix $\delta\in(2e^{-n}, 1]$, and $r\geq 1$ and suppose that $\|\mu^*\|\leq r$ and $\gamma < 1/2$.  There exists $c>0$, depending only on $r$, such that if $\omega_0 \leq \min\{c, (d\log n)^{-3}\}$ and $n\geq 108$, then the 
sequence of outputs $(\hat Q^{(T)})_{T\in \mathbb{N}}$ of Algorithm~\ref{Algo:EM} with inputs $(Z_1,Y_1),\ldots,(Z_{n},Y_{n})$, $\mathcal{C}$, $\pi_{\mathcal{C}} = U(\zeta_0\vee r\omega_0)$, $M \in \mathbb{N}$  and $T \in \mathbb{N}$ satisfies with probability at least $1-3\delta-e^{-M/50}$ that
\[
\limsup_{T\to\infty}\|\hat Q^{(T)} - \mu^*\mu^{*\top}\|_{\mathrm{op}}\lesssim_r \frac{\omega_0}{\|\mu^*\|} \wedge \zeta_0.
\]
\end{thm}
Finally in this section, we study the implications of Theorem~\ref{Thm:LowDimEM} for the recovery of the signal coordinates in the semi-supervised learning setting.  We write $\psi^{(M, T)}$ for the base procedure that takes $(z_i,y_i)_{i \in [n]} \in\mathbb{R}^{d}\times ([K]\cup\{0\})$ as input and returns the output of Algorithm~\ref{Algo:EM} when run with these inputs together with $\mathcal{C}$, $\pi_{\mathcal{C}}$, $M$ and $T$.  Let $S_0$ denote the set of coordinates where $\nu^* \in \mathbb{R}^p$ is non-zero, and let $s_0 := |S_0|$.
\begin{cor}
\label{Cor:MultipleInitialisaiton}
Fix $\epsilon \in (8e^{-n/2},1]$, $r \geq 1$, and suppose that $\|\mu^*\|\leq r$, $M \geq 50\log(4/\epsilon) + 50d\log p$ and $\gamma < 1/2$. Let  $\nu^*_{\max}:=\|\nu^*\|_\infty$ and let $\nu_{\min}^*$ denote the minimum absolute value of a non-zero component of $\nu^*$. There exist $C_1,C_2>0$, depending only on $r$, such that if $n \geq C_1(d \log p)^6\{d\log p  + \log(1/\epsilon)\}$, and 
\[
C_2\min\biggl[\biggl\{\frac{d\log (p \vee n) + \log(1/\epsilon)}{n}\biggr\}^{1/4},  \sqrt\frac{d\log (p \vee n) + \log(1/\epsilon)}{n_{\mathrm{L}}}\biggr] \leq \frac{(\nu^*_{\min})^2}{4},
\]
then the 
sequence of outputs $(\hat S^{(T)})_{T\ge 1}$ of
Algorithm~\ref{Algo:Sharp-SSL} with inputs $K=2$, $p$, $d \geq s_0$, $\ell \geq s_0$, $(X_i,Y_i)_{i\in[n]}$, $A$, $B$ and base procedure $\psi^{(M,T)}$ satisfies 
\[
\liminf_{T\to\infty}\mathbb{P}(S_0\subseteq \hat{S}^{(T)}) \geq 1-\epsilon - pe^{-A(\nu^*_{\min})^4/(50p^2(\nu^*_{\max})^4)}.
\]
\end{cor}
Corollary~\ref{Cor:MultipleInitialisaiton} reveals in particular that, treating $\nu^*_{\max}$ and $\nu^*_{\min}$ as constants and under the stated sample size conditions, 
we again recover all of the signal coordinates in the top $s_0$ output entries, provided that $A$ is large by comparison with $p^2$.  Thus, in this sense, we can achieve a similar guarantee to that provided by Corollary~\ref{Cor:LDA}, though the number of groups of projections required for a high probability guarantee in Corollary~\ref{Cor:MultipleInitialisaiton} may be significantly larger in settings where the ratio $\nu_{\max}^*/\nu_{\min}^*$ is large.  

\section{Numerical studies}
\label{Sec:Numerics}

Throughout this section, unless otherwise stated, data $(X_i, Y_i, Y_i^*)_{i\in[n]}$ are sampled from an equal-probability normal mixture as follows: $\mathbb{P}(Y_i^* = k) = 1/K$ for $k\in[K]$, $\mathbb{P}(Y_i = Y_i^*) = 1-\mathbb{P}(Y_i = 0) = \gamma$ and $X_i\mid Y_i^* \sim N_p(\mu_{Y_i^*}, \Sigma_{\mathrm{w}})$. The cluster means $(\mu_k)_{k\in[K]}$ are chosen to be $s_0$-sparse and we define the signal-to-noise ratio of the problem to be\footnote{In some of our simulations, $\Sigma_{\mathrm{w}}$ was generated randomly for convenience.  In such settings, we replaced $\mathrm{tr}(\Sigma_{\mathrm{w}})/p$ in the denominator of~\eqref{Eq:SNR} with $\mathbb{E}\{\mathrm{tr}(\Sigma_{\mathrm{w}})\}/p$.}
\begin{equation}
\mathrm{SNR}:=\frac{\min_{k,k'\in[K],k\neq k'}\|\mu_k-\mu_{k'}\|}{\sqrt{\mathrm{tr}(\Sigma_{\mathrm{w}})/p}}.
\label{Eq:SNR}    
\end{equation}
In our numerical studies, we slightly modify Algorithm~\ref{Algo:EM} so that instead of randomly initializing the cluster means and the covariance matrix, we use the output of hierarchical clustering to initialize the EM algorithm as implemented in the \texttt{mclust} \texttt{R} package \citep{fraley1998mclust}. This allow us to run Algorithm~\ref{Algo:EM} with $M=1$.

\subsection{Choice of tuning parameters}
\label{SubSec:TuningParameters}

The purpose of this subsection is to investigate the effect of the various input parameters $A$, $B$, $d$ and $\ell$ in Algorithm~\ref{Algo:Sharp-SSL}, and to recommend sensible default choices.  In Figure~\ref{Fig:Tuning}, we plot the misclustering rate with Algorithm~\ref{Algo:EM} as a base procedure in our Gaussian semi-supervised learning setting as each of these parameters varies, for four different $\mathrm{SNR}$ levels.  After applying Algorithm~\ref{Algo:Sharp-SSL}, we obtain our final estimated cluster labels by using Algorithm~\ref{Algo:EM} again on the data projected onto the selected coordinates with a single hierarchical clustering initialization.  
We then output the predicted labels, computed as $\hat{y}_i := \sargmax_{k \in [K]} L_{i,k}$ for $i\in[n]$, instead of $\hat {Q}$.  

The panels in Figure~\ref{Fig:Tuning} reveal that the misclustering rate is quite robust to the choices of $A$, $B$ and $d$, and that it is less serious (and may even help) to choose $\ell$ larger---rather than smaller---than $s_0$.  In particular, it seems that $A = 150$ suffices for almost optimal performance (though there appears to be some penalty for choosing it to be as small as $50$), and $B=75$ appears adequate.  There is no clear trend on performance with the choice of $d$, so for simplicity we took $d=s_0$ in our remaining simulations below.  Finally, the misclustering rate appears to decrease as $\ell$ increases, with an elbow in the curve visible at the highest value of the $\mathrm{SNR}$ when $\ell$ is set to the true sparsity level $s_0$.  Of course, if $\ell$ is chosen to be very large, then we will include many noise variables, and the misclustering rate will eventually deteriorate.  Nevertheless, the bottom-right panel of Figure~\ref{Fig:Tuning} indicates that the gain in increasing the probability of including all signal variables may outweigh the penalty of also including more noise variables---as expected, this effect is larger when the $\mathrm{SNR}$ is larger.  For simplicity we choose $\ell=s_0$ in our remaining simulations, though we recommend practitioners err on the side of choosing larger $\ell$.    

\begin{figure}
    \centering
    \begin{tabular}{cc}
    \includegraphics[width=0.45\textwidth]{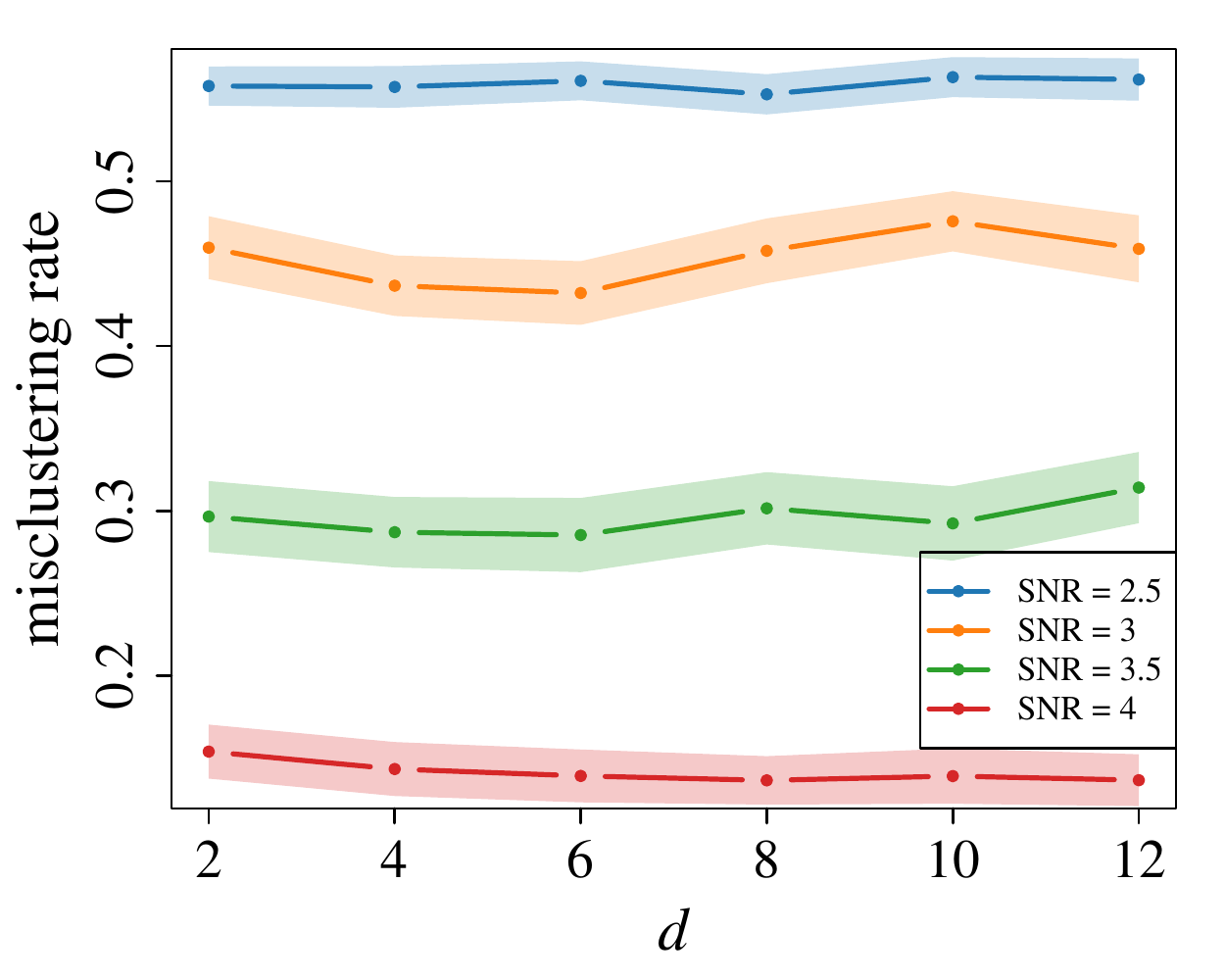} &
    \includegraphics[width=0.45\textwidth]{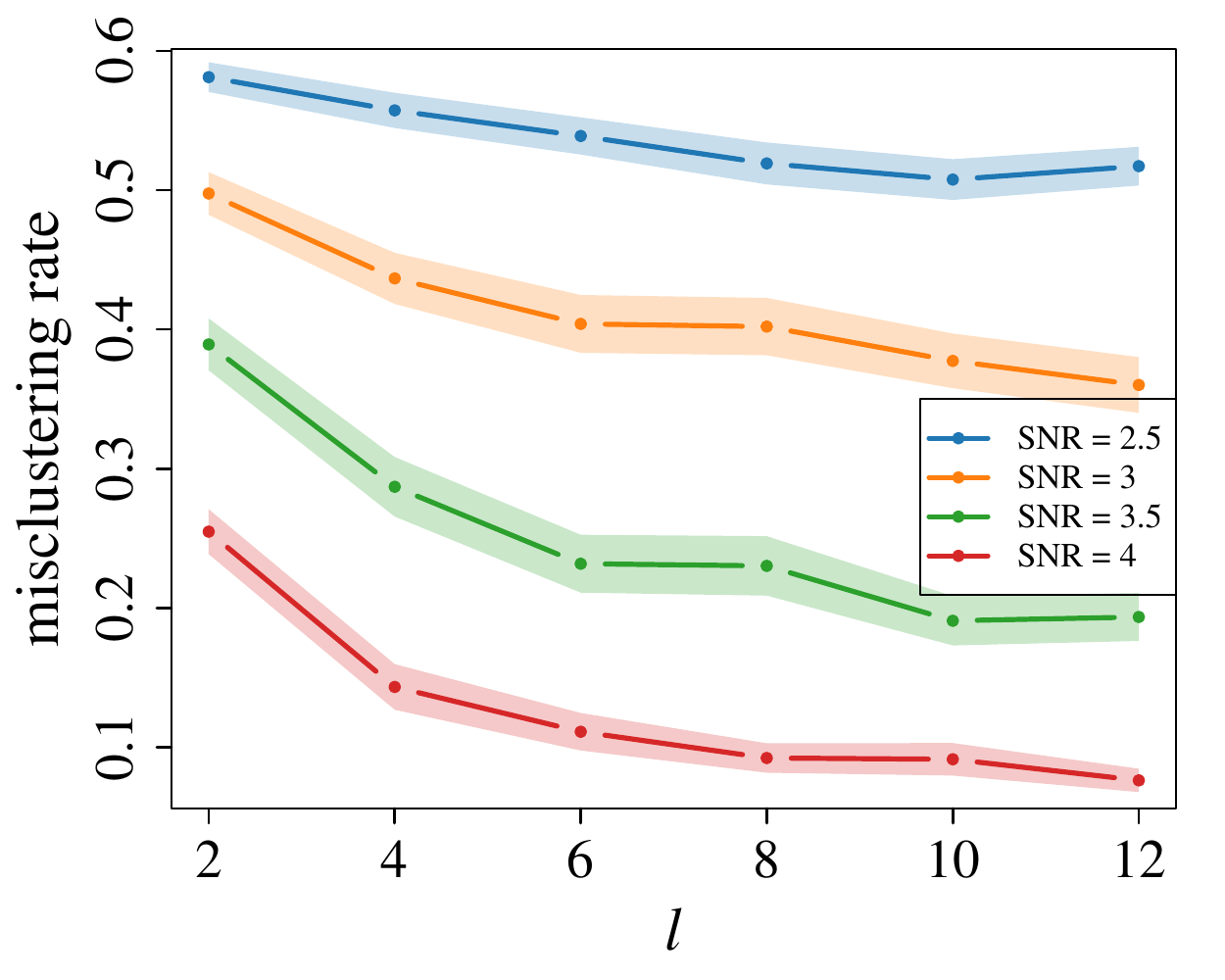} \\
    \includegraphics[width=0.45\textwidth]{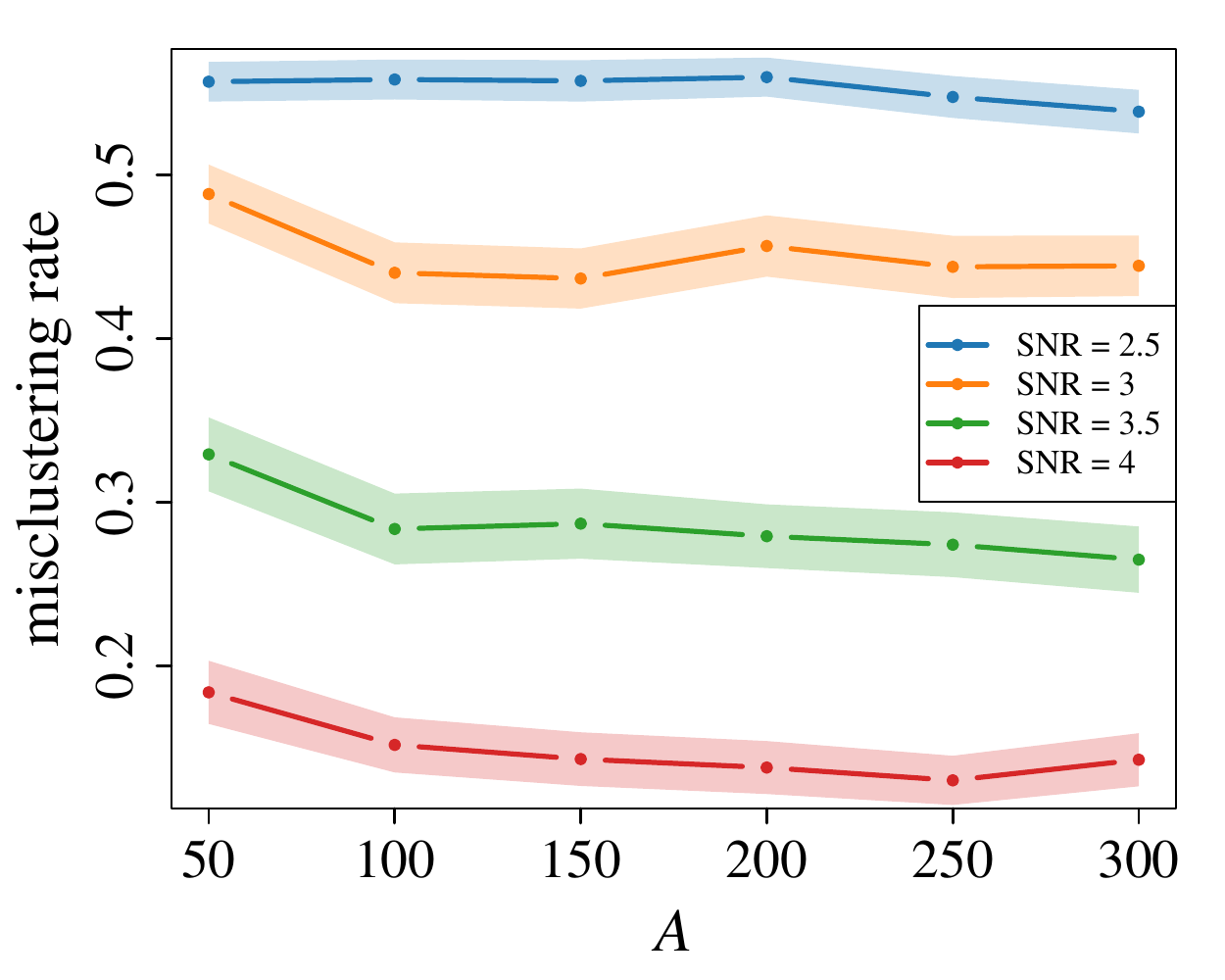} &
    \includegraphics[width=0.45\textwidth]{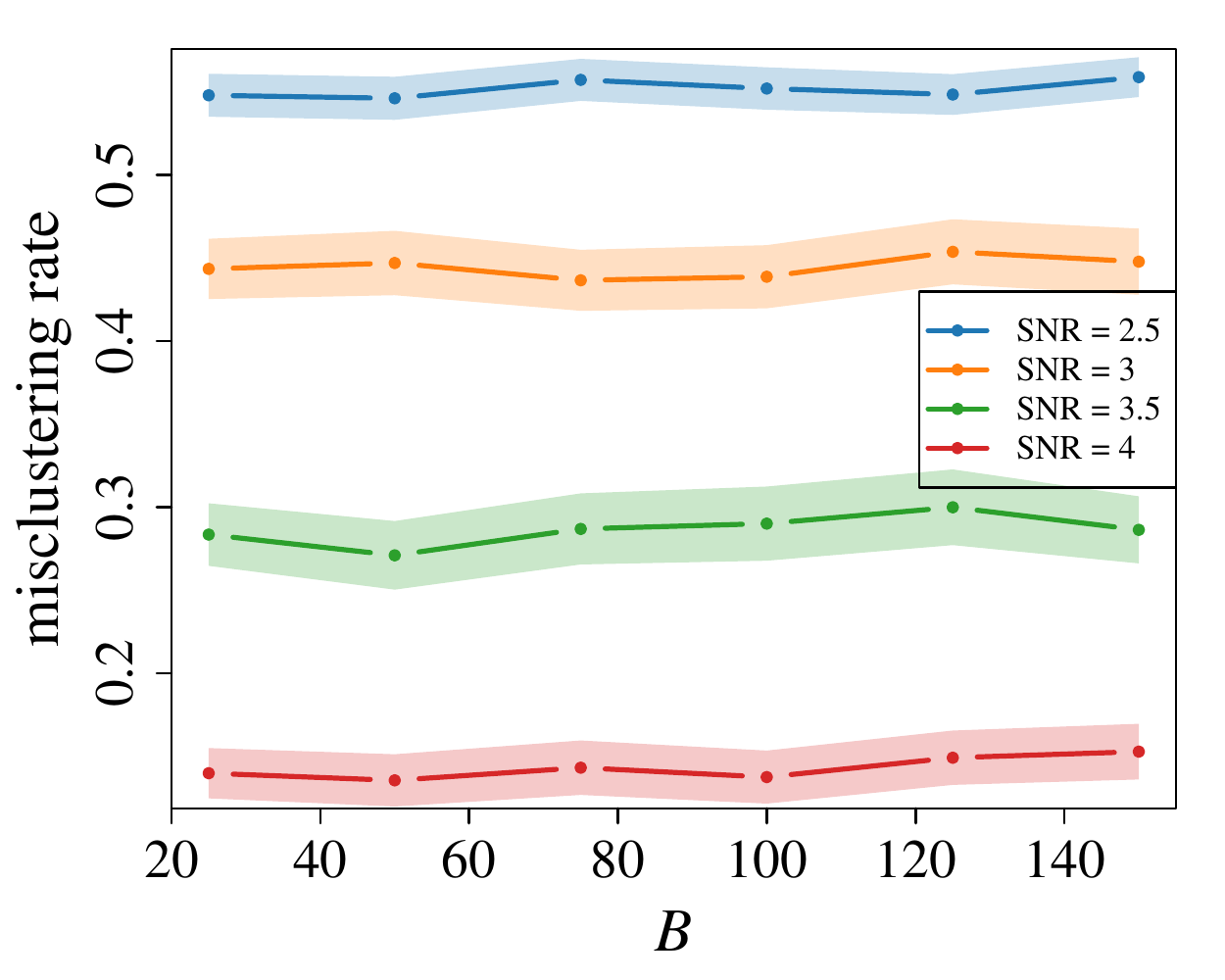} 
    \end{tabular}
    \caption{\label{Fig:Tuning}Average misclustering rate over 200 repetitions in our anisotropic Gaussian semi-supervised learning problem with $n=250$, $p=600$, $s=4$, $K=3$, $\gamma=0.05$, $\mathrm{SNR} \in \{2.5,3,3.5,4\}$ and $\Sigma_{\mathrm{w}} = V\Lambda V^\top$, where $\Lambda \in \mathbb{R}^{p \times p}$ is diagonal with independent $\mathrm{Unif}[0,2]$ diagonal entries and $V$ is independent of $\Lambda$, and generated according to the Haar measure on $\mathbb{O}^{p \times p}$.  For each of the four panels, we fix three of $d=4$, $\ell=4$, $A=150$, $B=75$, and vary the remaining one. The shaded regions represent interpolated 95\% confidence intervals at each of the~points.}
\end{figure}

\subsection{Comparison with existing methods}
\label{SubSec:Comparison}

In this subsection, we compare the empirical performance of the \texttt{Sharp-SSL} algorithm in high-dimensional clustering tasks with several existing approaches.  We apply the \texttt{Sharp-SSL} algorithm using the EM algorithm of Algorithm~\ref{Algo:EM} as a base procedure, with input parameters $A=150$, $B=75$, $d=\ell=s_0$ as discussed in Section~\ref{SubSec:TuningParameters}, and our final estimated cluster labels are then obtained as described there.

\begin{figure}
 \centering
\begin{tabular}{cc}
$n=250$, $p=200$, isotropic & $n=250$, $p=600$, anisotropic\\
 \includegraphics[width=0.45\textwidth]{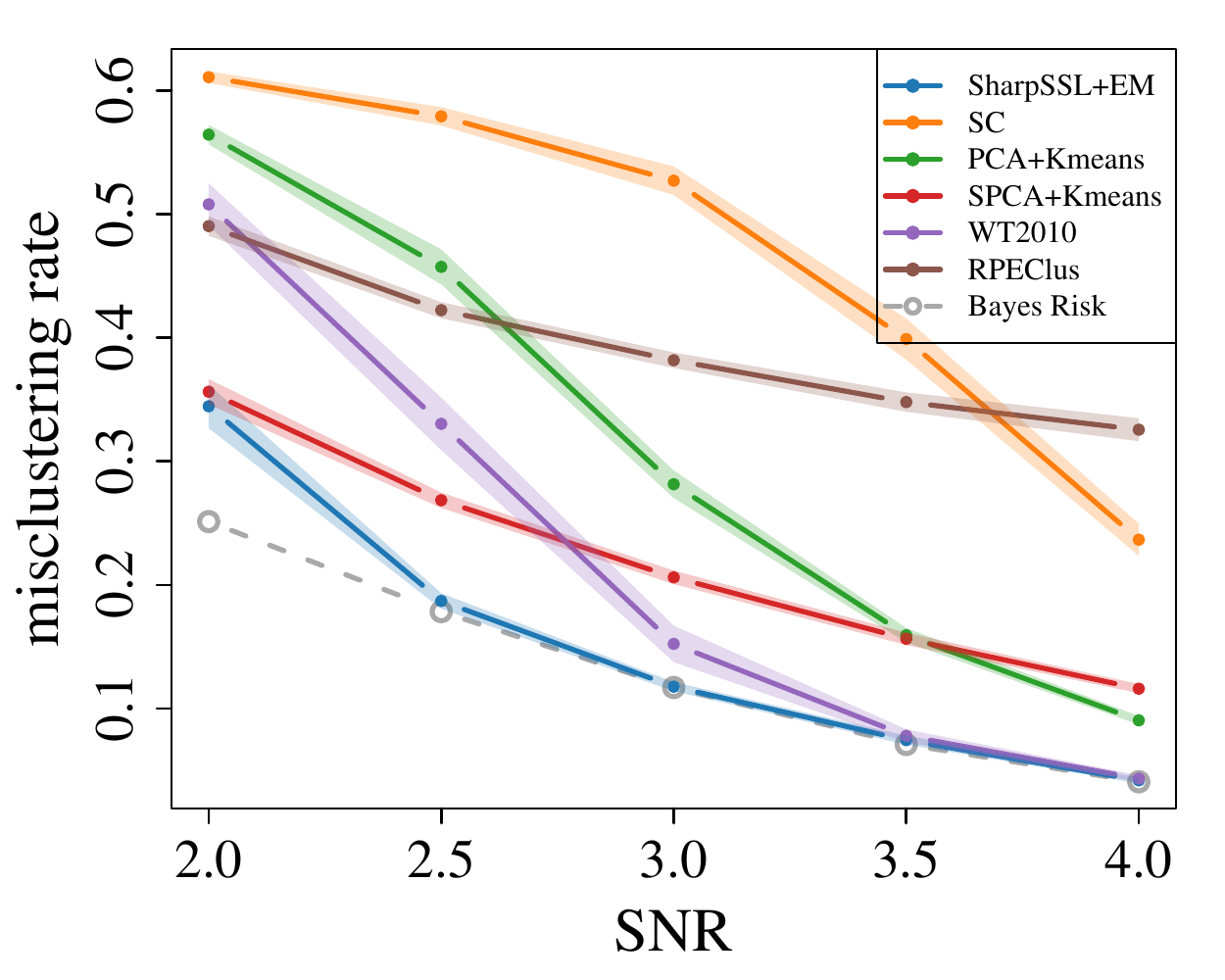} & \includegraphics[width=0.45\textwidth]{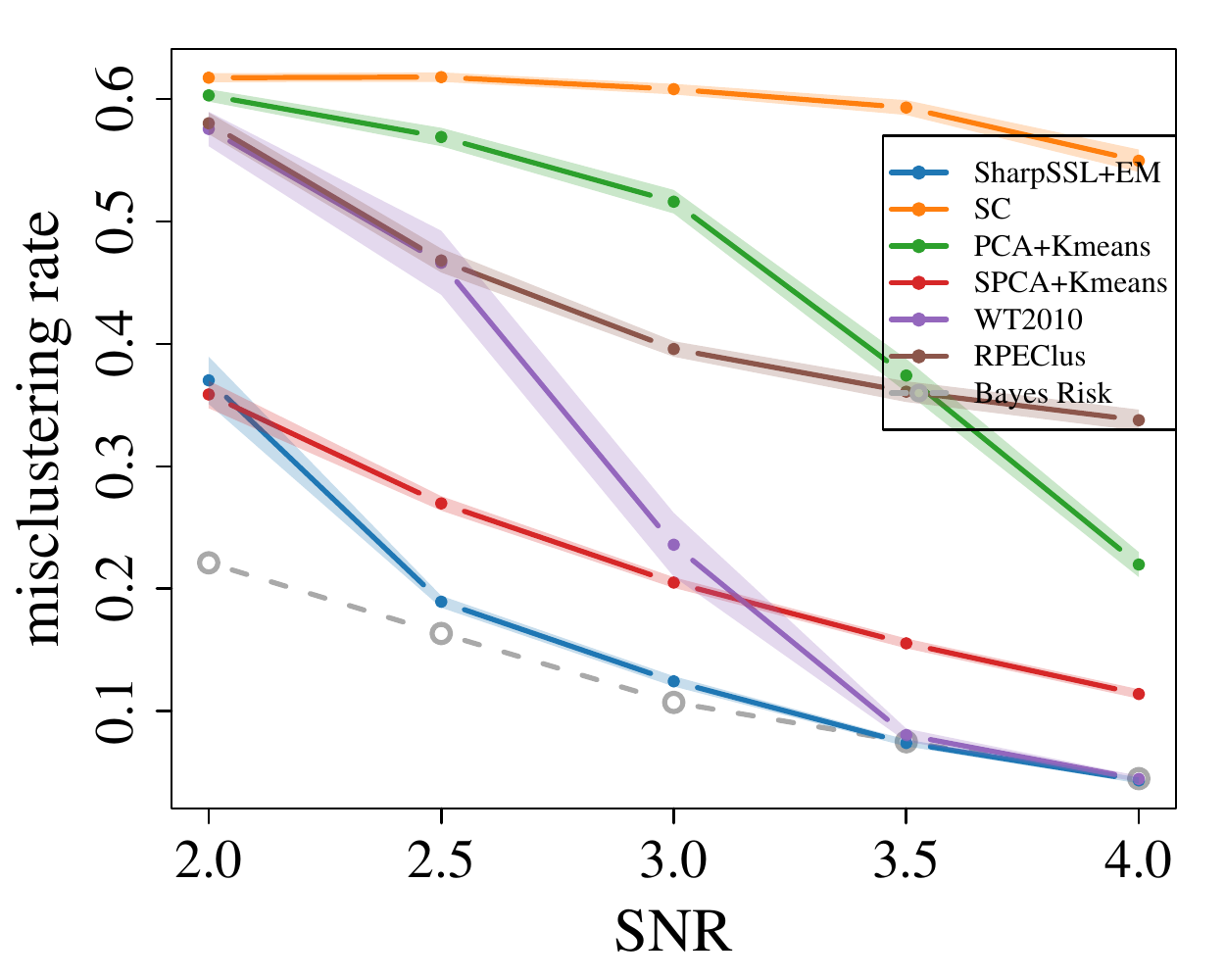} \\
$\mathrm{SNR}=3$, $p=200$, isotropic & $\mathrm{SNR}=3$, $p=600$, anisotropic\\
\includegraphics[width=0.45\textwidth]{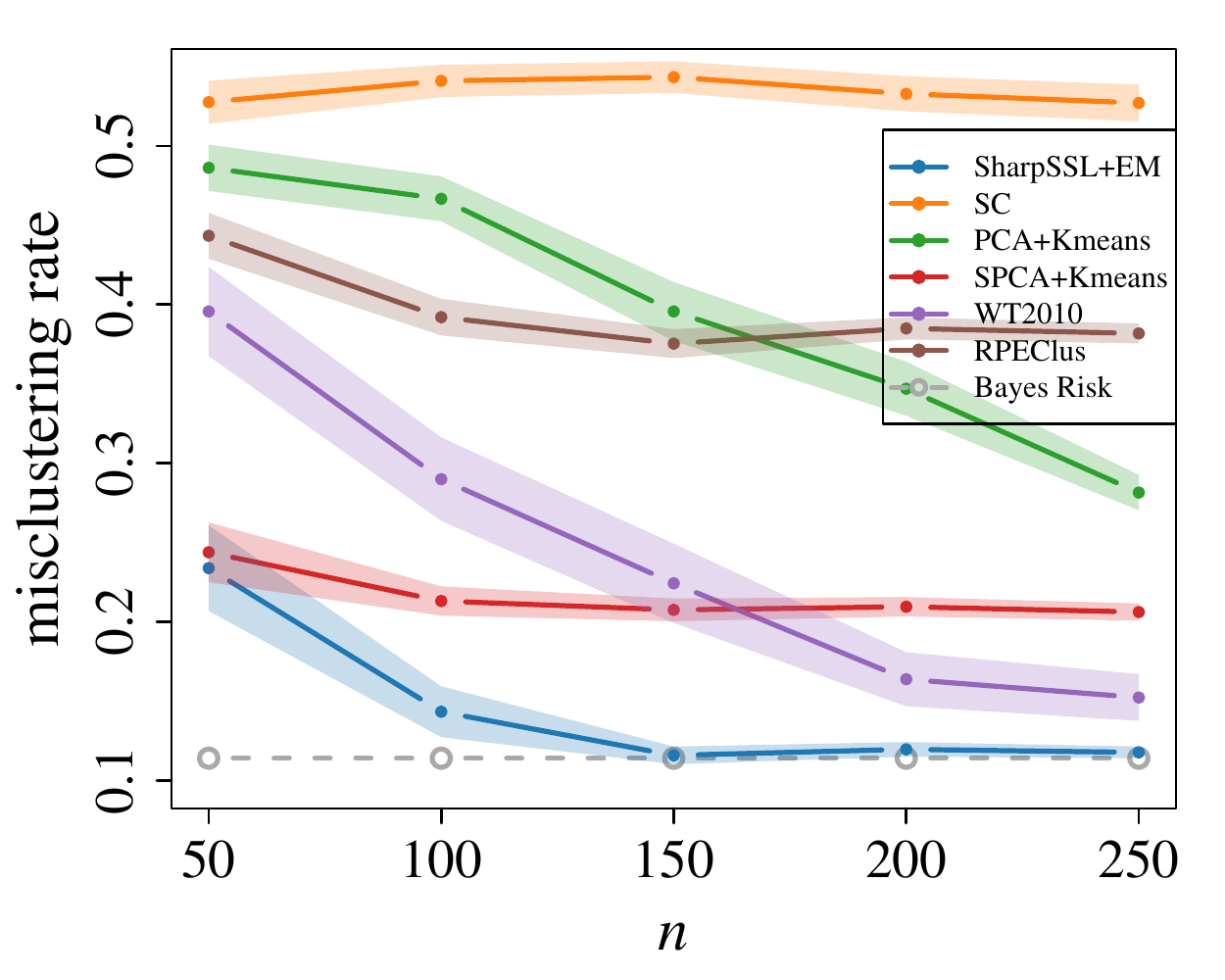} & \includegraphics[width=0.45\textwidth]{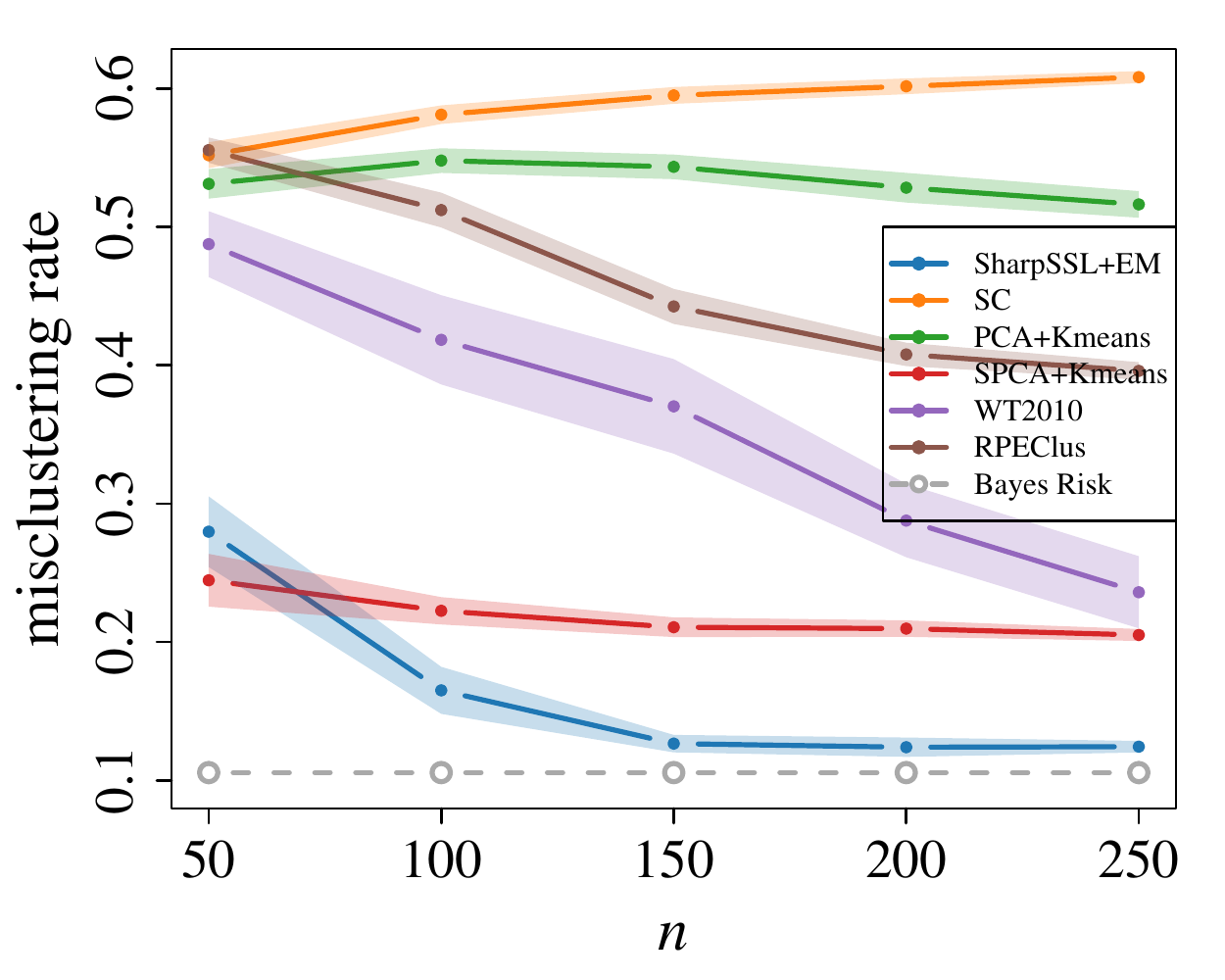}
\end{tabular}
 \caption{\label{fig:aniso_iso_results} Average misclustering rate over 100 repetitions using \texttt{Sharp-SSL} followed by the EM algorithm, as well as using the other methods from Section~\ref{SubSec:Comparison}.  Data are generated from the normal mixture distribution described at the beginning of Section~\ref{Sec:Numerics} with $K=3$ and $p=200$ (left) as well as $p=600$ (right).  The three cluster means are given by $\mu_1 = a(1,1,0,\boldsymbol{0}_{p-3})$, $\mu_2 = a(-1,0,1,\boldsymbol{0}_{p-3})$ and $\mu_3 = a(0,-1,-1,\boldsymbol{0}_{p-3})$, where the scale $a$ is chosen such that their pairwise distances are all equal to $\mathrm{SNR}$. For isotropic settings (left), $\Sigma_{\mathrm{w}} = I_p$; for anisotropic settings (right), $\Sigma_{\mathrm{w}} = V\Lambda V^\top$, where $\Lambda \in \mathbb{R}^{p \times p}$ is diagonal with independent $\mathrm{Unif}[0,2]$ diagonal entries and $V$ is independent of $\Lambda$, and sampled from the Haar measure on $\mathbb{O}^{p \times p}$.  The Bayes risk is shown as the gray dashed line.  In the top panels, $n = 250$ and the $\mathrm{SNR}$ varies; in the bottom panels, $\mathrm{SNR} = 3$ and $n$ varies.  The shaded regions represent interpolated 95\% confidence intervals at each of the~points.}
\end{figure}

We compare the \texttt{Sharp-SSL} algorithm with five alternative high-dimensional clustering methods: spectral clustering \citep[e.g.][]{von2007tutorial}, the $\ell_1$-penalized approach of \citet{WittenTibshirani2010} and the \texttt{RPEClus} algorithm of \citet{anderlucci2022high} as well as a pair of methods that, like \texttt{Sharp-SSL}, apply dimension reduction prior to a low-dimensional clustering algorithm.  

In more detail, the spectral clustering approach first constructs a $J$-nearest neighbour graph adjacency matrix $A = (A_{i,i'})_{i,i'\in[n]} \in \{0,1\}^{n \times n}$, where $A_{i,i'} := 1$ if either $X_i$ is one of the $J=10$ nearest neighbours of $X_{i'}$ in Euclidean distance or vice versa, and $A_{i,i'} := 0$ otherwise. It then computes an $n \times K$ matrix of eigenvectors associated with the $K$ smallest nonzero eigenvalues of the Laplacian matrix $L := D - A$, where $D \in \mathbb{R}^{n \times n}$ is a diagonal matrix with diagonal entries $D_{i,i} := \sum_{i'\in[n]} A_{i,i'}$. 
The final step is to apply the $K$-means clustering algorithm \citep{lloyd1982least}, as implemented in the \texttt{kmeans} base \texttt{R} function with 100 random initializations, to the rows of $L$ with the oracle choice of $K$.   

The \citet{WittenTibshirani2010} method, which is implemented in the \texttt{sparcl} \texttt{R} package, determines the estimated cluster memberships by maximizing a coordinatewise-weighted between-cluster sum of squares criterion, subject to an $\ell_1$ constraint on the weights.  A permutation approach is used to select the $\ell_1$ tuning parameter.  

In the \texttt{RPEClus} algorithm of \citet{anderlucci2022high}, we generate~$B$ random orthogonal projections and incorporate the $d$-dimensional projected data as covariates for a linear regression with the orthogonal complement of the projected data as the response.  We then use the Bayesian Information Criteria (BIC) from both an application of the EM algorithm to the projected data and the aforementioned regression to identify good projections, and aggregate using the consensus clustering technique of \citet{dimitriadou2002combination} over the best $B^*$ projections chosen according to the sum of the BIC scores.  Following the recommendation of \citet{anderlucci2022high}, we took $B=1000$ and $B^*=100$ as well as $d = s_0$.  
It turned out that this approach had a misclustering rate almost identical to that of a random guess, likely because it did not leverage the sparsity of the signal.  
We therefore modified this method by generating random axis-aligned projections instead of orthogonal ones, and report this version in our comparison.  

The first of the two-stage approaches applies principal component analysis (PCA) to project the data into the oracle choice of $K-1$ dimensions (the dimension of the space spanned by the $K$ cluster means); the second uses sparse principal component analysis (SPCA), as implemented in the  \texttt{SPCAvRP} algorithm \citep{GWS2020} with inputs $A = 600$, $B = 200$, and the oracle choices $d = \ell = s_0$, to project into $s_0$ dimensions.  Thereafter, both algorithms apply $K$-means to the projected data as above.  We also explored the option of replacing the $K$-means steps in these latter algorithms with the EM algorithm, but observed very little difference, so do not report these results here.

Given true labels $y_1,\ldots,y_n \in [K]$ and estimated labels $\hat{y}_1,\ldots,\hat{y}_n \in [K]$ from a clustering algorithm, we measure the performance of the algorithm via its \emph{misclustering rate}, defined as\footnote{Here, the minimum over permutations is taken because it is only the cluster groupings, and not the labels themselves, that are important.}
\[
L(\{y_1,\ldots,y_n\},\{\hat{y}_1,\ldots,\hat{y}_n\}) := \min_{\sigma \in \mathcal{S}_K} \frac{1}{n} \sum_{i=1}^n \mathbbm{1}_{\{ \sigma(\hat{y}_i) \neq y_i \}},
\]
where $S_K$ is the group of all permutations of $[K]$.  In particular, Figure~\ref{fig:aniso_iso_results} presents the average misclustering rates over 100 Monte Carlo repetitions of the different high-dimensional clustering algorithms described above.  Across two different dimensions $p \in \{200,600\}$, isotropic and anisotropic settings, and for different values of $n$ and $\mathrm{SNR}$, we see a consistent picture of the \texttt{Sharp-SSL} algorithm combined with EM producing the lowest misclustering rates, often by a large margin.  
Indeed, for all but the smallest sample sizes or values of $\mathrm{SNR}$, the \texttt{Sharp-SSL+EM} algorithm nearly attains the Bayes risk in all of the problems considered here.

\subsection{Effect of observed fraction on misclustering rate}
\label{SubSec:gammafraction}

One of the key attractions of our procedure is that it offers a unified framework to perform classification or clustering with an arbitrary fraction of labeled observations.  In this subsection, we explore the performance of the algorithm as we vary the proportion of observed labels. 

\begin{figure}
\begin{center}
\begin{tabular}{cc}
$\mathrm{SNR}=3.5$, $p=200$, $s=3$ & $\mathrm{SNR}=3.5$, $p=600$, $s=3$\\ 
 \includegraphics[width=0.45\textwidth]{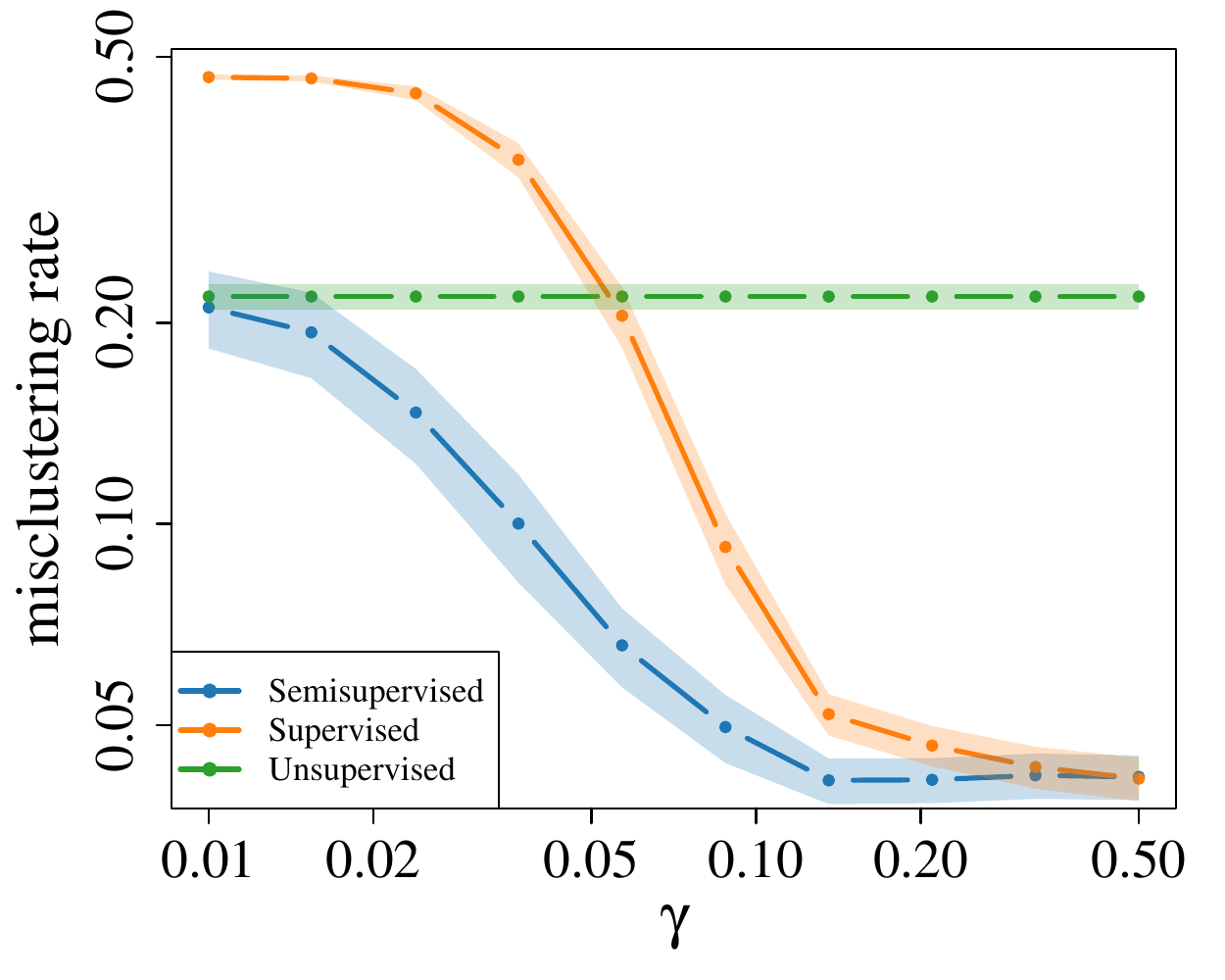} & \includegraphics[width=0.45\textwidth]{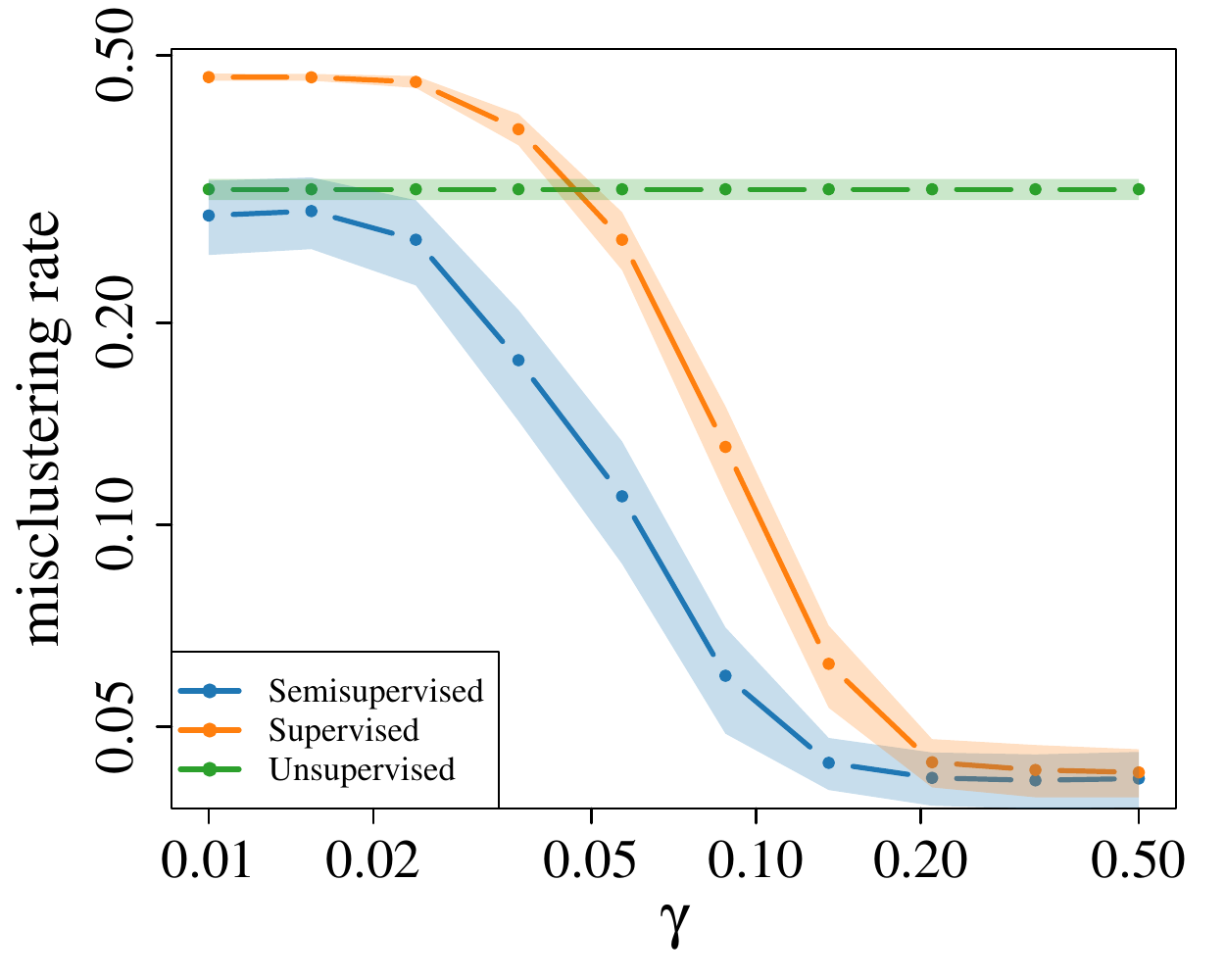} \\
 $\mathrm{SNR}=1.5$, $p=s=3$& $\mathrm{SNR}=2$, $p=s=3$\\
\includegraphics[width=0.45\textwidth]{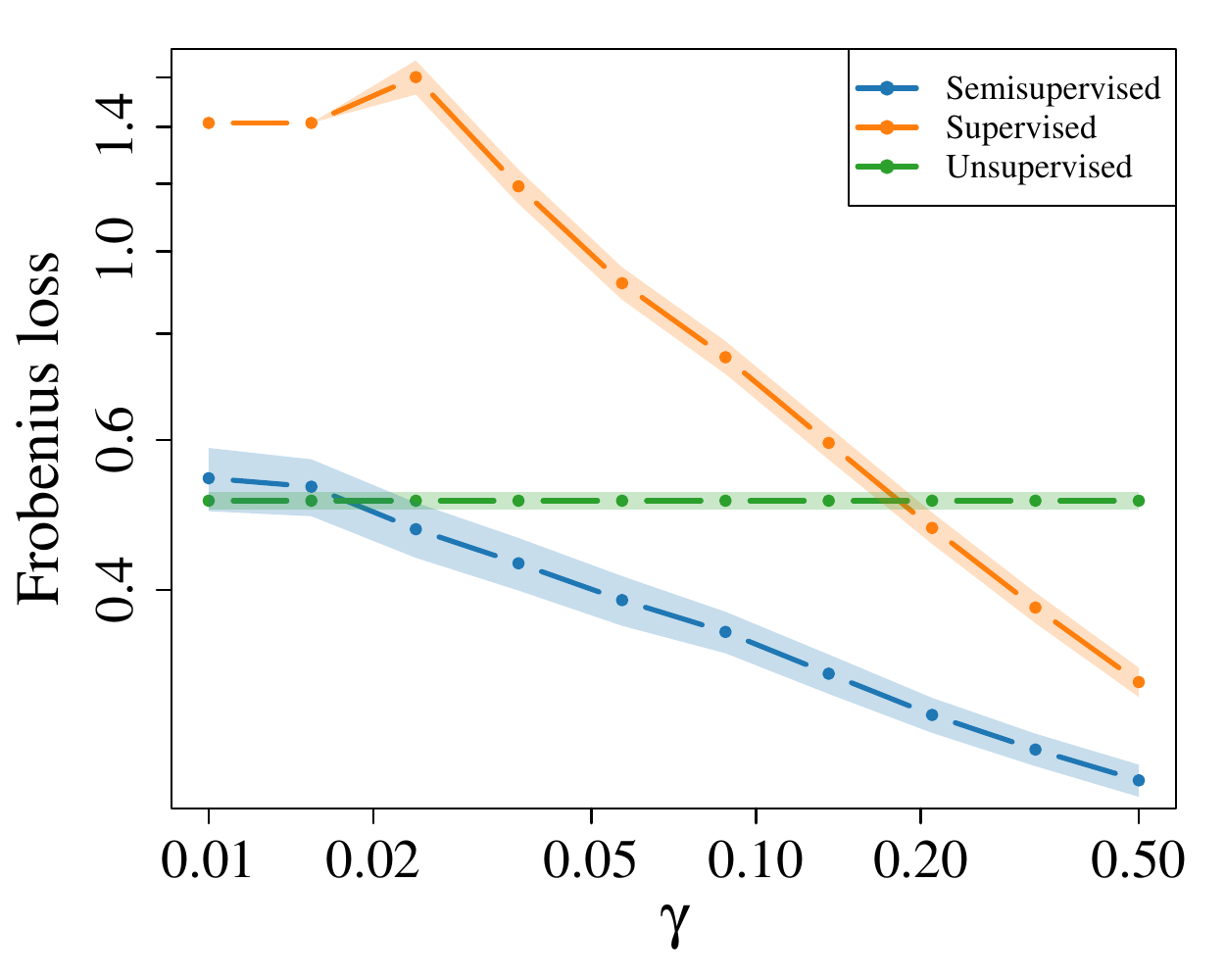} & \includegraphics[width=0.45\textwidth]{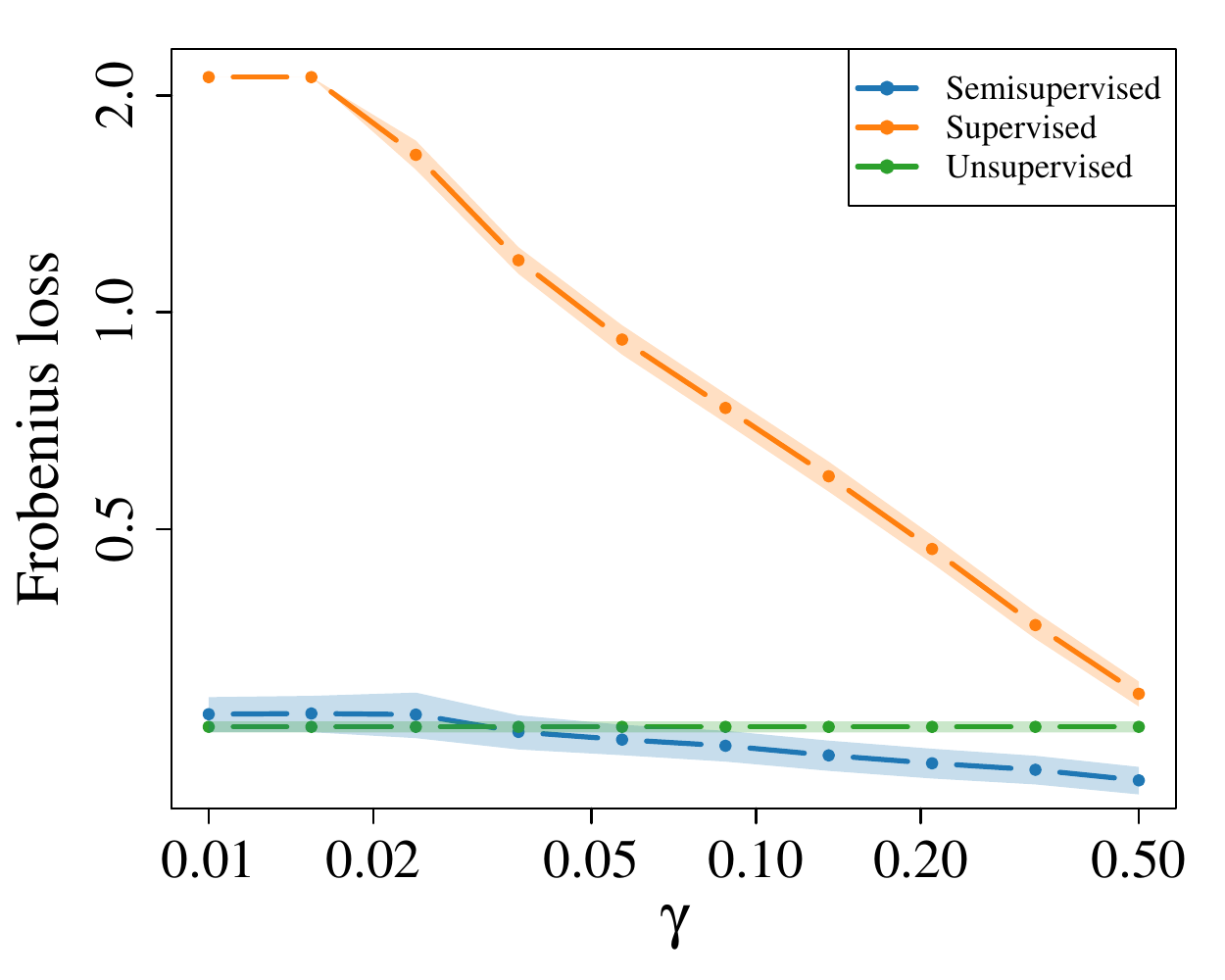}
\end{tabular}
    \caption{\label{Fig:EMSemisupervised} Effect of label fraction on performance of supervised, semi-supervised and unsupervised \texttt{Sharp-SSL} learning methods. Data are generated from the normal mixture distribution described at the beginning of Section~\ref{Sec:Numerics} with $K=2$ and $\Sigma_{\mathrm{w}} = I_p$, $\mu_1 = -\mu_2 = a(\boldsymbol{1}_s, \boldsymbol{0}_{p-s})^\top \in \mathbb{R}^p$, where $a$ is chosen such that $\|\mu_1-\mu_2\| = \mathrm{SNR}$.  Bottom: average Frobenius loss of estimating the $(\mu_1,\mu_2) \in \mathbb{R}^{p \times 2}$ over 100 repetitions via the semi-supervised approach (Algorithm~\ref{Algo:EM}), supervised approach (Algorithm~\ref{Algo:LDA}) and unsupervised approach (Algorithm~\ref{Algo:EM} without using the labels).  Top: average misclustering rate over 100 repetitions from applying the above three methods as base algorithms in Algorithm~\ref{Algo:Sharp-SSL}.  The shaded regions represent interpolated 95\% confidence intervals at each of the~points.}
    \end{center}
\end{figure}

Recall that we have two different options for the way in which we implement the \texttt{Sharp-SSL} algorithm to estimate the set of signal coordinates: we can either use only the labeled data, as in the supervised learning approach of Algorithm~\ref{Algo:LDA}, or we can try to leverage in addition the unlabeled data via the semi-supervised EM approach of Algorithm~\ref{Algo:EM}.  In the extreme case of this latter version, we have no labeled data, so the algorithm is unsupervised.  In Figure~\ref{Fig:EMSemisupervised} we compare the performance of these three methods in both high- and low-dimensional versions of the normal mixture distribution data generation mechanism described at the beginning of Section~\ref{Sec:Numerics} as the proportion $\gamma$ of observed labels varies.  

More precisely, for the semi-supervised and unsupervised algorithms, we adopt the same implementation of \texttt{Sharp-SSL} as described at the beginning of Section~\ref{SubSec:Comparison}.  The supervised algorithm is very similar, but applies Algorithm~\ref{Algo:LDA} in place of Algorithm~\ref{Algo:EM} to select coordinates, and obtains final predicted labels by applying LDA again on the projected labeled data.  In cases where the proportion of labeled data was so small that the convex hull of the projected labeled data was less than full-dimensional for every class, we forced Algorithm~\ref{Algo:LDA} to return a zero matrix (this only happened when $\gamma$ was very small).

The top panels of Figure~\ref{Fig:EMSemisupervised} present the results in high-dimensional settings with $p \in \{200,600\}$.  Since the unsupervised approach has no access to the labels, it has constant misclustering rate.  The performance of the semi-supervised approach is always at least as good as that of the unsupervised algorithm, and improves as $\gamma$ increases. 
In other words, it effectively leverages  the additional information provided by the class labels.  When $\gamma$ is very small, the supervised algorithm---which ignores the unlabeled data---is inaccurate, as it has very little data to work with. 
On the other hand, its performance also improves as $\gamma$ increases, and once around 5$\%$ of our data are labeled, it outperforms the unsupervised algorithm. 
Further, it essentially matches the semi-supervised approach when about a third of the data are labeled.  We truncate the plot at $\gamma = 1/2$ to ensure that we have enough test data on which to compute the misclustering rate.

In the bottom panels of Figure~\ref{Fig:EMSemisupervised}, we explore the performance of the three algorithms above in two low-dimensional settings with different values of \texttt{SNR}, in order to provide further insight into the phenomena described in the previous paragraph.  Here, we take $K=2$ and report the average Frobenius norm loss
\[
\mathcal{L}\bigl((\hat{\mu}_1,\hat{\mu}_2),(\mu_1,\mu_2)\bigr) := \min\bigl\{\|(\hat{\mu}_1,\hat{\mu}_2) - (\mu_1,\mu_2)\|_{\mathrm{F}},\|(\hat{\mu}_2,\hat{\mu}_1) - (\mu_1,\mu_2)\|_{\mathrm{F}}\bigr\}
\]
of the estimated means, over 100 repetitions.  If there are insufficient labeled data to run Algorithm~\ref{Algo:LDA}, then we output $\hat{\mu}_1 = \hat{\mu}_2 = \boldsymbol{0}_p$.  We see that, already in these low-dimensional problems, a similar picture emerges: if the proportion of labeled data is small, then the unsupervised algorithm outperforms the supervised one, but this situation may be reversed when $\gamma$ is larger.  The semi-supervised algorithm is able to leverage both the unlabeled and labeled data to obtain the best of both worlds.  These empirical observations agree with our theory from Section~\ref{Sec:Theory}, in particular in the way in which Theorem~\ref{Thm:LowDimEM} bounds the accuracy of mean estimation for the semi-supervised algorithm by a minimum of a term that does not depend on $\gamma$ and one that decreases as $\gamma$ increases.  
It appears that the switch in the minimum occurs around $\gamma=0.02$ in these examples.

\subsection{Empirical data analysis}
\label{SubSec:EmpiricalData}
In this subsection we apply the \texttt{Sharp-SSL} algorithm, as well as several competing methods, to the gene expression data set from \citet{alon1999broad}, which contains observations on 62 patients.  A preprocessed version of the data can be downloaded from the \textsf{R} package `datamicroarray' \citep{ramey2016collection}, with a total of $2000$ features (genes) measured on $40$ patients with colon tumors and $22$ without tumors.  We first exclude 9 genes to remove perfect collinearity and then standardize each of the remaining $p=1991$ columns of the dataset to have unit variance. 

We apply the \texttt{Sharp-SSL} algorithm using EM (Algorithm~\ref{Algo:EM}) as the base procedure, with input parameters $A=150$, $B=75$, $d=\ell=5$. In addition to our approach (\texttt{Sharp-SSL+EM}), we also compare the performance of the spectral clustering (\texttt{SC}) method, the \citet{WittenTibshirani2010} method (\texttt{WT2010}, as well as four two-stage methods (\texttt{PCA+Kmeans}, \texttt{PCA+EM}, \texttt{SPCA+Kmeans}, \texttt{SPCA+EM}), where we first reduce dimension of the data to a $5$-dimensional subspace using either PCA or SPCA and then apply either the EM algorithm or $K$-means clustering on the low-dimensional data. For SPCA, we use the \texttt{SPCAvRP} algorithm \citep{GWS2020} with inputs $A=600$, $B=200$ and $d=\ell=5$. The true labels are hidden to all algorithms and are only used to evaluate the final misclustering rate.  

Over 100 Monte Carlo repetitions of the randomized algorithms, the \texttt{Sharp-SSL+EM} method had an average misclustering rate of 28.8\%, whereas all other competitors had a misclustering rate above 40\%, as can be seen from the right-hand data points in Figure~\ref{Fig:Colontumor}.  To investigate this performance further, we applied each method to a subset of the features.
These were constructed from the top $\ell=5$ genes identified through \texttt{Sharp-SSL}, together with $m=0, 10, 50, 200$ and $600$ randomly chosen genes from the remaining $1986$.  The results are presented as the other data points in Figure~\ref{Fig:Colontumor}.  We see that the improved performance of the \texttt{Sharp-SSL+EM} relative to the other methods persists, even when only a small number of potentially non-discriminative covariates are present.  When $m=0$, \texttt{Sharp-SSL+EM} has a slight disadvantage as other algorithms benefit from the ensemble effect of combining two different learning methods; nevertheless it remains competitive.  This reinforces the point that the primary contribution of the \texttt{Sharp-SSL} algorithm is to identify signal coordinates that are helpful for semi-supervised learning, and once this task has been accomplished, a variety of low-dimensional procedures are available to the practitioner. 

\begin{figure}
    \centering
    \includegraphics[width=0.6\textwidth]{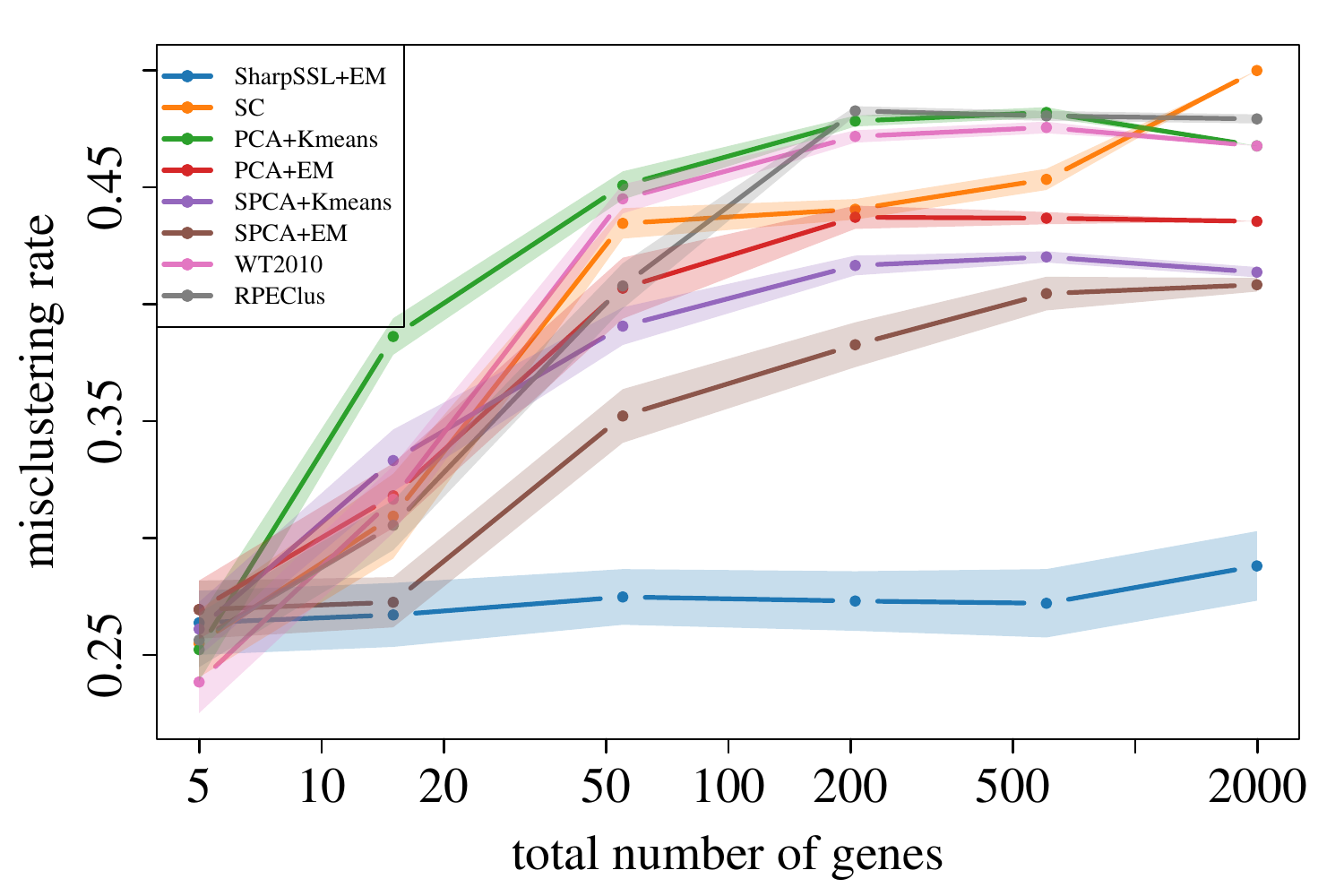}
    \caption{Average misclustering rate (over 100 repetitions for randomized algorithms) for the colon tumor data, using \texttt{Sharp-SSL} followed by the EM algorithm, as well as the other methods described in Section~\ref{SubSec:EmpiricalData}.  The right-hand data points plot the average misclustering rate on the full data set.  The other points were obtained by applying each method to a subset of genes formed from the top five genes identified by \texttt{Sharp-SSL} together with randomly sampled genes.  The shaded regions represent interpolated 95\% confidence intervals at each of the points. }
    \label{Fig:Colontumor}
\end{figure}

\section{Proof of the main results}
\label{Sec:Proofs}
\subsection{Proof of Proposition~\ref{Prop:Motivation}}
\label{pf:Prop:Motivation}

The assumption that the convex hull of $\nu_1,\ldots,\nu_K$ is $(K-1)$-dimensional implies that $\Sigma_{\mathrm{b}}$ is of rank $K-1$.  Define $A:= \Sigma_{\mathrm{w}}^{-1/2}\Sigma_{\mathrm{b}}\Sigma_{\mathrm{w}}^{-1/2} \in \mathbb{R}^{p \times p}$, which has rank $K-1$.  Given $V\in\mathbb{O}^{p\times d}$, we can find $Q \in \mathbb{O}^{p\times d}$ with the same column span as that of $\Sigma_{\mathrm{w}}^{1/2}V$ and let $R:=Q^\top \Sigma_{\mathrm{w}}^{1/2}V \in \mathbb{R}^{d\times d}$, so that $R$ is invertible, and $\Sigma_{\mathrm{w}}^{1/2}V = QR$. We observe that 
\[
\mathrm{tr}\{(V^\top\Sigma_{\mathrm{w}} V)^{-1}(V^\top \Sigma_{\mathrm{b}} V)\} = \mathrm{tr}\{(R^\top R)^{-1}(R^\top Q^\top  A Q R)\} = \mathrm{tr}(Q^\top A Q).
\]
Thus, $J(V;\Sigma_{\mathrm{b}}, \Sigma_{\mathrm{w}})$ depends on $V$ only through the column space of $\Sigma_{\mathrm{w}}^{1/2}V$. Moreover, $\mathrm{tr}(Q^\top A Q)$ is maximized when $Q$, or equivalently $\Sigma_{\mathrm{w}}^{1/2}V$, spans a $d$-dimensional space that contains the $(K-1)$-dimensional eigenspace corresponding to the non-zero eigenvalues of $A$. Note that if for some $v \in \mathbb{R}^p \setminus \{0\}$ and $\lambda \geq 0$, we have $Av = \lambda v$, then $\Sigma_{\mathrm{w}}^{-1}\Sigma_{\mathrm{b}}\Sigma_{\mathrm{w}}^{-1/2}v = \Sigma_{\mathrm{w}}^{-1/2}Av= \lambda \Sigma_{\mathrm{w}}^{-1/2}v$, so $\Sigma_{\mathrm{w}}^{-1/2}v$ is an eigenvector of $\Sigma_{\mathrm{w}}^{-1}\Sigma_{\mathrm{b}}$ with eigenvalue $\lambda$. 
Hence $V$ maximizes $J(V;\Sigma_{\mathrm{b}}, \Sigma_{\mathrm{w}})$ over $\mathbb{O}^{p \times d}$ if and only if
$V$ spans
a $d$-dimensional space that contains the
$(K-1)$-dimensional eigenspace corresponding to the $K-1$ non-zero eigenvalues of $\Sigma_{\mathrm{w}}^{-1}\Sigma_{\mathrm{b}}$. Finally, for any $v \in \mathbb{R}^p \setminus \{0\}$,
\[
v^\top \Sigma_{\mathrm{w}}^{-1/2}\Sigma_{\mathrm{b}}\Sigma_{\mathrm{w}}^{-1/2}v = \sum_{k=1}^K \pi_k v^\top \Sigma_{\mathrm{w}}^{-1/2}(\nu_k-\nu)(\nu_k-\nu)^\top \Sigma_{\mathrm{w}}^{-1/2}v \neq 0
\]
if and only if $v^\top \Sigma_{\mathrm{w}}^{-1/2}(\nu_k-\nu)\neq 0$ for some $k \in [K]$. Thus, the eigenspace corresponding to the non-zero eigenvalues of $A$ is spanned by $\bigl(\Sigma_{\mathrm{w}}^{-1/2}(\nu_k-\nu): k\in[K]\bigr)$, and so the eigenspace corresponding to non-zero eigenvalues of $\Sigma_{\mathrm{w}}^{-1}\Sigma_{\mathrm{b}}$ is spanned by $\bigl(\Sigma_{\mathrm{w}}^{-1}(\nu_k-\nu):k \in [K]\bigr)$.

\subsection{Proof of Theorem~\ref{Thm:Meta}}
\label{pfThm:Meta}

We write $S^{a,b}:=\{j\in[p]: (P^{a,b,\top}P^{a,b})_{j,j} = 1\}$. For any $S = \{j_1,\ldots,j_d\}^\top \in\binom{[p]}{d}$, 
we identify the set $S$ with the sequence 
$j_{i_1}< \ldots < j_{i_d}$ sorted in increasing order; and with a slight abuse of notation, we will use $S$ to refer to either object, which will always be clear depending on the context.
We define $P^S \in \mathcal{P}_d$ by $(P^S)_{\ell,j} := \mathbbm{1}_{\{j = j_\ell\}}$, 
so that $P^{S,\top}P^S = \mathrm{diag}\bigl((\mathbbm{1}_{\{j\in S\}})_{j\in[p]}\bigr)$. 
Define $Q^S := (P^S\Sigma_{\mathrm{w}}P^{S,\top})^{-1}P^S\Sigma_{\mathrm{b}}P^{S,\top} \in \mathbb{R}^{d\times d}$ and $\hat Q^S:=\psi\bigl((P^SX_i,Y_i)_{i\in[n]}\bigr) \in \mathbb{R}^{d\times d}$.
Note that $\hat Q^{a,b} = \hat Q^{S^{a,b}}$ in this notation, and we will similarly denote $Q^{a,b} := Q^{S^{a,b}}$ for simplicity.  Recalling the definition of $\Omega$ from~\eqref{Eq:Omega}, in our new notation, 
and recalling that $\psi$ is permutation-equivariant,
we can write
\begin{align*}
  \Omega = \biggl\{\max_{S \in \binom{[p]}{d}}\bigl\|\hat Q^S - Q^S\bigr\|_{\mathrm{op}} < \frac{\gamma_{\min}}{4(K-1)}\biggr\},
\end{align*}
and have $\mathbb{P}(\Omega)\geq 1-\epsilon$ by~\eqref{Eq:AssumptionUniformConsistency}. We will work on the event $\Omega$ throughout the remainder of the proof, and assume also that $\gamma_{\min} > 0$, because otherwise the conclusion is trivial.

By Weyl's inequality \citep[Corollary IV.4.9]{Weyl1912, StewartSun1990}, we have on $\Omega$ that 
\begin{align*}
|\mathrm{tr}(\hat Q^S) - \mathrm{tr}(Q^S)|\leq (K-1)\|\hat Q^S - Q^S\|_{\mathrm{op}} \leq \frac{\gamma_{\min}}{4}.
\end{align*}
On the other hand, we note that $\mathrm{tr}(Q^S) = \sum_{j\in S\cap S_0} (\Sigma_{\mathrm{w}}^{-1}\Sigma_{\mathrm{b}})_{j,j}$. Therefore, by the triangle inequality, for any $S,S'\in\binom{[p]}{d}$ such that $S\cap S_0$ is a proper subset of $S'\cap S_0$, we have on $\Omega$ that 
\begin{equation}
\label{Eq:JComparison}
\mathrm{tr}(\hat Q^S) - \mathrm{tr}(\hat Q^{S'}) \leq \frac{\gamma_{\min}}{2} - \!\!\sum_{j\in  (S'\setminus S)\cap S_0} (\Sigma_{\mathrm{w}}^{-1}\Sigma_{\mathrm{b}})_{j,j} < 0.
\end{equation}

Fix $a\in[A]$, and for any $\tilde j \in [p]$, define $q_{\tilde j}:=\mathbb{P}\bigl(\tilde j \in S^{a,b^*(a)} \mid (X_i,Y_i)_{i\in[n]}\bigr)$. Now fix some $j\in S_0$ and $j'\in [p]\setminus S_0$. We claim that
\begin{equation}
\label{Eq:ClaimOrdering}
(q_j - q_{j'})\mathbbm{1}_{\Omega} \geq 0.
\end{equation}
To verify this claim, define for $\tilde j\in\{j,j'\}$ and $b\in[B]$ the sets
\[
\mathcal{S}_{b,\tilde j}:=\bigl\{(S^{a,1},\ldots,S^{a,B}): b^*(a)=b, \tilde j \in S^{a,b}\bigr\}\quad \text{and}\quad \mathcal{S}_b :=\bigl\{(S^{a,1},\ldots,S^{a,B}): b^*(a)=b\bigr\}.
\]
Let $f:\binom{[p]}{d}\to\binom{[p]}{d}$ be a map defined by
\[
f(S):=\begin{cases}
(S\setminus \{j'\})\cup\{j\} &  \text{if $j\notin S$ and $j'\in S$}\\
S &\text{otherwise}.
\end{cases}
\]
If $j \notin S^{a,b}$ and $j' \in S^{a,b}$, then $f(S^{a,b}) \cap S_0 = (S^{a,b} \cup \{j\}) \cap S_0 = (S^{a,b} \cap S_0) \cup \{j\}$, so $S^{a,b} \cap S_0$ is a proper subset of $f(S^{a,b}) \cap S_0$; on the other hand, if either $j \in S^{a,b}$ or $j,j' \notin S^{a,b}$, then $f(S^{a,b}) = S^{a,b}$.  It follows by~\eqref{Eq:JComparison} that on $\Omega$ we have
\begin{equation}
  \label{Eq:trace}
  \mathrm{tr}(\hat Q^{S^{a,b}}) \leq \mathrm{tr}(\hat Q^{f(S^{a,b})}).
\end{equation}
Now let $F: \mathcal{S}_{b,j'}\to\mathcal{S}_{b,j}$ be defined as
\[
F(S^{a,1},\ldots,S^{a,B}) := \bigl(S^{a,1},\ldots,S^{a,b-1}, f(S^{a,b}), S^{a,b+1},\ldots,S^{a,B}\bigr).
\]
We claim that $F$ is both well-defined and injective on $\Omega$.  For the first of these claims, we note that since $j' \in S^{a,b}$, we must have $j \in f(S^{a,b})$.  Moreover, if $(S^{a,1},\ldots,S^{a,B}) \in \mathcal{S}_{b,j'}$, then $b^*(a) = b$.  But~\eqref{Eq:trace} holds on $\Omega$, so $\bigl(S^{a,1},\ldots,S^{a,b-1}, f(S^{a,b}), S^{a,b+1},\ldots,S^{a,B}\bigr) \in \mathcal{S}_{b,j}$.  Hence~$F$ is well-defined.  For the second claim, suppose that $S_1,S_2 \in \binom{[p]}{d}$ are such that $j'\in S_1 \cap S_2$ and $f(S_1) = f(S_2)$.  If $j \in S_1 \cap S_2$, then $S_1 = f(S_1) = f(S_2) = S_2$; if $j \in S_1$ but $j \notin S_2$, then $j \in f(S_2) \setminus f(S_1)$, a contradiction.  Similarly, we cannot have $j \notin S_1$ but $j \in S_2$.  Finally, if $j \in S_1^{\mathrm{c}} \cap S_2^{\mathrm{c}}$, then $S_1 = \bigl(f(S_1) \setminus \{j\}\bigr) \cup \{j'\} = \bigl(f(S_2) \setminus \{j\}\bigr) \cup \{j'\} = S_2$.  We deduce that~$f$ is injective on $\{S: j'\in S\}$. Since $j' \in S^{a,b}$ for $(S^{a,1},\ldots,S^{a,B}) \in \mathcal{S}_{b,j'}$, this establishes the injectivity of $F$.  In particular, $|\mathcal{S}_{b,j'}|\leq |\mathcal{S}_{b,j}|$. Consequently, on $\Omega$, we have for all $b\in[B]$ that
\begin{align*}
\mathbb{P}\bigl(j\in S^{a,b^*(a)}\mid (X_i,Y_i)_{i\in[n]}, b^*(a)=b\bigr) 
&= \frac{\mathbb{P}\bigl(j\in S^{a,b^*(a)}, b^*(a)=b\mid (X_i,Y_i)_{i\in[n]}\bigr)}{\mathbb{P}\bigl(b^*(a)=b\mid (X_i,Y_i)_{i\in[n]}\bigr)} 
= \frac{|\mathcal{S}_{b,j}|}{|\mathcal{S}_b|}\\
&\geq \frac{|\mathcal{S}_{b,j'}|}{|\mathcal{S}_b|} = \frac{\mathbb{P}\bigl(j'\in S^{a,b^*(a)}, b^*(a)=b\mid (X_i,Y_i)_{i\in[n]}\bigr)}{\mathbb{P}\bigl(b^*(a)=b\mid (X_i,Y_i)_{i\in[n]}\bigr)} \\
&= \mathbb{P}\bigl(j'\in S^{a,b^*(a)}\mid (X_i,Y_i)_{i\in[n]}, b^*(a)=b\bigr),
\end{align*}
which implies Claim~\eqref{Eq:ClaimOrdering}. We remark that one consequence of~\eqref{Eq:ClaimOrdering} is that, since $d \geq s_0$, we have on $\Omega$ that
\begin{equation}
\label{Eq:SignalProbLowerBound}
q_j\geq \frac{\sum_{\tilde j\in([p]\setminus S_0)\cup\{j\}} q_{\tilde j}}{p-s_0+1} = \frac{d - \sum_{\tilde j\in S_0\setminus \{j\}} q_{\tilde j}}{p-s_0+1}\geq \frac{d-s_0+1}{p-s_0+1} \geq \frac{1}{p}.
\end{equation}

Again fixing $j\in S_0$ and $j'\notin S_0$, we observe on $\Omega\cap\{j\in S^{a,b^*(a)}\}$ that
\begin{align*}
 \frac{3}{4}\gamma_{\min} &\leq \bigl[P^{a,b^*(a),\top} Q^{a,b^*(a)} P^{a,b^*(a)}\bigr]_{j,j} - \bigl\|\hat Q^{a,b^*(a)} \!-\! Q^{a,b^*(a)}\bigr\|_{\mathrm{op}} \leq \bigl[P^{a,b^*(a),\top}\hat Q^{a,b^*(a)} P^{a,b^*(a)}\bigr]_{j,j} \\
&\leq [P^{a,b^*(a),\top} Q^{a,b^*(a)} P^{a,b^*(a)}]_{j,j} + \bigl\|\hat Q^{a,b^*(a)} - Q^{a,b^*(a)}\bigr\|_{\mathrm{op}} \leq \frac{5}{4}\gamma_{\max},
\end{align*}
and similarly on $\Omega\cap\{j'\in S^{a,b^*(a)}\}$ that $\bigl|\bigl[P^{a,b^*(a),\top}\hat Q^{a,b^*(a)}P^{a,b^*(a)}\bigr]_{j',j'}\bigr| \leq \gamma_{\min}/4$. Recall also that $\bigl[P^{a,b^*(a),\top}\hat Q^{a,b^*(a)}P^{a,b^*(a)}\bigr]_{{\tilde j,\tilde j}} = 0$ for all $\tilde j\notin S^{a,b^*(a)}$. Combining the above bounds on the diagonal entries of $\hat Q^{a,b^*(a)}$ with~\eqref{Eq:ClaimOrdering} and~\eqref{Eq:SignalProbLowerBound}, we have on $\Omega$ that
\begin{align}
  \mathbb{E}\bigl(&\bigl[P^{a,b^*(a),\top}\hat Q^{a,b^*(a)} P^{a,b^*(a)}\bigr]_{j,j} - \bigl[P^{a,b^*(a),\top}\hat Q^{a,b^*(a)} P^{a,b^*(a)}\bigr]_{j',j'} \mid (X_i,Y_i)_{i\in[n]}\bigr) \nonumber \\
            &=\mathbb{E}\bigl(\bigl[P^{a,b^*(a),\top}\hat Q^{a,b^*(a)} P^{a,b^*(a)}\bigr]_{j,j}\mathbbm{1}_{\{j\in S^{a,b^*(a)}\}} \nonumber \\
  &\hspace{2cm}- \bigl[P^{a,b^*(a),\top}\hat Q^{a,b^*(a)} P^{a,b^*(a)}\bigr]_{j',j'}\mathbbm{1}_{\{j'\in S^{a,b^*(a)}\}} \mid (X_i,Y_i)_{i \in [n]}\bigr)\nonumber\\
&\geq \frac{q_j\gamma_{\min}}{2} \geq \frac{\gamma_{\min}}{2p}.\label{Eq:MeanDiff}
\end{align}
Now, let $a,j,j'$ be freely varying again. Since $\hat w_j = A^{-1}\sum_{a\in[A]} [P^{a,b^*(a),\top}\hat Q^{a,b^*(a)} P^{a,b^*(a)}]_{j,j}$, on $\Omega$ we have for any $j\in S_0$ and $j'\notin S_0$ that $\mathbb{E}\bigl(\hat w_j - \hat w_{j'} \mid (X_i, Y_i)_{i\in[n]}\bigr)\geq \gamma_{\min}/(2p)$ from~\eqref{Eq:MeanDiff}. Since $\ell\geq s_0$, we have by Hoeffding's inequality that on $\Omega$,
\begin{align*}
\mathbb{P}\bigl(S_0\not\subseteq \hat S\mid (X_i,Y_i)_{i\in[n]}\bigr)&\leq \mathbb{P}\Bigl(\min_{j\in S_0}\hat w_j \leq \max_{j'\notin S_0} \hat w_{j'}\Bigm| (X_i,Y_i)_{i\in[n]}\Bigr)\\
 & \leq \sum_{j\in S_0} \mathbb{P}\biggl\{\hat w_j - \mathbb{E}(\hat w_j\mid (X_i,Y_i)_{i\in[n]}) \leq -\frac{\gamma_{\min}}{4p}\biggm| (X_i,Y_i)_{i\in[n]}\biggr\} \\
  &\quad + \sum_{j\notin S_0} \mathbb{P}\biggl\{\hat w_j - \mathbb{E}(\hat w_j\mid (X_i,Y_i)_{i\in[n]}) \geq \frac{\gamma_{\min}}{4p}\biggm| (X_i,Y_i)_{i\in[n]}\biggr\} \\
&\leq p\exp\biggl\{-\frac{A}{2}\biggl(\frac{\gamma_{\min}}{4p}\biggr)^2\biggm/ \biggl(\frac{5\gamma_{\max}}{4}\biggr)^2\biggr\}\leq pe^{-A\gamma_{\min}^2/(50p^2\gamma_{\max}^2)},
\end{align*}
as desired.

\subsection{Proof of Theorem~\ref{Thm:LDA}}
\label{pf:Thm:LDA}

The main ingredient of the proof of Theorem~\ref{Thm:LDA} is the following proposition, which controls the rate of convergence of the sample between- and within-class covariance matrices to their respective population versions in a classification problem. 
\begin{prop}[Rate of convergence for LDA]
\label{Prop:LowDimConcentration}
Suppose that $(Z_1,Y_1),\ldots,(Z_n,Y_n)\in\mathbb{R}^{d}\times [K]$ are independent and identically distributed data-label pairs, such that $\mathbb{P}(Y_1 = k) = \pi_k$ and $Z_1\mid Y_1 = k \sim \mathcal{N}_d(\mu_k, \Sigma_{\mathrm{w}})$ for $k\in[K]$. Write $\mu := \sum_{k\in[K]} \pi_k\mu_k$ and $\Sigma_{\mathrm{b}}:= \sum_{k\in[K]}\pi_k(\mu_k-\mu)(\mu_k-\mu)^\top$ and let $\hat\Sigma_{\mathrm{w}}$ and $\hat\Sigma_{\mathrm{b}}$ be computed as in~\eqref{Eq:withinbetween} applied to $(Z_1,Y_1),\ldots,(Z_n,Y_n)$. If $\|\mu_k-\mu\|\leq R_1$ for all $k\in[K]$ and $\|\Sigma_{\mathrm{w}}\|_{\mathrm{op}}\leq R_2$ for some $R_1,R_2>0$, then for every $\delta \in (0,1/4]$, we have with probability at least $1-\delta$ that 
\begin{align*}
    \|\hat \Sigma_{\mathrm{b}} - \Sigma_{\mathrm{b}}\|_{\mathrm{op}} & \leq \frac{12R_2\{K + \log(8\cdot 9^d/\delta)\}}{n} + \bigl(4R_1\sqrt{R_2} + R_1^2 + R_1\bigr)\sqrt\frac{2\log(8\cdot 9^d/\delta)}{n}, \\
    \|\hat\Sigma_{\mathrm{w}} - \Sigma_{\mathrm{w}}\|_{\mathrm{op}} &\leq \frac{4R_2\{K + \log(8\cdot 9^d/\delta)\}}{n} + 4R_2\sqrt\frac{\log(8\cdot 9^d/\delta)}{n}.
\end{align*}
\end{prop}
\begin{proof}

We first control the rate of convergence of $\hat\Sigma_{\mathrm{b}}$. For $n_k := \sum_{i=1}^n \mathbbm{1}_{\{Y_i=k\}}$ and $\tilde \mu :=\sum_{k\in[K]} (n_k/n)\mu_k$, we define  $\tilde\Sigma_{\mathrm{b}}^{(1)}:=\sum_{k=1}^K(n_k/n)(\mu_k-\mu)(\mu_k-\mu)^\top$ and $\tilde\Sigma_{\mathrm{b}}^{(2)}:=\sum_{k=1}^K(n_k/n)(\mu_k-\tilde \mu)(\mu_k-\tilde \mu)^\top$. We have the following decomposition
\begin{equation}
\label{Eq:Decomp14}
\|\hat\Sigma_{\mathrm{b}} - \Sigma_{\mathrm{b}}\|_{\mathrm{op}} \leq \|\tilde\Sigma_{\mathrm{b}}^{(1)} - \Sigma_{\mathrm{b}}\|_{\mathrm{op}} + \|\tilde\Sigma_{\mathrm{b}}^{(2)} - \tilde\Sigma_{\mathrm{b}}^{(1)}\|_{\mathrm{op}}  + \|\hat\Sigma_{\mathrm{b}} - \tilde\Sigma_{\mathrm{b}}^{(2)}\|_{\mathrm{op}}.
\end{equation}
We will control the three terms on the right-hand side above separately. For the first term, we note that $\|\mu_k -\mu\|\leq R_1$ for all $k\in[K]$. Since $n_1,\ldots, n_K$ are functions of $Y_1,\ldots,Y_n$, we have by McDiarmid's inequality \citep[see, e.g.][Theorem~6.2]{boucheron2013concentration} that with probability at least $1-\delta/4$,
\begin{equation}
\label{Eq:Sigmab1}
\|\tilde \Sigma_{\mathrm{b}}^{(1)} - \Sigma_{\mathrm{b}}\|_{\mathrm{op}} \leq R_1^2\sum_{k=1}^K\biggl|\frac{n_k}{n}-\pi_k\biggr| \leq R_1^2\sqrt\frac{2\log(4/\delta)}{n}.
\end{equation}
For the second term, we first apply McDiarmid's inequality again to see that with probability at least $1-\delta/4$, we have
\[
\|\tilde \mu- \mu\| \leq \sum_{k=1}^K \biggl|\frac{n_k}{n}-\pi_k\biggr|\|\mu_k-\mu\| \leq R_1\sqrt\frac{2\log(4/\delta)}{n}.
\]
Thus, we have with probability at least $1-\delta/4$ that
\begin{align}
\|\tilde\Sigma_{\mathrm{b}}^{(2)} - \tilde \Sigma_{\mathrm{b}}^{(1)}\|_{\mathrm{op}} &\leq 2\sum_{k=1}^K \frac{n_k}{n}\|(\mu_k-\mu)(\mu-\tilde\mu)^\top\|_{\mathrm{op}} + \sum_{k=1}^K \frac{n_k}{n} \|(\mu-\tilde \mu)(\mu-\tilde\mu)^\top\|_{\mathrm{op}}\nonumber\\
&\leq 2R_1\|\mu-\tilde\mu\| + \|\mu-\tilde\mu\|^2 \leq 3R_1\|\mu-\tilde\mu\|\leq 3R_1^2\sqrt\frac{2\log(4/\delta)}{n}.\label{Eq:Sigmab3}
\end{align}
Finally, for the third term, we write $\hat{\mu}_k := \sum_{i:Y_i=k} Z_i$ and note that $V_k := n_k^{1/2}\hat{\mu}_k$ satisfies $V_k\mid Y_1,\ldots,Y_n \sim \mathcal{N}_d(n_k^{1/2}\mu_k, \Sigma_{\mathrm{w}})$.  Defining $N := (n_1^{1/2},\ldots,n_K^{1/2})^\top$ and $P := NN^\top/n \in \mathbb{R}^{K \times K}$, we may write
\[
  n\hat\Sigma_{\mathrm{b}} = \sum_{k=1}^K n_k\hat{\mu}_k\hat{\mu}_k^\top - n\hat{\mu}\hat{\mu}^\top = V^\top (I_K-P)V,
\]
where $V := (V_1,\ldots,V_K)^\top$, and where $\hat{\mu} := n^{-1}\sum_{i=1}^n Z_i = \sum_{k=1}^K (n_k/n)\hat{\mu}_k$.  By Lemma~\ref{Lem:CochranVariant}, we deduce that $n\hat\Sigma_{\mathrm{b}}$ conditional on $Y_1,\ldots,Y_n$ has a $d$-dimensional non-central Wishart distribution with $K-1$ degrees of freedom, covariance matrix $\Sigma_{\mathrm{w}}$ and non-centrality matrix $n\tilde \Sigma_{\mathrm{b}}^{(2)}$, which we denote as 
\[
n\hat\Sigma_{\mathrm{b}}\mid Y_1,\ldots,Y_n \sim \mathcal{W}_d(K-1, \Sigma_{\mathrm{w}}; n\tilde \Sigma_{\mathrm{b}}^{(2)});
\]
a formal definition is given just before Lemma~\ref{Lem:CochranVariant}. For any fixed $u\in\mathbb{S}^{d-1}$, we have by \citet[Theorem~10.3.6]{muirhead2009aspects} that
\[
u^\top \hat\Sigma_{\mathrm{b}} u \mid Y_1,\ldots,Y_n \sim \frac{u^\top \Sigma_{\mathrm{w}} u}{n}\chi^2_{K-1}\biggl(\frac{nu^\top \tilde\Sigma_{\mathrm{b}}^{(2)}u}{u^\top\Sigma_{\mathrm{w}}u}\biggr),
\]
where $\chi^2_r(\lambda)$ denotes a non-central chi-squared distribution with $r$ degrees of freedom and non-centrality parameter $\lambda$.
By \citet[Lemma~8.1]{Birge2001}, for every $\delta'\in (0,1/2]$, we have with probability at least $1-2\delta'$ conditional on $Y_1,\ldots,Y_n$, that
\begin{align*}
|u^\top(\hat\Sigma_{\mathrm{b}} - \tilde\Sigma_{\mathrm{b}}^{(2)}) u| &\leq \frac{u^\top \Sigma_{\mathrm{w}} u}{n}\biggl\{K+2\sqrt{\biggl(K+\frac{2nu^\top\tilde \Sigma_{\mathrm{b}}^{(2)}u}{u^\top\Sigma_{\mathrm{w}}u}\biggr)\log(1/\delta')} + 2\log(1/\delta')\biggr\}\\
&\leq  \frac{u^\top \Sigma_{\mathrm{w}} u}{n}\bigl\{2K + 3\log(1/\delta')\bigr\} + \sqrt\frac{8u^\top\Sigma_{\mathrm{w}} uu^\top\tilde\Sigma_{\mathrm{b}}^{(2)} u\log(1/\delta')}{n}\\
&\leq \frac{3\bigl(K+\log(1/\delta')\bigr)\|\Sigma_{\mathrm{w}}\|_{\mathrm{op}}}{n} + \sqrt\frac{8\|\Sigma_{\mathrm{w}}\|_{\mathrm{op}}\|\tilde\Sigma_{\mathrm{b}}^{(2)}\|_{\mathrm{op}}\log(1/\delta')}{n}.
\end{align*}
Let $\mathcal{N}$ be a 1/4-net of the sphere $\mathbb{S}^{d-1}$, which can be chosen to have cardinality at most $9^d$ \citep[Lemma~5.2]{Vershynin2012}. Hence, through a union bound, and taking $\delta' := \delta/(8\cdot 9^d)$, we have with probability at least $1-\delta/4$  conditional on $Y_1,\ldots,Y_n$ that
\begin{align}
\|\hat\Sigma_{\mathrm{b}} - \tilde\Sigma_{\mathrm{b}}^{(2)}\|_{\mathrm{op}} &\leq 2\max_{u\in\mathcal{N}} |u^\top(\hat\Sigma_{\mathrm{b}} - \tilde\Sigma_{\mathrm{b}}^{(2)}) u|\nonumber\\
&\leq \frac{6R_2\bigl(K+\log(8\cdot 9^d/\delta)\bigr)}{n} + \sqrt\frac{32 R_1^2R_2\log(8\cdot 9^d/\delta)}{n}.\label{Eq:Sigmab2}
\end{align}
Combining~\eqref{Eq:Decomp14}, \eqref{Eq:Sigmab1}, \eqref{Eq:Sigmab3} and~\eqref{Eq:Sigmab2}, we have that  the desired bound on $\|\hat\Sigma_{\mathrm{b}}-\Sigma_{\mathrm{b}}\|_{\mathrm{op}}$ occurs on an event with probability at least $1-3\delta/4$.

We now turn to control $\|\hat\Sigma_{\mathrm{w}} - \Sigma_{\mathrm{w}}\|_{\mathrm{op}}$. Let $Q := \sum_{k \in [K]} (n_k^{-1}\mathbbm{1}_{\{Y_i=k,Y_{i'}=k\}})_{i,i'=1}^n \in \mathbb{R}^{n \times n}$, so that
  \[
    n\hat\Sigma_{\mathrm{w}} = \sum_{i=1}^n Z_iZ_i^\top - \sum_{k=1}^K n_k \hat{\mu}_k \hat{\mu}_k^\top = Z^\top(I-Q)Z,
  \]
where $Z := (Z_1,\ldots,Z_n)^\top$. It therefore follows again by Lemma~\ref{Lem:CochranVariant} that $\hat\Sigma_{\mathrm{w}} \mid Y_1,\ldots,Y_n \sim n^{-1}\mathcal{W}_d(n-K, \Sigma_{\mathrm{w}})$. Another application of \citet[Theorem~10.3.6]{muirhead2009aspects} then yields for any $u\in\mathbb{S}^{d-1}$ that 
\[
u^\top \hat\Sigma_{\mathrm{w}} u \mid Y_1,\ldots,Y_n \sim \frac{u^\top \Sigma_{\mathrm{w}} u}{n}\chi^2_{n-K}.
\]
By \citet[Lemma~1]{LaurentMassart2000},  we have with probability at least $1-2\delta' = 1-\delta/(4\cdot 9^d)$ that 
\begin{align*}
|u^\top(\hat\Sigma_{\mathrm{w}} - \Sigma_{\mathrm{w}}) u| & \leq \frac{R_2}n \bigl\{K+2\sqrt{n\log(1/\delta')} + 2\log(1/\delta')\bigr\}.
\end{align*}
Again, taking a union bound over the $1/4$-net $\mathcal{N}$, we conclude that with probability at least $1-\delta/4$, we have
\[
\|\hat\Sigma_{\mathrm{w}} - \Sigma_{\mathrm{w}}\|_{\mathrm{op}} \leq 2\max_{u\in\mathcal{N}} |u^\top(\hat\Sigma_{\mathrm{w}} - \Sigma_{\mathrm{w}}) u| \leq \frac{4R_2}{n}\Bigl\{K+\sqrt{n\log(8\cdot 9^d/\delta)} + \log(8\cdot 9^d/\delta)\Bigr\},
\]
as desired.
\end{proof}

\begin{proof}[Proof of Theorem~\ref{Thm:LDA}]
Define $\delta:=\binom{p}{d}^{-1}\epsilon$. 
Since Algorithm \ref{Algo:LDA} is permutation-equivariant,
by a union bound, it suffices to show that for every $P\in\mathcal{P}_d$, with probability at least $1-\delta$, the desired upper bound holds for $\bigl\|\psi\bigl((PX_i,Y_i)_{i\in[n]}\bigr) - (P\Sigma_{\mathrm{w}}P^\top)^{-1}P\Sigma_{\mathrm{b}}P^\top \bigr\|_{\mathrm{op}}$. Recall that $n_k:=\sum_{i=1}^n \mathbbm{1}_{\{Y_i=k\}}$. Write $\Sigma_{\mathrm{w}, P}:= P\Sigma_{\mathrm{w}}P^\top$ and $\Sigma_{\mathrm{b}, P}:= P\Sigma_{\mathrm{b}}P^\top$, $\hat\mu_{k,P}:=n_k^{-1}\sum_{i:y_i=k}PX_i$, $\hat\mu_{P}:=n^{-1}\sum_{i=1}^n PX_i$ and
\[
\hat\Sigma_{\mathrm{w}, P} := \frac{1}{n}\sum_{i=1}^n(PX_i - \hat\mu_{Y_i, P})(PX_i - \hat\mu_{Y_i, P})^\top \quad \text{and} \quad \hat\Sigma_{\mathrm{b}, P} := \sum_{k=1}^K \frac{n_k}{n}(\hat\mu_{k,P}- \hat\mu_{P})(\hat\mu_{k,P} - \hat\mu_{P})^\top.
\]
Observe that since $n \geq Kd+1$, we have $\max_{k \in [K]}n_k\geq d+1$, so $\hat{\Sigma}_{\mathrm{w},P}$ is positive definite with probability 1.  Thus, by the triangle inequality, for each $P\in\mathcal{P}_d$, we have 
\begin{align}
\label{Eq:Triangle}
\bigl\|\psi\bigl((PX_i,Y_i)_{i\in[n]}\bigr) &- (P\Sigma_{\mathrm{w}}P^\top)^{-1}P\Sigma_{\mathrm{b}}P^\top\bigr\|_{\mathrm{op}} = \bigl\|\hat\Sigma_{\mathrm{w},P}^{-1}\hat\Sigma_{\mathrm{b},P} - \Sigma_{\mathrm{w},P}^{-1}\Sigma_{\mathrm{b},P}\bigr\|_{\mathrm{op}} \nonumber \\
&\leq \|\hat\Sigma_{\mathrm{w},P}^{-1}\hat\Sigma_{\mathrm{b},P}  - \Sigma_{\mathrm{w}, P}^{-1} \hat\Sigma_{\mathrm{b},P}\|_{\mathrm{op}} + \|\Sigma_{\mathrm{w}, P}^{-1} \hat\Sigma_{\mathrm{b},P} - \Sigma_{\mathrm{w}, P}^{-1} \Sigma_{\mathrm{b},P}\|_{\mathrm{op}}.
\end{align}
By Proposition~\ref{Prop:LowDimConcentration} and our hypothesis, there is an event $\Omega_P$ with probability at least $1-\delta$, on which
\begin{equation}
\label{Eq:ApplyProp1}
\begin{aligned}
  \|\hat\Sigma_{\mathrm{w}, P} - \Sigma_{\mathrm{w}, P}\|_{\mathrm{op}} &\leq \frac{4R_2\{K+\log(8\cdot 9^d/\delta)\}}{n} + 4R_2\sqrt\frac{\log(8\cdot 9^d/\delta)}{n} \leq \frac{1}{2R_2}.\\
  \|\hat\Sigma_{\mathrm{b}, P} - \Sigma_{\mathrm{b}, P}\|_{\mathrm{op}} &\lesssim_{R_1,R_2} \frac{K + d+\log(1/\delta)}{n} +  \sqrt\frac{d+\log(1/\delta)}{n} \lesssim_{R_2} 1.
\end{aligned}
\end{equation}
Thus, for the first term in~\eqref{Eq:Triangle}, by Weyl's inequality, on $\Omega_P$, we have 
\begin{align}
\|\hat\Sigma_{\mathrm{w},P}^{-1}\hat\Sigma_{\mathrm{b},P}- \Sigma_{\mathrm{w}, P}^{-1} \hat\Sigma_{\mathrm{b},P}\|_{\mathrm{op}} &\leq \|\Sigma_{\mathrm{w}, P}^{-1} - \hat\Sigma_{\mathrm{w}, P}^{-1}\|_{\mathrm{op}} \|\hat\Sigma_{\mathrm{b},P} \|_{\mathrm{op}}\nonumber\\
&\leq \|\Sigma_{\mathrm{w}, P}^{-1}\|_{\mathrm{op}} \|\hat\Sigma_{\mathrm{w}, P}^{-1}\|_{\mathrm{op}}\|\hat\Sigma_{\mathrm{b},P} \|_{\mathrm{op}} \|\hat\Sigma_{\mathrm{w}, P} - \Sigma_{\mathrm{w}, P}\|_{\mathrm{op}}\nonumber\\
&\leq \frac{\bigl(\|\Sigma_{\mathrm{b},P} \|_{\mathrm{op}} + \|\hat{\Sigma}_{\mathrm{b},P} - \Sigma_{\mathrm{b},P} \|_{\mathrm{op}}\bigr)\|\hat\Sigma_{\mathrm{w}, P} - \Sigma_{\mathrm{w}, P}\|_{\mathrm{op}}}{\lambda_{\min}(\Sigma_{\mathrm{w},P})\bigl(\lambda_{\min}(\Sigma_{\mathrm{w},P}) -  \|\hat\Sigma_{\mathrm{w}, P} - \Sigma_{\mathrm{w}, P}\|_{\mathrm{op}}\bigr)}\nonumber \\
  &\leq \frac{R_2\bigl(R_1^2 + \|\hat{\Sigma}_{\mathrm{b},P} - \Sigma_{\mathrm{b},P} \|_{\mathrm{op}}\bigr)\|\hat\Sigma_{\mathrm{w}, P} - \Sigma_{\mathrm{w}, P}\|_{\mathrm{op}}}{\bigl(1/R_2 -  \|\hat\Sigma_{\mathrm{w}, P} - \Sigma_{\mathrm{w}, P}\|_{\mathrm{op}}\bigr)}\nonumber\\
  &\lesssim_{R_1,R_2}\|\hat\Sigma_{\mathrm{w}, P} - \Sigma_{\mathrm{w}, P}\|_{\mathrm{op}} \lesssim_{R_2} \frac{K}{n} +  \sqrt\frac{d+\log(1/\delta)}{n},\label{Eq:LDATerm2}
\end{align}
where we used~\eqref{Eq:ApplyProp1} in the penultimate inequality.
For the second term in~\eqref{Eq:Triangle}, we also have on $\Omega_P$ that 
\begin{equation}
\label{Eq:LDATerm1}
\bigl\|\Sigma_{\mathrm{w}, P}^{-1} \hat\Sigma_{\mathrm{b},P} - \Sigma_{\mathrm{w}, P}^{-1} \Sigma_{\mathrm{b},P}\bigr\|_{\mathrm{op}} \leq \|\Sigma_{\mathrm{w}, P}^{-1}\|_{\mathrm{op}}\|\hat\Sigma_{\mathrm{b},P} - \Sigma_{\mathrm{b},P}\|_{\mathrm{op}} \lesssim_{R_1,R_2} \frac{K}{n} +  \sqrt\frac{d+\log(1/\delta)}{n}.
\end{equation}
The desired result follows by combining~\eqref{Eq:LDATerm2} and~\eqref{Eq:LDATerm1}, and using the fact that $\log(1/\delta) \leq d\log(ep/d)+\log(1/\epsilon)$. 
\end{proof}

\subsection{Proofs of Proposition~\ref{Prop:OneGoodInitialization} and Theorem~\ref{Thm:LowDimEM}}
\label{pfProp:OneGoodInitialization}

In the proof of Proposition~\ref{Prop:OneGoodInitialization}, we show the convergence of the EM iterates $\hat{\mu}^{(t)}$ by analyzing their components parallel and orthogonal to $\mu^*$ separately.  Writing $\eta := \mu^*/\|\mu^*\|$, let $\alpha_t \in \mathbb{R}$, $\beta_t\geq 0$ be defined by   
\begin{equation}
\label{Eq:OrthognoalDecomposition}
\hat{\mu}^{(t)} = \alpha_t\eta + \beta_t \xi_t,
\end{equation}
where $\xi_t \in \mathbb{S}^{d-1}$ is orthogonal to $\eta$.
Our proof will combine several propositions that control $\alpha_t$ and $\beta_t$ under different conditions.  We begin by laying some groundwork and defining some quantities that will be used throughout this subsection.  

First, it will be convenient to relabel the two classes as $\{-1,1\}$ instead of $\{1,2\}$.  By the rotational symmetry of the problem, we may assume without loss of generality that $\mu^* = (s,0,\ldots,0)^\top \in \mathbb{R}^d$ for some $s \geq 0$, and that the first $n_{\mathrm{L}}$ observations are labeled (i.e.,\ $Y_i\neq 0$ for $i \in [n_{\mathrm{L}}]$).  We assume throughout this section that $s\leq r$ and $r \geq 1$.  Let $\hat{\mu}_{n_\mathrm{L}}:=n_{\mathrm{L}}^{-1}\sum_{i=1}^{n_{\mathrm{L}}}Z_iY_i$, with the convention that $\hat{\mu}_{n_\mathrm{L}}:=0$ if $n_{\mathrm{L}}=0$, and define the function $f_{n_{\mathrm{U}}}:\mathbb{R}^d \rightarrow \mathbb{R}^d$ by
\begin{equation}\label{eq:fnu}
f_{n_{\mathrm{U}}}(v) := \frac{1}{n_{\mathrm{U}}}\sum_{i=n_{\mathrm{L}}+1}^n Z_i \tanh\langle Z_i, v\rangle,
\end{equation}
with $f_{n_{\mathrm{U}}}:=0$ if $n_{\mathrm{U}}=0$. Throughout, and without further comment, we assume that $n = n_{\mathrm{L}} + n_{\mathrm{U}} \geq 2$.  In this notation, the EM update \eqref{Eq:EMUpdate} can be rewritten, defining the function $g_n:\mathbb{R}^d \rightarrow \mathbb{R}^d$, as
\[
\hat{\mu}^{(t)} 
= g_n(\hat{\mu}^{(t-1)})
:= \gamma \hat{\mu}_{n_\mathrm{L}} + (1-\gamma)f_{n_{\mathrm{U}}}(\hat{\mu}^{(t-1)}) .
\]
The corresponding population quantities are
\[
f(v) := \mathbb{E}Z_1\tanh\langle v, Z_1\rangle \quad \text{and} \quad g(v) := \gamma\mu^*+(1-\gamma)f(v).
\]
Writing $\Delta_{n_{\mathrm{U}}} := f_{n_{\mathrm{U}}} - f$, we have
\begin{equation}
\label{Eq:EMIteration}
g_n(v) = g(v) + (1-\gamma)\Delta_{n_{\mathrm{U}}}(v) + \gamma (\hat{\mu}_{n_\mathrm{L}}-\mu^*).
\end{equation}
For $\omega, \phi > 0$ and $r \geq 1$, we define the following two events that control the terms in the EM iteration involving the unlabeled and labeled data respectively:
\begin{equation}
\begin{aligned}
\label{Eq:Omega1Omega2}
    \Omega_1(\omega)&:=\biggl\{\sup_{v\in\mathbb{R}^d} \|g_n(v)\|\leq 2(r+\sqrt{d})\biggr\}\cap \biggl\{\sup_{\substack{\|v\|\leq 2(r+\sqrt{d})\\ v\neq 0}}\frac{\|\Delta_{n_\mathrm{U}}(v)\|}{\|v\|} \leq \omega\biggr\}\\
    \Omega_2(\phi)&:=\bigl\{\|\hat\mu_{n_{\mathrm{L}}} - \mu^*\|\leq \phi\bigr\}.
    \end{aligned}
\end{equation}
\begin{prop}
\label{Prop:ExceptionalEvents}
There exists $C_r > 0$, depending only on $r$,  such that for any $\delta \in (2e^{-n}, 1]$ and $\omega = C_r\sqrt{\frac{d\log n+\log(1/\delta)}{n_{\mathrm{U}}}}$,  we have 
$\mathbb{P}\bigl(\Omega_1(\omega)^{\mathrm{c}}\bigr) \leq \delta$.
Moreover, for any $\delta \in (0,1]$ and for $\phi = \sqrt{\frac{2d+3\log(1/\delta)}{n_{\mathrm{L}}}}$, we have $\mathbb{P}\bigl(\Omega_2(\phi)^{\mathrm{c}}\bigr) \leq \delta$.
\end{prop}
\begin{proof}
For any $v \in \mathbb{R}^d$,
\begin{align*}
\|g_n(v)\| = \|(1-\gamma)f_{n_\mathrm{U}}(v) + \gamma \hat{\mu}_{n_\mathrm{L}}\| &\leq (1-\gamma)\cdot \frac{1}{n_{\mathrm{U}}}\sum_{i=n_\mathrm{L}+1}^n \|Z_i\| + \frac{\gamma}{n_{\mathrm{L}}} \sum_{i=1}^{n_{\mathrm{L}}} \|Z_i\| \\
&= \frac{1}{n}\sum_{i=1}^{n} \|Z_i\| \leq \biggl(\frac{1}{n}\sum_{i=1}^{n} \|Z_i\|^2\biggr)^{1/2}.
\end{align*}
Since $\sum_{i=1}^{n}\|Z_i\|^2 \sim \chi^2_{nd}(ns^2)$, by  \citet[Lemma~8.1]{Birge2001}, we have with probability at least $1-\delta/2$ that 
\begin{align}
\label{Eq:BirgeBound}
\sup_{v\in\mathbb{R}^d} \|g_n(v)\|^2 
&\leq d + s^2 + 2\sqrt\frac{(d+2s^2)\log(2/\delta)}{n} + \frac{2\log(2/\delta)}{n} \nonumber\\
&\leq 2d+3s^2+\frac{3\log(2/\delta)}{n} \leq 4(r+\sqrt{d})^2.
\end{align}
Also, by a very similar argument as in the proof of~\citet[Theorem~4]{WuZhou2019}, we have with probability at least $1-\delta/2$ that 
\begin{equation}
\label{Eq:WuZhouBound}
    \sup_{\substack{\|v\|\leq 2(r+\sqrt{d})\\ v\neq 0}}\frac{\|\Delta_{n_\mathrm{U}}(v)\|}{\|v\|} \leq C_r\sqrt\frac{d\log n+\log(1/\delta)}{n_{\mathrm{U}}},
\end{equation}
for some $C_r>0$ depending only on $r$. The first claim follows by combining~\eqref{Eq:BirgeBound} and~\eqref{Eq:WuZhouBound}. 

For the second claim, we have $\hat\mu_{n_{\mathrm{L}}} \sim N_d(\mu^*, n_{\mathrm{L}}^{-1} I_d)$. Hence, by \citet[][Lemma~1]{LaurentMassart2000}, we have 
\begin{align*}
\mathbb{P}\bigl(\Omega_2(\phi)^{\mathrm{c}}\bigr) 
&= \mathbb{P}\bigl(n_{\mathrm{L}}\|\hat\mu_{n_{\mathrm{L}}} - \mu^*\|^2 > n_{\mathrm{L}}\phi^2\bigr) \\
&\leq \mathbb{P}\bigl(n_{\mathrm{L}}\|\hat\mu_{n_{\mathrm{L}}} - \mu^*\|^2 >  d+2\sqrt{d\log(1/\delta)} + 2\log(1/\delta) \bigr)\leq \delta,
\end{align*}
as required.
\end{proof}

For any $a\in\mathbb{R}$, $b\in[0,\infty)$ and $\xi\in\mathbb{S}^{d-1}$ that is orthogonal to $\eta$, we define $F(a,b) := \eta^\top f(a\eta + b\xi)$ and $G(a,b):=\|(I_d - \eta\eta^\top)f(a\eta + b\xi)\|$. Note that the distribution of $Z_1$ is orthogonally invariant along the axis $\mu^*$; in other words, if $P \in \mathbb{R}^{d \times d}$ is orthogonal and has $\mu^*$ as an eigenvector with eigenvalue 1, then $PZ_1 \stackrel{\mathrm{d}}{=} Z_1$.  It follows that $f(a\eta + b\xi)$, and hence $F(a,b)$ and $G(a,b)$, do not depend on $\xi$. We remark that
\[
f(\alpha_t\eta+\beta_t\xi_t) = F(\alpha_t,\beta_t)\eta + G(\alpha_t,\beta_t)\xi_{t+1}'
\]
for some $\xi_{t+1}' \in \mathbb{S}^{d-1}$ that is  orthogonal to $\eta$.

Proposition~\ref{p1} controls the magnitude of the component $\beta_t$ of the EM algorithm iterates that is orthogonal to the signal direction $\eta$.  We define $\zeta := \omega\gamma^{-1/2} \wedge \omega^{1/2}$.
\begin{prop} \label{p1}
Assume that $\phi\gamma^{1/2}\leq \omega \leq \min\{1/12, 1/(r+3)\}$ and that $\|\hat{\mu}^{(0)}\| \leq r+3$.  On the event $\Omega_1(\omega)\cap\Omega_2(\phi)$, we have
\[
\limsup_{t\to\infty}\beta_t \leq  60 (\zeta \vee  r\omega).
\]
Moreover, on the same event, if $\beta_{t_0} \leq 60(\zeta\vee r\omega)$ for some $t_0 \in \mathbb{N}_0$, then $\beta_t\leq 60(\zeta\vee r\omega)$ for all $t \geq t_0$.
\end{prop}
\begin{proof}
We first claim that on the event  $\Omega_1(\omega)\cap\Omega_2(\phi)$, we have $\|\hat{\mu}^{(t)}\|\leq r+3$ for all $t\in\mathbb{N}_0$. The case $t=0$ is true by the assumption on the initializer $\hat{\mu}^{(0)}$, and if the claim holds for $t\in\mathbb{N}_0$, then since $2(r+\sqrt{d}) \geq 2(r+1) \geq r+3$, we have on $\Omega_1(\omega)\cap\Omega_2(\phi)$ that
\begin{align*}
\|\hat{\mu}^{(t+1)}\| &\leq (1-\gamma)\bigl\{|F(\alpha_{t},\beta_{t})| + G(\alpha_{t},\beta_{t}) + \|\Delta_{n_{\mathrm{U}}}(\hat{\mu}^{(t)})\|\bigr\} + \gamma(s+\phi)\\
&\leq s+2\sqrt{2/\pi}+\omega\|\hat{\mu}^{(t)}\|+\gamma \phi \leq r + 2 + \frac{\|\hat{\mu}^{(t)}\|}{r+3} \leq r+3,
\end{align*}
where the second inequality uses \citet[Lemma~5(5)]{WuZhou2019}.  
Moreover, from~\eqref{Eq:EMIteration}, we have on $\Omega_1(\omega)\cap\Omega_2(\phi)$ that for $t \in \mathbb{N}$,
\begin{align}
\beta_{t+1} &= \bigl\|(I_d-\eta\eta^\top)\bigl\{(1-\gamma)\bigl(f(\hat{\mu}^{(t)})+\Delta_{n_{\mathrm{U}}}(\hat{\mu}^{(t)})\bigr)+\gamma\hat{\mu}_{n_{\mathrm{L}}}\bigr\}\bigr\|\nonumber\\
&\leq (1-\gamma)\bigl\{G(\alpha_t,\beta_t) + \omega(|\alpha_t|+\beta_t)\bigr\}+\gamma\phi\nonumber\\
&\leq \beta_t(1-\gamma)\biggl\{1+\omega-\frac{(\alpha_t^2+\beta_t^2)\wedge 1}{6}\biggr\} + \gamma\phi + \omega|\alpha_t|,\label{Eq:BoundonBeta}
\end{align}
where the final bound uses \citet[Lemma~5(8)]{WuZhou2019}. 
If $\alpha_t^2+\beta_t^2> 1$ or $\gamma>1/2$, then using the fact that $\omega\leq 1/12$, we have from \eqref{Eq:BoundonBeta} that
\begin{equation}
\label{Eq:BetaRecursion1}
\beta_{t+1}\leq \frac{11}{12}\beta_t+\gamma\phi + (r+3)\omega \leq  \frac{11}{12}\beta_t + (r+4)\omega. 
\end{equation}
On the other hand, if $\alpha_t^2+\beta_t^2\leq 1$ and $\gamma\leq 1/2$, then 
\begin{align}
\label{Eq:BetaRecursion2}
\beta_{t+1} & \leq \beta_t\biggl(1+\omega-\gamma - \frac{\alpha_t^2+\beta_t^2}{12}\biggr) + \gamma\phi+\omega|\alpha_t|.
\end{align}
Note that the right-hand side of~\eqref{Eq:BetaRecursion1} is increasing in $\beta_t$ and the right-hand side of~\eqref{Eq:BetaRecursion2} is increasing in $\beta_t$ for $\alpha_t^2+\beta_t^2\leq 1$ and $\gamma\leq 1/2$.  Combining~\eqref{Eq:BetaRecursion1} and~\eqref{Eq:BetaRecursion2}, denoting $\beta_\infty:=\limsup_{t\to\infty}\beta_t$ and using the fact that $0 \leq \frac{3}{\beta_\infty}(\omega - |\alpha_t|\beta_\infty/6)^2 = 3\omega^2/\beta_\infty - \omega|\alpha_t|+\alpha_t^2\beta_\infty/12$, we have
\begin{equation}
\label{Eq:BetaInf}
\beta_\infty \leq \max\biggl\{\frac{11}{12}\beta_\infty+(r+4)\omega,\, \beta_\infty \biggl(1+\omega-\gamma-\frac{\beta_\infty^2}{12}\biggr)+\gamma\phi+\frac{3\omega^2}{\beta_\infty}\biggr\}.
\end{equation}
From the first term in the maximum in~\eqref{Eq:BetaInf}, we obtain
\begin{equation}
\label{Eq:BetaBound0}
\beta_\infty \leq (1r+38)\omega \leq 60r\omega.
\end{equation}
From the second term in the maximum in~\eqref{Eq:BetaInf}, we obtain
\begin{equation}
\label{Eq:BetaInfSecondTerm}
\beta_\infty \biggl(\gamma-\omega+\frac{\beta_\infty^2}{12}\biggr) \leq \gamma\phi + \frac{3\omega^2}{\beta_\infty}.
\end{equation}
If $\gamma < 2\omega$, then from~\eqref{Eq:BetaInfSecondTerm},
\begin{equation}
\label{Eq:BetaBound1}
\beta_\infty\leq 5\omega^{1/2} \leq 5\sqrt{2}\zeta,
\end{equation}
since otherwise we would have that the left-hand side would be at least $(-5+125/12)\omega^{3/2}$ and the right-hand side would at most $(\sqrt{2}+3/5)\omega^{3/2}$, contradicting the inequality.  On the other hand, if $\gamma \geq 2\omega$, then we derive from~\eqref{Eq:BetaInfSecondTerm} that 
\[
\beta_\infty^2 - 2\phi \beta_\infty - \frac{6\omega^2}{\gamma} \leq 0.
\]
Solving this inequality, we find that
\begin{equation}
\label{Eq:BetaBound2}
\beta_\infty \leq \phi + \sqrt{\phi^2+6\omega^2/\gamma} \leq 4\omega\gamma^{-1/2} = 4\zeta.
\end{equation}
The first claim of the proposition follows by combining~\eqref{Eq:BetaInf},~\eqref{Eq:BetaBound0}, \eqref{Eq:BetaBound1} and~\eqref{Eq:BetaBound2}. We now prove the second claim by induction on $t$.
The base case $t=t_0$ is true by assumption, so we assume that $\beta_t\leq 60(\zeta\vee r\omega)$ for some $t \geq t_0$.  
Again we consider two cases.  If $\beta_t \leq 36(\zeta\vee r\omega)$, then from~\eqref{Eq:BoundonBeta} and using that $|\alpha_t| \leq \|\hat{\mu}^{(t)}\| \leq r+3$ for $t \geq 2$,
\[
\beta_{t+1}\leq \beta_t(1+\omega) + \gamma\phi+\omega|\alpha_t| \leq 44(\zeta\vee r\omega),
\]
as desired.  On the other hand, if $\beta_t > 36(\zeta\vee r\omega)$, then combining~\eqref{Eq:BetaRecursion1} and~\eqref{Eq:BetaRecursion2}, we obtain that
\begin{align*}
\beta_{t+1} \leq \max\biggl\{60(\zeta\vee r\omega),\, \beta_t \biggl(1+\omega-\gamma-\frac{\beta_t^2}{12}\biggr)+\gamma\phi+\frac{3\omega^2}{\beta_t}\biggr\}.
\end{align*}
It suffices to show the second term in the maximum is no larger than $\beta_t$. To this end, if $\gamma\leq 2\omega$, then $\zeta^2\leq \omega\leq 2\zeta^2$, and so
\begin{align*}
\beta_t \biggl(1&+\omega-\gamma-\frac{\beta_t^2}{12}\biggr)+\gamma\phi+\frac{3\omega^2}{\beta_t} \leq \beta_t + 2\zeta^2\beta_t - \frac{\beta_t^3}{12} + \gamma^{1/2}\omega + \frac{3\omega^2}{\beta_t}\\
&\leq \beta_t + 120(\zeta\vee r\omega)\zeta^2 - \frac{36^3(\zeta\vee r\omega)^3}{12} + 4\zeta^3 + \frac{\zeta^3}{3} \leq \beta_t.
\end{align*}
On the other hand, if $\gamma > 2\omega$, then $\zeta = \omega\gamma^{-1/2} \geq \phi$, and so
\begin{align*}
\beta_t\biggl(1+\omega-\gamma-\frac{\beta_t^2}{12}\biggr)+\gamma\phi+\frac{3\omega^2}{\beta_t} &\leq \beta_t - \frac{\gamma\beta_t}{2} + \gamma\phi +  \frac{3\omega^2}{\beta_t} \\
&\leq \beta_t - 18\gamma\zeta + \gamma \zeta + \frac{\omega^2}{12\zeta} \leq \beta_t,
\end{align*}
as desired, which completes the induction.
\end{proof}

The following result bounds the magnitude of the signal component, $\alpha_t$, of the EM iterates.

\begin{prop}\label{p2}
Assume that $\phi\gamma^{1/2}\leq \omega\leq \min\{1/12, 1/(r+3)\}$ and that $\|\hat{\mu}^{(0)}\|\leq r+3$.  
Then there exists $C_r > 0$, depending only on $r$, such that on the event $\Omega_1(\omega)\cap \Omega_2(\phi)$, we have
\[
\limsup_{t\to\infty}|\alpha_t| \leq C_r(\zeta \vee s).
\]
\end{prop}
\begin{proof}
By definition of $\alpha_{t+1}$ and~\eqref{Eq:EMIteration}, we have for every $t \in \mathbb{N}_0$ that
\[
\alpha_{t+1} = \eta^{\top} \bigl\{(1-\gamma)\bigl(f(\hat\mu^{(t)}) + \Delta_{n_{\mathrm{U}}}(\hat\mu^{(t)})\bigr)+\gamma \hat\mu_{n_{\mathrm{L}}} \bigr\}.
\]
Thus, by the first claim in the proof of Proposition~\ref{p1}, we have on the event $\Omega_1(\omega)\cap\Omega_2(\phi)$ that
\begin{equation}
\label{Eq:alphatplusone}
\bigl|\alpha_{t+1} - (1-\gamma) F(\alpha_t,\beta_t)-\gamma s\bigr| \leq  (1-\gamma)\omega(|\alpha_t|+\beta_t) + \gamma\phi
\end{equation}
for every $t \in \mathbb{N}_0$.  From~\citet[Lemma~5(1) and Lemma~5(7)]{WuZhou2019}, $\alpha \mapsto F(\alpha,\beta)$ is an increasing and odd function satisfying $|F(\alpha,\beta) - F(\alpha,0)|\leq (1+s^2)|\alpha|\beta^2$ for every $\alpha, \beta \in \mathbb{R}$.  
Hence, by~\eqref{Eq:alphatplusone}, we have on $\Omega_1(\omega)\cap\Omega_2(\phi)$ that
\begin{equation}
\label{Eq:AlphaUpdate}
|\alpha_{t+1}| \leq (1-\gamma)\Bigl\{F(|\alpha_t|,0) + (1+s^2)|\alpha_t|\beta_t^2 + \omega(|\alpha_t|+\beta_t)\Bigr\} + \gamma(s+\phi).
\end{equation}
Note the right-hand side of \eqref{Eq:AlphaUpdate} is increasing in $|\alpha_t|$. Define $\alpha_\infty := \limsup_{t\to\infty} |\alpha_t|$, so that $\alpha_\infty \leq r+3$ on $\Omega_1(\omega)\cap\Omega_2(\phi)$, again by the first claim in the proof of Proposition~\ref{p1}.  We may also assume that $\alpha_\infty > s$, because otherwise the result is clear. Since  $\alpha \mapsto F(\alpha,0)$ is continuous, we have from~\eqref{Eq:AlphaUpdate} that on $\Omega_1(\omega)\cap\Omega_2(\phi)$,
\begin{equation}
\label{Eq:alphainfinity}
\alpha_\infty \leq (1-\gamma)\bigl\{F(\alpha_\infty,0) + (1+r^2)\alpha_\infty\beta_\infty^2 + \omega(\alpha_\infty+\beta_\infty)\bigr\} + \gamma(s+\phi),
\end{equation}
where we recall that $\beta_\infty := \limsup_{t \rightarrow \infty} \beta_t$.  Define $q:[0,\infty) \rightarrow \mathbb{R}$ by
\begin{equation}
\label{Eq:q}
q(\alpha):= \left\{ \begin{array}{ll} F(\alpha,0)/\alpha & \mbox{if $\alpha \neq 0$} \\
1 + s^2 & \mbox{if $\alpha = 0$.} \end{array} \right.
\end{equation}
By Lemma~\ref{Lemma:q}, we have $q(s)=1$ (which confirms that $\mu_*$ is a fixed point of the population EM iteration), and that $q'(\alpha) \leq -c_r \alpha$ for all $\alpha \in (0,r]$, where $c_r\in(0,1]$ depends only on $r$.  Thus, dividing both sides of~\eqref{Eq:alphainfinity} by $\alpha_\infty$, we have
\begin{align}
\label{Eq:WaitingForContradiction}
1 &\leq (1-\gamma)\biggl\{q(s) + \int_s^{\alpha_\infty} q'(\alpha) \, d\alpha  + (1+r^2)\beta_\infty^2 + \omega\biggl(1+\frac{\beta_\infty}{\alpha_\infty}\biggr)\biggr\} + \frac{\gamma(s+\phi)}{\alpha_\infty} \nonumber \\
&\leq (1-\gamma)\biggl\{1-\frac{c_r}{2}(\alpha_\infty^2-s^2) + (1+r^2)\beta_\infty^2 + \omega\biggl(1+\frac{\beta_\infty}{\alpha_\infty}\biggr)\biggr\} + \frac{\gamma(s+\phi)}{\alpha_\infty}.
\end{align}
Now $\beta_\infty \leq 60(1+r)\zeta$ by Proposition~\ref{p1}.  We now claim that $\alpha_\infty\leq 4s+120c_r^{-1/2}(1+r)^2\zeta$.  Indeed, assuming the contrary, we would have $c_r\alpha_\infty^2/4 > c_r s^2/2+(1+r^2)\beta_\infty^2$ and $\beta_\infty/\alpha_\infty < 1$. Hence from~\eqref{Eq:WaitingForContradiction}, we have
\begin{align*}
1 &\leq (1-\gamma)\biggl(1-\frac{c_r\alpha_\infty^2}{4}\biggr) + \frac{\gamma s}{\alpha_\infty} + \biggl(2+\frac{\gamma^{1/2}}{\alpha_\infty}\biggr)\omega.
\end{align*}
We consider two cases. First, if $\gamma \leq 4\omega$, then $2\zeta \geq \omega^{1/2} \geq \gamma^{1/2}/2$ and hence
\begin{align*}
&(1-\gamma)\biggl(1-\frac{c_r\alpha_\infty^2}{4}\biggr) + \frac{\gamma s}{\alpha_\infty} + \biggl(2+\frac{\gamma^{1/2}}{\alpha_\infty}\biggr)\omega \\
& \leq \max\biggl\{1-\frac{c_r\alpha_\infty^2}{4}, 0\biggr\} + \frac{\gamma}{4} + 3\omega \leq \max\biggl\{1-\zeta^2, \frac{\gamma}{4}+3\omega\biggr\} < 1,
\end{align*}
a contradiction. Second, if $\gamma > 4\omega$, then $\zeta = \omega\gamma^{-1/2}$ and
\[
(1-\gamma)\biggl(1-\frac{c_r\alpha_\infty^2}{4}\biggr) + \frac{\gamma s}{\alpha_\infty} + \biggl(2+\frac{\gamma^{1/2}}{\alpha_\infty}\biggr)\omega \leq 1-\gamma + \frac{\gamma}{4} +2\omega + \frac{\gamma\zeta}{\alpha_\infty} < 1,
\]
again a contradiction.  This establishes the claimed upper bound on $\alpha_\infty$.
\end{proof}

Recall the definitions $L(\mu,\mu^*) = \|\mu-\mu^*\|\wedge \|\mu+\mu^*\|$.  Our next result shows that if $\alpha_t$ ever becomes sufficiently large, then improved bounds can be derived on the limiting behaviour of $\alpha_t$, $\beta_t$ and $L(\hat{\mu}^{(t)},\mu^*)$.

\begin{prop}\label{p3}
Assume that $\gamma \in [0, 1/2)$.  Given any $c>0$, there exists $C,c_1 > 0$, depending only on $r$ and $c$, such that if $|\alpha_{t_0}|\geq c s$, $\beta_{t_0}\leq 60(\zeta\vee r\omega)$ for some iteration $t_0$, and $\phi\gamma^{1/2}\leq \omega\leq c_1$ and $s\geq C\zeta$,
then on the event $\Omega_1(\omega)\cap \Omega_2(\phi)$, we  have
\begin{align}
\limsup_{t\to\infty}|\alpha_t -s | &\lesssim_{r,c} \frac{\omega}{s}\wedge \frac{\omega}{\gamma^{1/2}},\label{abo} \\
\limsup_{t\to\infty}\beta_t &\lesssim_{r,c}  \frac{\omega}{s}\wedge \frac{\omega}{\gamma^{1/2}},\label{bbo} \\
\limsup_{t\to\infty}L(\hat{\mu}^{(t)},\mu^*) &\lesssim_{r,c} \frac{\omega}{s}\wedge \frac{\omega}{\gamma^{1/2}}.\label{cbo}
\end{align}
\end{prop}
\begin{proof}
By flipping the sign of $\mu^*$ if necessary, we may assume without loss of generality that $\alpha_{t_0} \geq 0$ and that 
$c_1\leq \min\{1/12, 1/(r+3)\}$.  From~\eqref{Eq:alphatplusone} and the argument immediately below it, we have 
\begin{align}
\label{Eq:AlphaUpdate2}
(1-\gamma)\Bigl\{F(\alpha_{t},0)& - (1+s^2)\alpha_{t}\beta_{t}^2 - \omega(\alpha_{t}+\beta_{t})\Bigr\} + \gamma(s-\phi) \leq \alpha_{t+1} \nonumber\\
&\leq (1-\gamma)\Bigl\{F(\alpha_{t},0) + (1+s^2)\alpha_{t}\beta_{t}^2 + \omega(\alpha_{t}+\beta_{t})\Bigr\} + \gamma(s+\phi).
\end{align}
For any $t$ such that $\alpha_t \geq c s$, since $\beta_t \leq 60r\omega^{1/2}$ by Proposition~\ref{p1}, we have that  
\begin{equation}
\label{Eq:comegaalpha}
(1+s^2)\alpha_t\beta_t^2 + \omega(\alpha_t+\beta_t) \leq \biggl\{(1+r^2)60^2r^2\omega + \omega\biggl(1 + \frac{60s\zeta}{c s}\biggr)\biggr\}\alpha_t\leq c''\omega\alpha
\end{equation}
where $c'' := 60^2(1+r^2)r^2 + (1 + 60r c^{-1}C^{-1})$.  Moreover, if $\alpha_t\geq cs$, then
\begin{equation}
\label{Eq:gammaphi}
\gamma\phi \leq \omega\gamma^{1/2} \leq \begin{cases}
\gamma\zeta \leq C^{-1}\gamma s & \text{if $\gamma\geq \omega$}\\
\omega^{3/2} = \omega\zeta\leq c^{-1}C^{-1}\omega\alpha_t & \text{otherwise}.
\end{cases}
\end{equation}
Let $c':=c''+2c^{-1}C^{-1}$ and define functions $H, L: [0,\infty) \to \mathbb{R}$ by
 \begin{align*}
 H(\alpha)&:= (1-\gamma)\bigl\{F(\alpha,0) + c'\omega\alpha \bigr\} + \gamma (s+\phi), \\
 L(\alpha)&:= (1-\gamma)\bigl\{F(\alpha,0) -  c' \omega\alpha\bigr\} + \gamma \max(s-\phi, s/2),
 \end{align*}
From~\eqref{Eq:AlphaUpdate2}, \eqref{Eq:comegaalpha} and~\eqref{Eq:gammaphi}, we obtain that for $\alpha_t\geq c s$ and $C\geq 2$,  
\begin{equation}
\label{Eq:Sandwich}
L(\alpha_t)\leq \alpha_{t+1}\leq H(\alpha_t).
\end{equation}
Define auxiliary sequences $(\alpha_{t}^+)_{t \geq t_0}$ and $(\alpha_{t}^-)_{t \geq 0}$ by $\alpha_{t_0}^+ := \alpha_{t_0} =: \alpha_{t_0}^-$ and for $t \geq t_0$,
\[
\alpha_{t+1}^+ := H(\alpha_t^+) \quad \text{and}\quad 
\alpha_{t+1}^- := L(\alpha_t^-).
\]
We first derive some properties of the two recursion maps $H$ and $L$. For the former, we have by \citet[Lemma~3]{WuZhou2019} that $F$, and hence $H$, is increasing and concave on $[0,\infty)$ with $H(0)>0$ when $\gamma > 0$ and $H'(0) > \partial_1 F(0,0)>1$ when $\gamma = 0$. Moreover, since $F$ is bounded, we can choose $c_1 > 0$, depending only on $r$ and $c$, such that $\lim_{\alpha\to\infty} H'(\alpha) = (1-\gamma) c'\omega \leq (1-\gamma) c'c_1 < 1/2$.  On the other hand, we have $L(0)>0$ when $\gamma>0$. When $\gamma = 0$, we have $\omega^{1/2}=\zeta\leq s/C$, which means that after increasing $C \equiv C(r,c) > 0$ if necessary, $L'(0) = \partial_1 F(0,0) - c'\omega \geq 1+s^2-c's^2/C^2 > 1$. By \citet[Lemma~3]{WuZhou2019}, $\alpha \mapsto F(\alpha,0)$ is differentiable, increasing and concave for $\alpha \in[0,\infty)$. Reducing $c_1 \equiv c_1(r,c) > 0$ if necessary to ensure that $c_1\leq \partial_1 F(r+3,0)/c'$, we have for $\alpha \in [0,r+3]$ that
\[
L'(\alpha) = \partial_1 F(\alpha,0) - c' \omega \geq \partial_1 F(r+3,0) - c' c_1 \geq 0.
\]
In other words, $L$ is increasing on $[0,r+3]$, and moreover, similarly to $H$, it is also concave on this interval.  Finally, we claim that for $\tilde c := \min\bigl\{c, 32(3+r^4)/3\}^{-1/2}\bigr\}$, and $\tilde\alpha := \tilde c s$, we have $L(\tilde \alpha) \geq \tilde \alpha$. To verify this, we note by \citet[Lemma~3(3) and Equation~(98)]{WuZhou2019}, we have  
\begin{equation}
\label{Eq:Lalphatilde}
L(\tilde\alpha) \geq (1-\gamma)\tilde\alpha \biggl\{1 + s^2 - \frac{8}{3}(3+r^4)\tilde\alpha^2-c'\omega\biggr\} + 2\gamma \tilde\alpha \geq (1-\gamma)\tilde\alpha\biggl(1+\frac{3s^2}{4}-c'\omega\biggr) + 2\gamma\tilde\alpha.
\end{equation}
To control the right-hand side of~\eqref{Eq:Lalphatilde}, if $\gamma \leq c'\omega$, we have $\zeta = \omega^{1/2}\wedge \omega\gamma^{-1/2} \geq (\omega/c')^{1/2}$.  Hence, from the condition $s\geq C\zeta$, if we choose $C>2c'$ (which is possible because $c'$ is a decreasing function of $C$), we have $c'\omega \leq (c's/C)^2 \leq s^2/4$ and consequently the right-hand side of~\eqref{Eq:Lalphatilde} is at least $\tilde\alpha$.  If $\gamma > c'\omega$, then we have $L(\tilde\alpha) \geq (1-\gamma)^2\tilde\alpha + 2\gamma\tilde\alpha \geq \tilde \alpha$ as desired. This establishes the claim.

We now show by induction that for all $t\geq t_0$, 
\begin{equation}
\label{ab}
\tilde\alpha \leq \alpha_{t}^- \leq \alpha_{t} \leq \alpha_{t}^+. 
\end{equation}
The base case is clear by the definition of  $\alpha_{t_0}^-$ and $\alpha_{t_0}^+$ above.  Now suppose that~\eqref{ab} holds for some iteration $t \geq 0$, so in particular, \eqref{Eq:Sandwich} applies.

Using the monotonicity of $H$ and \eqref{Eq:Sandwich}, we have $\alpha_{t+1} \leq  H(\alpha_t) \leq H(\alpha_t^+) =  \alpha_{t+1}^+$.  
Observe that $\alpha_t\leq \|\hat\mu^{(t)}\|\leq r+3$ by the proof of Proposition~\ref{p1}. 
Using the monotonicity of $L$ on $[0,r+3]$, we find that  $\alpha_{t+1} \geq L(\alpha_t)\geq L(\alpha_t^-) =\alpha_{t+1}^-$. Moreover, 
\[
\alpha_{t+1}^- =  L(\alpha_t^-) \geq L(\tilde\alpha) \geq \tilde\alpha,
\]
which completes the induction.

To prove~\eqref{abo}, we will analyze the sequences $(\alpha_t^+)_{t \geq t_0}$ and $(\alpha_t^-)_{t \geq t_0}$, which sandwich $(\alpha_t)_{t\geq t_0}$.  We start by considering the behaviour of $(\alpha_t^+)_{t\geq t_0}$. The properties of $H$ derived above mean that we can apply Lemma~\ref{Lemma:H}  to obtain that $\alpha_{t}^+$ converges to a limit, denoted $\alpha^+$, satisfying $\alpha^+=H(\alpha^+)$. By Lemma~\ref{Lemma:q}, we have $F(s,0)=s$ and so $H(s) = (1-\gamma)(s+c'\omega s)+\gamma(s+\phi) > s$. Hence from Lemma~\ref{Lemma:H} again, we have $\alpha^+ > s$. On the other hand, since $F(\alpha,0) \leq \mathbb{E}|Z_{1,1}|\leq (\mathbb{E}Z_{1,1}^2)^{1/2}\leq (1+s^2)^{1/2}\leq 1+r$, we have $\alpha^+=H(\alpha^+) \leq (1-\gamma)(1+r)+\alpha^+/2 + \gamma r + 1 \leq r+2+\alpha^+/2$, so $\alpha^+ \leq 2r+4$. Recalling the definition of $q$ from~\eqref{Eq:q}, by Lemma~\ref{Lemma:q} again, we have 
\[
q(\alpha^+) = q(s)+\int_s^{\alpha^+} q'(\alpha)\,d\alpha \leq 1 - c_2  \bigl((\alpha^+)^2-s^2\bigr)
\]
for some $c_2>0$ depending only on $r$. Consequently, 
\begin{align*}
\alpha^+ = H(\alpha^+) &= (1-\gamma)\alpha^+\bigl\{q(\alpha^+)+c'\omega\bigr\}+\gamma(s+\phi) \\
&\leq(1-\gamma)\alpha^+\bigl\{1-c_2\bigl((\alpha^+)^2-s^2\bigr)+c'\omega\bigr\}+\gamma(s+\omega\gamma^{-1/2}),
\end{align*}
so 
\begin{equation}
\label{Eq:intermediate}
(\alpha^+)^2-s^2 \leq \frac{c'\omega}{c_2} - \frac{\gamma}{(1-\gamma)c_2}\frac{\alpha^+ - s-\omega\gamma^{-1/2}}{\alpha^+}.
\end{equation}
We now prove that 
\begin{equation}
\label{Eq:alphaplus}
\alpha^+-s \lesssim_{r,c} \frac{\omega}{s} \wedge \frac{\omega}{\gamma^{1/2}}
\end{equation}
by considering two cases. If $\alpha^+ \leq 2s$, then from~\eqref{Eq:intermediate}, we have
\[
\alpha^+-s \leq \frac{(1-\gamma)c'\omega + \gamma^{1/2}\omega/\alpha^+}{(1-\gamma)c_2(\alpha^++s)+\gamma/\alpha^+}\lesssim_{r,c} \frac{\omega}{s} \cdot \frac{1+\gamma^{1/2}/s}{1+\gamma/s^2} \lesssim \frac{\omega}{s}\wedge \frac{\omega}{\gamma^{1/2}}.
\]
On the other hand, if $\alpha^+ > 2s$, then we have from~\eqref{Eq:intermediate} again that
\begin{equation}
\label{Eq:intermediate2}
\frac{3(\alpha^+)^2}{4} +\frac{\gamma}{2c_2} \leq (\alpha^+)^2-s^2+\frac{\gamma(\alpha^+-s)}{(1-\gamma)c_2\alpha^+}\leq \frac{(c'+2\gamma^{1/2}/\alpha^+)\omega}{c_2}.
\end{equation}
In particular, $\gamma/2 \leq (c' + \gamma^{1/2}/s)\omega \leq c'\omega +C^{-1}(\gamma^{1/2}\omega^{1/2} \vee \gamma)$, so $\gamma\lesssim_{r,c} \omega$ and $\alpha^+\gtrsim_{r,c}\zeta\gtrsim_{r,c}\gamma^{1/2}$.  Consequently from~\eqref{Eq:intermediate2}, 
\[
\alpha^+-s \leq \alpha^+ \lesssim_{r,c} \omega^{1/2} \lesssim_{r,c} \frac{\omega}{\alpha^+} \wedge \frac{\omega}{\gamma^{1/2}} \lesssim \frac{\omega}{s} \wedge \frac{\omega}{\gamma^{1/2}},
\]
which establishes~\eqref{Eq:alphaplus}.  We now consider $(\alpha_t^-)_{t\geq t_0}$.  Define $\tilde L:[0,\infty)\to [0,\infty)$ by $\tilde L(\alpha) := L\bigl(\alpha \wedge (r+3)\bigr)$. Since $\alpha_t^- \leq \alpha_t \leq r+3$ for all $t \geq 0$, we have $\alpha_{t+1}^- = \tilde L(\alpha_t^-)$ for all $t \geq t_0$. From the properties of $L$ derived above, we see that $\tilde L$ satisfies the conditions of Lemma~\ref{Lemma:H}, and hence $\alpha_t^-$ converges to a limit, denoted $\alpha^-$, satisfying $\alpha^- = \tilde L (\alpha^-)  = L(\alpha^-)$. By Lemma~\ref{Lemma:q}, $F(s,0)=s$, so we have $\tilde L(s)= L(s) \leq (1-\gamma)(s-c'\omega s) + \gamma s < s$, so by Lemma~\ref{Lemma:H}, we must have $\alpha^- < s$. By Lemma~\ref{Lemma:q} again, we have
\[
q(\alpha^-) = q(s)-\int^s_{\alpha^-} q'(\alpha)\,d\alpha \geq 1 + c_2' \bigl(s^2-(\alpha^-)^2\bigr),
\]
where $c_2' > 0$ depends only on $r$. Consequently, we have
\[
\alpha^- = L(\alpha^-) \geq (1-\gamma)\alpha^-\bigl\{1+c_2' \bigl(s^2-(\alpha^-)^2\bigr) - c'\omega\bigr\} + \gamma (s-\omega\gamma^{-1/2}),
\]
which after rearranging and using the fact that $\alpha^- \geq \tilde\alpha \gtrsim_r s$ leads to 
\begin{equation}
\label{Eq:alphaminus}
s-\alpha^- \leq \frac{(1-\gamma)c'\omega + \gamma^{1/2}\omega/\alpha^-}{(1-\gamma)c_2'(s+\alpha^-)+\gamma /\alpha^-} \lesssim_{r,c} \frac{\omega}{s} \cdot \frac{1+\gamma^{1/2}/s}{1+\gamma/s^2} \lesssim \frac{\omega}{s}\wedge \frac{\omega}{\gamma^{1/2}}.
\end{equation}
Combining~\eqref{ab}, \eqref{Eq:alphaplus} and \eqref{Eq:alphaminus}, we have established~\eqref{abo}.

We now turn to prove \eqref{bbo}. By increasing $C$ if necessary, we have for all sufficiently large $t$ that $s/2\leq\alpha_{t}\leq 2s$. Consequently, we have by~\eqref{Eq:BoundonBeta} that for all large $t$,
\begin{equation}
\beta_{t+1} \leq \beta_{t}(1-\gamma)\biggl(1+\omega-\frac{s^2/4\wedge 1}{6}\biggr)+ \gamma\phi +  2s\omega\leq \beta_{t}(1-\gamma)(1+\omega-c_3s^2)+ (\gamma^{1/2} +  2s)\omega,\label{Eq:BoundonBeta2}
\end{equation}
with $c_3:=1/(6r^2+24)$. Denote $\beta_\infty :=\limsup_{t\to\infty}\beta_{t}$.  If $\gamma \leq 2\omega$, then $\zeta \in [(\omega/2)^{1/2},\omega^{1/2}]$ so $\omega \leq 2s^2/C^2 \leq c_3s^2/2$ for $C$ sufficiently large. Hence, from~\eqref{Eq:BoundonBeta2},
and the fact that $\gamma^{1/2}\leq 2\zeta\leq s$,
\[
\beta_\infty \leq \frac{(\gamma^{1/2}+ 2s)\omega}{c_3s^2/2} \lesssim_r \frac{\omega}{s}= \frac{\omega}{s}\wedge \frac{\omega}{\gamma^{1/2}}.
\] 
On the other hand, if $\gamma > 2\omega$, then $(1-\gamma)(1+\omega -c_3s^2)\leq 1-\gamma/2-c_3 s^2/2$ and from~\eqref{Eq:BoundonBeta2}, we obtain
\[
\beta_\infty \leq \frac{(\gamma^{1/2}+2s)\omega}{(\gamma+c_3 s^2)/2} \lesssim_r \frac{\omega}{s} \cdot \frac{1+\gamma^{1/2}/s}{1+\gamma/s^2} \lesssim \frac{\omega}{s}\wedge\frac{\omega}{\gamma^{1/2}}.
\]
Combining the above two cases establishes~\eqref{bbo}.

Finally, recalling the decomposition of $\hat\mu^{(t)}$ in~\eqref{Eq:OrthognoalDecomposition}, we see that \eqref{cbo} follows immediately from~\eqref{abo} and~\eqref{bbo}.
\end{proof}

Next, we show that provided the initialization is not too uncorrelated with the true parameter, $|\alpha_t|$ reaches a level that makes Proposition~\ref{p3} applicable after a sufficient number of iterations. 
\begin{prop}\label{p4}
Assume that $n\geq 3$, that $\phi\gamma^{1/2}  \leq \omega\leq \min\{1/12, 1/(r+3)\}$ and that $\gamma \in [0, 1/2)$.  Suppose that $\hat\mu^{(0)}$ is chosen such that $c'(\zeta \vee r\omega)\leq \|\hat\mu^{(0)}\| \leq 60(\zeta\vee r\omega)$ for some $c'\in(0,1)$ and that $|\langle \hat\mu^{(0)} / \|\hat\mu^{(0)}\|, \eta\rangle| \geq \sqrt{1/(d\log n_{\mathrm{U}})}$.   Then there exist $c_4, c_5 > 0$, depending only on $r$ and $c'$, such that if 
$s \geq c_4 \zeta \sqrt{d\log n_{\mathrm{U}}}$, then on $\Omega_1(\omega)\cap \Omega_2(\phi)$, we have  
\begin{align*}
|\alpha_t| \geq c_5 s
\end{align*}
for some $t > 0$. 
\end{prop}
\begin{proof}
By flipping the sign of $\mu^*$ if necessary, we may assume without loss of generality that $\alpha_0 \geq 0$. Assuming that the desired result is not true, we will prove by induction that on $\Omega_1(\omega)\cap\Omega_2(\phi)$, (a) $\alpha_t/\beta_t\geq c'/(60\sqrt{d\log n_{\mathrm{U}}})$ and (b) $\alpha_{t+1}\geq (1+\omega\sqrt{d\log n_{\mathrm{U}}})\alpha_t$ for all $t\geq 0$. We show this by first verifying the base case of (a), then proving that (a) implies (b) for each $t$, and finally proving that $\alpha_{t+1}/\beta_{t+1} \geq 1/(60\sqrt{d\log n_{\mathrm{U}}})$ once~(b) holds for a given $t$.  

For the base case, from the assumption on $\hat\mu^{(0)}$, we have 
\[
\frac{\alpha_0}{\beta_0} \geq \frac{\alpha_0}{\|\hat\mu^{(0)}\|} = \bigg\langle \frac{\hat\mu^{(0)}}{\|\hat\mu^{(0)}\|}, \eta\biggr\rangle \geq \frac{1}{\sqrt{d \log n_{\mathrm{U}}}} \geq \frac{c'}{60\sqrt{d \log n_{\mathrm{U}}}}
\]
since $c'\leq 60$.

Now assume that $\alpha_t/\beta_t\geq c'/(60\sqrt{d\log n_{\mathrm{U}}})$ and that $\alpha_0 \leq \alpha_t < c_5s$ for some $t\geq 0$. We aim to show that (b) holds for the same $t$, and start by controlling $e_1^\top f(\hat\mu^{(t)})$. Let $W = (W_1,\ldots,W_d)^\top$ be an independent copy of $Z_1$, independent of all other randomness in the problem, and define $u_t := \bigl\{\tanh(W_{1}\alpha_t + W_{-1}^\top \hat\mu^{(t)}_{-1}) - \tanh(W_{1}\alpha_t-W_{-1}^\top \hat\mu^{(t)}_{-1})\bigr\}/2$.  Then 
by applying the second part of Lemma~\ref{tanh-lem} with 
$a = W_{1} \alpha_t$ and $b=W_{-1}^\top \hat{\mu}^{(t)}_{-1}$ (so that $a + b = W^\top \hat{\mu}^{(t)}$), we have 
\begin{align*}
e_1^\top f(\hat \mu^{(t)})
&= \mathbb{E} \bigl\{W_{1} \tanh( W_{1}\alpha_t + W_{-1}^\top \hat{\mu}^{(t)}_{-1})\bigr\} \nonumber\\
&\geq \mathbb{E}\biggl\{\alpha_t W_1^2 - \frac{\alpha_t^3 W_1^4}{3} - \alpha_tW_1^2(W_{-1}^\top \hat\mu^{(t)}_{-1})^2\biggr\} +\mathbb{E} (W_1 u_t)\nonumber\\
&=  \alpha_t(1+s^2)(1-\beta_t^2) - \alpha_t^3 (1+2s^2+s^4/3),
\end{align*}
where in the final step we have used the fact that $u_t$ is an odd function of $W_{-1}^\top \hat\mu^{(t)}_{-1}$, which has a symmetric distribution about 0, conditional on $(W_1,\hat\mu^{(t)}_{-1})$, and hence $\mathbb{E}(W_1u_t) = \mathbb{E}\bigl\{\mathbb{E}(W_1u_t \mid W_1,\hat{\mu}_{-1}^{(t)})\bigr\} = 0$. 
From the assumption $s \geq c_4 \zeta $, and 
using $\beta_t\leq 60(1+r)\zeta$ from
Proposition~\ref{p1}, we have for sufficiently large $c_4$ that 
$\beta_t^2(1+s^2) \leq \{60(1+r)\}^2(1+r^2)s^2/c_4^2 \leq s^2/4$.
By choosing $c_5 > 0$, depending only on $r$, sufficiently small, we may assume that $\alpha_t^2(1+2s^2+s^4/3) < c_5^2(1+2r^2+r^4/3) s^2 \leq s^2/4$.  Recall the definition of $f_{n_{\mathrm{U}}}$ from~\eqref{eq:fnu}.  Since on the event $\Omega_1(\omega) \cap \Omega_2(\phi)$, we have $\|\hat{\mu}^{(t)}\| \leq r+3 \leq 2(r + \sqrt{d})$ as in the first line of the proof of Proposition~\ref{p1}, 
we have on the event $\Omega_1(\omega) \cap \Omega_2(\phi)$ that 
\begin{align*}
\alpha_{t+1}  &=(1-\gamma)e_1^\top f_{n_{\mathrm{U}}}(\hat\mu^{(t)}) + \gamma e_1^\top\hat\mu_{n_{\mathrm{L}}}\\
&\geq (1-\gamma)\bigl\{e_1^\top f(\hat\mu^{(t)}) - \omega\|\hat\mu^{(t)}\|\bigr\} + \gamma(s-\phi)\\
&\geq (1-\gamma)\bigl\{\alpha_t(1+s^2/2) - \omega(\alpha_t+\beta_t)\bigr\} + \gamma(s-\phi).
\end{align*}
If $\gamma \geq \omega$, then $\phi\leq \omega\gamma^{-1/2} = \zeta \leq s/2$ (assuming $c_4 \geq 2$). If $\gamma < \omega$, then 
\[
\gamma\phi\leq \omega\gamma^{1/2} < \omega^{3/2} = \omega\zeta \leq \frac{\alpha_0\omega}{c'}\sqrt{d\log n_{\mathrm{U}}}\leq \frac{2(1-\gamma)\alpha_t\omega}{c'}\sqrt{d\log n_{\mathrm{U}}}.
\]
Hence, in either case, we have on $\Omega_1(\omega) \cap \Omega_2(\phi)$ that
\begin{align*}
\alpha_{t+1}  &\geq (1-\gamma)\alpha_t\bigl\{1+s^2/2 - \omega(1+62/c')\sqrt{d\log n_{\mathrm{U}}})\bigr\} + \frac{\gamma s}{2}\nonumber\\
&\geq (1-\gamma)\alpha_t\bigl\{1+s^2/2+(1+62/c')(2\gamma - \omega\sqrt{d\log n_{\mathrm{U}}})\bigr\},
\end{align*}
where the final bound holds provided we  reduce $c_5$ to be at most $1/(2+124/c')$ if necessary.

Now, when $\gamma \leq \omega \sqrt{d\log n_{\mathrm{U}}}$,  we have $\zeta = \omega^{1/2}\wedge \omega\gamma^{-1/2} \geq \omega^{1/2}(d\log n_{\mathrm{U}})^{-1/4}$ and hence by the condition on $s$ in the proposition,  we have $s^2\geq c_4^2\zeta^2 d\log n_{\mathrm{U}} \geq c_4^2\omega \sqrt{d\log n_{\mathrm{U}}}$. Thus, by increasing $c_4$ to be at least $\sqrt{4+248/c'}$ if necessary, we have $(1+62/c')\omega\sqrt{d\log n_{\mathrm{U}}}\leq s^2/4$.
Hence, in this case, and on the event $\Omega_1(\omega) \cap \Omega_2(\phi)$,
\[
\alpha_{t+1} \geq (1-\gamma)(1+s^2/4+(2+124/c')\gamma)\alpha_t \geq \biggl(1+\frac{s^2}{8}\biggr)\alpha_t \geq (1+\omega\sqrt{d\log n_{\mathrm{U}}})\alpha_t.
\]
On the other hand, when $\gamma > \omega \sqrt{d\log n_{\mathrm{U}}}$, we have
\[
\alpha_{t+1}\geq (1-\gamma)(1+(1+62/c')\gamma) \alpha_t \geq (1+\gamma)\alpha_t \geq (1+\omega\sqrt{d\log n_{\mathrm{U}}})\alpha_t.
\]
Combining the two bounds above proves (b) for this given $t$. 

It remains to verify (a) for $t+1$, assuming that (a) and (b) hold up to and including $t$.  Since $\beta_{0}\leq \|\hat\mu^{(0)}\| \leq 60(\zeta\vee r\omega)$, we have 
by Proposition~\ref{p1} that $\beta_{t+1}\leq 60(\zeta\vee r\omega) \leq 60(c')^{-1}\|\hat\mu^{(0)}\|$. Thus,
\[
\frac{\alpha_{t+1}}{\beta_{t+1} } \geq \frac{\alpha_0}{60\|\hat\mu^{(0)}\|/c'} \geq \frac{c'}{60\sqrt{d\log{n_{\mathrm{U}}}}},
\]
which completes the induction. In particular, the geometric growth of $\alpha_t$ implied by (b) means that $\alpha_t$ will exceed $c_5 s$ for sufficiently large $t > 0$.  This establishes our desired contradiction, and hence proves the result.
\end{proof}

\begin{proof}[Proof of Proposition~\ref{Prop:OneGoodInitialization}]
Define $\phi_0:=\omega_0\gamma^{-1/2}$, and recall the definitions of $\Omega_1(\omega)$ and $\Omega_2(\phi)$ from~\eqref{Eq:Omega1Omega2}.    By Proposition~\ref{Prop:ExceptionalEvents}, there exists $C_r \geq 1$, depending only on $r$, such that for $\omega=C_r\omega_0$ and $\phi=C_r\phi_0$, we have  $\mathbb{P}\bigl(\Omega_1(\omega)\cap\Omega_2(\phi)\bigr)\geq 1-2\delta$. 

\emph{(i)} By the definition of $\omega$ and $\phi$, we have $\omega=\phi\gamma^{1/2}$. If we choose $c$ such that $c\leq C_r^{-1}\min\{1/12, 1/(r+3)\}$, then $\omega\leq \min\{1/12,1/(r+3)\}$. Thus, we may apply Propositions~\ref{p1} and~\ref{p2} to obtain that on $\Omega_1(\omega)\cap\Omega_2(\phi)$, we have 
\[
\limsup_{t\to\infty} \|\hat\mu^{(t)} - \mu^*\| \leq \limsup_{t\to\infty}(|\alpha_t| + \beta_t + \|\mu^*\|) \lesssim_r \zeta \vee \|\mu^*\|.
\]
The first claim follows.

\medskip

\emph{(ii)} From Lemma~\ref{Lemma:BetterThanWZLemma6} and by considering the case $d=1$ separately, for the chosen $\eta_0$, we have 
\[
\mathbb{P}\bigl(|e_1^\top \eta_0| \leq 1/\sqrt{d\log n_{\mathrm{U}}}\bigr) \leq \sqrt\frac{2}{\pi\log n_{\mathrm{U}}}. 
\]
Again, if we choose $c\leq C_r^{-1}\min\{1/12, 1/(r+3)\}$, then $\phi\gamma^{1/2}=\omega\leq \min\{1/12,1/(r+3)\}$. Also, $\|\hat\mu^{(0)}\| = \zeta_0\vee r\omega_0 \in [C_r^{-1}(\zeta\vee r\omega), \zeta\vee r\omega]$. Thus, applying Proposition~\ref{p4} with $c'=1/C_r$, there exists $c>0$, depending only on $r$, and $t_0\in\mathbb{N}$ such that on $\Omega_1(\omega)\cap\Omega_2(\phi)\cap \{|e_1^\top \eta_0| > 1/\sqrt{d\log n_{\mathrm{U}}}\}$, we have  $|\alpha_{t_0}| \geq c s$. 

Since $\beta_0 \leq \|\hat\mu^{(0)}\| \leq \zeta\vee r\omega$, we can apply Proposition~\ref{p1} to obtain that $\beta_t \leq 60(\zeta\vee r\omega)$ for all $t\geq 0$. Hence all conditions of Proposition~\ref{p3} are satisfied, and the desired result then follows from~\eqref{cbo}.
\end{proof}

To prove Theorem~\ref{Thm:LowDimEM}, we need the following proposition, which relates the loss of estimating $\mu^*$ to the operator norm loss of estimating $\mu^*\mu^{*\top}$.
\begin{prop}
\label{Prop:LossTransfer}
Assume that $n \geq 3$ and that $(Z_1,Y_1,Y_1^*),\ldots,(Z_n,Y_n,Y_n^*)$ are independent with
\[
Y_i^*\sim \mathrm{Unif}\{-1,1\},\; Z_i\mid Y_i^*\sim \mathcal{N}_d(Y_i^*\mu^*, I_d), \; Y_i = Y_i^*\mathbbm{1}_{\{i\leq n_{\mathrm{L}}\}}\quad \text{for $i\in[n]$.}
\]
For $\mu\in\mathbb{R}^d$ and $i \in [n]$, let
$L_{i}(\mu) := Y_i\mathbbm{1}_{\{Y_i\neq 0\}} + \tanh\langle Z_i, \mu\rangle\mathbbm{1}_{\{Y_i=0\}}$, $\mu_{\mathrm{tot}}(\mu) := \mu n^{-1}\sum_{i=1}^n L_i(\mu)$ and $\Sigma_{\mathrm{b}}(\mu) := \mu\mu^\top - \mu_{\mathrm{tot}}(\mu)\mu_{\mathrm{tot}}(\mu)^\top$.
For any $\delta\in(0,1)$ and $B > 0$, we have with probability at least $1-\delta$ that
\[
\sup_{\mu:\|\mu\|\leq B}
\bigl\{
\|\Sigma_{\mathrm{b}}(\mu) - \mu^*\mu^{*\top}\|_{\mathrm{op}} - (B+s)L(\mu,\mu^*)\bigr\}
\lesssim 
B^2(s^2\vee 1)\biggl(\frac{d\log(2Bn+e)+\log(1/\delta)}{n}\biggr).
\]
\end{prop}

\begin{proof}

For any $\mu \in \mathbb{R}^d$ with $\|\mu\| \leq B$, we have
\begin{align*}
\|\Sigma_{\mathrm{b}}(\mu) - \mu^*\mu^{*\top}\|_{\mathrm{op}} &\leq 
\|\mu\mu^\top - \mu^*\mu^{*\top}\|_{\mathrm{op}} + \|\mu_{\mathrm{tot}}(\mu)
\mu_{\mathrm{tot}}(\mu)^\top\|_{\mathrm{op}} \\
&\leq (B+s)L(\mu, \mu^*) + B^2\biggl(\frac{1}{n}\sum_{i=1}^n L_i(\mu)\biggr)^2.
\end{align*}
Thus, it is enough to show that $\sup_{\mu:\|\mu\|\leq B}$ $ n^{-1}\sum_{i=1}^n$ $L_i(\mu)$ $\lesssim \sqrt\frac{(s^2\vee 1)\{d\log(2Bn+e)+\log(1/\delta)\}}{n}$
with probability at least $1-\delta$. To this end, we have
\begin{align}
\sup_{\mu:\|\mu\|\leq B} \frac{1}{n}\sum_{i=1}^n L_i(\mu) &= \frac{1}{n}\sum_{i=1}^{n_{\mathrm{L}}} Y_i + \frac{1}{n}\sup_{\mu:\|\mu\|\leq B} \sum_{i=n_{\mathrm{L}}+1}^n \tanh\langle Z_i, \mu\rangle. \label{Eq:MeanLi}
\end{align}
For the first term on the right-hand side of~\eqref{Eq:MeanLi}, by Hoeffding's inequality, we have 
\begin{equation}
    \label{Eq:FirstDelta}
    \mathbb{P}\biggl(\frac{1}{n}\sum_{i=1}^{n_{\mathrm{L}}} Y_i > \frac{\sqrt{2n_{\mathrm{L}}\log(3/\delta)}}{n}\biggr)\leq \frac{\delta}{3}.
\end{equation}
For the second term on the right-hand side of~\eqref{Eq:MeanLi}, let $\mathcal{N}$ be a $\epsilon$-net of $\{v:\|v\|\leq B\}$ with respect to the Euclidean distance, for some $\epsilon \in (0,1/2]$ to be specified later. Since a maximal $\epsilon$-packing set is an $\epsilon$-net, we may assume that $|\mathcal{N}| \leq (B+\epsilon/2)^d / (\epsilon/2)^d= (2B/\epsilon + 1)^d$. Using the fact that $x\mapsto \tanh x$ is $1$-Lipschitz, together with the Cauchy--Schwarz inequality, we have  
\begin{align*}
\sup_{v:\|v\|\leq B}\sum_{i=n_{\mathrm{L}}+1}^n &\tanh\langle Z_i, v\rangle \\
&\leq \sup_{v\in\mathcal{N}}\sum_{i=n_{\mathrm{L}}+1}^n \tanh\langle Z_i, v\rangle + \sup_{u,v: \|u-v\|\leq \epsilon} \sum_{i=n_{\mathrm{L}}+1}^n (\tanh\langle Z_i, u\rangle - \tanh\langle Z_i, v\rangle) \\
&\leq \sup_{v\in\mathcal{N}}\sum_{i=n_{\mathrm{L}}+1}^n \tanh\langle Z_i, v\rangle + \epsilon \sum_{i=n_{\mathrm{L}}+1}^n \|Z_i\|.
\end{align*}
Hence taking $\epsilon = 1/n$, and defining $\tau := \frac{\log(3/\delta)}{d\log(2Bn+e)} > 0$, we have 
\begin{align}
    &\mathbb{P}\biggl(\frac{1}{n}\sup_{v:\|v\|\leq B}\sum_{i=n_{\mathrm{L}}+1}^n \tanh\langle Z_i, v\rangle > 2\sqrt\frac{2(s^2\vee 1)(1+\tau)d\log(2Bn+e)}{n}\biggr)\nonumber\\ 
    &\leq \mathbb{P}\biggl(\frac{1}{n}\sup_{v\in\mathcal{N}}\sum_{i=n_{\mathrm{L}}+1}^n \tanh\langle Z_i, v\rangle > \sqrt\frac{2(s^2\vee 1)(1+\tau)d\log(2Bn+e)}{n}\biggr) \nonumber \\
    &\hspace{3cm}+\mathbb{P}\biggl(\frac{1}{n}\sum_{i=n_{\mathrm{L}}+1}^n \|Z_i\| \geq \sqrt{2(s^2\vee 1)(1+\tau)nd\log(2Bn+e)}\biggr)\nonumber\\
    &\leq \frac{|\mathcal{N}|}{e^{(1+\tau)d\log(2Bn+e)}} + \mathbb{P}\biggl(\frac{1}{n}\sum_{i=n_{\mathrm{L}}+1}^n  \|Z_i\|^2 \geq 2(s^2\vee 1)(1+\tau)nd\log(2Bn+e)\biggr) \leq \frac{2\delta}{3},\label{Eq:SecondThirdDelta}
\end{align}
where the penultimate bound uses Hoeffding's inequality and the Cauchy--Schwarz inequality and the final bound uses the fact that $\sum_{i=n_{\mathrm{L}}+1}^n \|Z_i\|^2 \sim \chi^2_{n_{\mathrm{U}}d}(n_{\mathrm{U}}s^2)$ and 
\citet[Lemma~8.1]{Birge2001}. 
Combining~\eqref{Eq:MeanLi}, \eqref{Eq:FirstDelta} and~\eqref{Eq:SecondThirdDelta}, we have with probability at least $1-\delta$ that 
\begin{align*}
\sup_{\mu:\|\mu\|\leq B} \frac{1}{n}\sum_{i=1}^n L_i(\mu) &\leq \frac{\sqrt{2n_{\mathrm{L}}\log(3/\delta)}}{n} + 2\sqrt\frac{2(s^2\vee1)(1+\tau)d\log(2Bn+e)}{n} \\
&\lesssim \sqrt\frac{(s^2\vee 1)\{d\log(2Bn+e)+\log(1/\delta)\}}{n},
\end{align*}
as desired.
\end{proof}

\begin{proof}[Proof of Theorem~\ref{Thm:LowDimEM}]
We write $\hat\mu_{[m]} \equiv \hat\mu^{(T)}_{[m]}$ for the $T$th (final) iterate of the EM update in Algorithm~\ref{Algo:EM} starting from the $m$th random initializer $\hat\mu^{(0)}_{[m]}$. Let $\omega=C_r\omega_0$, $\phi=C_r\phi_0$, $\Omega_1(\omega)$ and $\Omega_2(\phi)$ be defined as in the proof of Proposition~\ref{Prop:OneGoodInitialization}.  Further, let $\Sigma_{\mathrm{b}}(\mu)$ be defined as in Proposition~\ref{Prop:LossTransfer}.  By Proposition~\ref{Prop:ExceptionalEvents}, the first claim in the proof of Proposition~\ref{p1} and Proposition~\ref{Prop:LossTransfer}, we have for some $C>0$ depending only on $r$ that 
\begin{align*}
\mathbb{P}\biggl[\max_{m\in[M]}\sup_{T\in\mathbb{N}} \Bigl\{\|\Sigma_{\mathrm{b}}(\hat\mu_{[m]}^{(T)}) - \mu^*\mu^{*\top}\|_{\mathrm{op}} - (2r+3)L(\hat\mu_{[m]}^{(T)}, \mu^*) \Bigr\} \leq &\frac{C\{d\log(rn)+\log(1/\delta)\}}{n}\biggr] \\
&\geq 1-\delta.
\end{align*}
In this balanced two-cluster setup, for the $t$th EM iteration starting from the $m$th random initializer, we have $-\hat\mu_1 = \hat\mu_2 = \hat\mu$, where we suppress the dependence on $t$ and $m$ for convenience.  For $i\geq n_{\mathrm{L}}+1$, we have
$L_{i,1} = e^{z_i^\top \hat\mu_1} / (e^{z_i^\top \hat\mu_1} + e^{z_i^\top \hat\mu_2})$, $L_{i,2} = e^{z_i^\top \hat\mu_2} / (e^{z_i^\top \hat\mu_1} + e^{z_i^\top \hat\mu_2})$ and hence $L_{i,2}-L_{i,1} = \tanh\langle Z_i, \hat\mu\rangle$. Thus, $\hat\mu_{\mathrm{tot}} = n^{-1}\hat\mu\sum_{i=1}^n \{ \tanh\langle Z_i,\hat\mu\rangle \mathbbm{1}_{\{Y_i=0\}} + Y_i\mathbbm{1}_{\{Y_i\neq 0\}} \}$. Also, we note that 
\[
\hat\Sigma_{\mathrm{b}} = \frac{1}{n}\sum_{i=1}^n \sum_{k=1}^2 L_{i,k}(\hat{\mu}_k - \hat{\mu}_{\mathrm{tot}})(\hat{\mu}_k - \hat{\mu}_{\mathrm{tot}})^\top = \frac{1}{n}\sum_{i=1}^n\sum_{k=1}^2 L_{i,k}\hat\mu_k\hat\mu_k^\top - \hat\mu_{\mathrm{tot}}\hat\mu_{\mathrm{tot}}^\top = \hat\mu\hat\mu^\top - \hat\mu_{\mathrm{tot}}\hat\mu_{\mathrm{tot}}^\top.
\]
Consequently, using the notation of Proposition~\ref{Prop:LossTransfer} and Algorithm~\ref{Algo:EM}, we have $\hat Q \equiv \hat Q^{(T)} \equiv \hat Q^{[\hat m]}= \Sigma_{\mathrm{b}}\bigl(\hat\mu_{[\hat{m}]}^{(T)}\bigr)$. 

We consider two cases. If $\|\mu^*\|\leq \omega_0^{1/3}\wedge\zeta_0^{1/2}$, then by the proof of Proposition~\ref{Prop:OneGoodInitialization}\emph{(i)}, we have on the event $\Omega_1(\omega)\cap \Omega_2(\phi)$ that  $\limsup_{T\to\infty} L(\hat\mu_{[m]}^{(T)}, \mu^*)\lesssim_r \zeta_0 \vee \|\mu^*\| \lesssim_r \omega_0^{1/3} \wedge \zeta_0^{1/2}$ for every $m\in[M]$. Thus, by Proposition~\ref{Prop:LossTransfer}, with probability at least $\mathbb{P}\bigl(\Omega_1(\omega)\cap\Omega_2(\phi)\bigr) -\delta \geq 1-3\delta$, we have 
\begin{align}
\limsup_{T\to\infty} \|\hat Q - \mu^*\mu^{*\top}\|_{\mathrm{op}} &\lesssim_r  (\omega_0^{1/3}\wedge \zeta_0^{1/2}) \limsup_{T\to\infty}  \max_{m\in [M]}
 L(\hat\mu_{[m]}^{(T)},\mu^*) 
 \nonumber\\
&\quad+ (\omega_0^{2/3}\wedge \zeta_0)\biggl(\frac{d\log n \!+\! \log(1/\delta)}{n}\biggr) \nonumber\\
&\lesssim_r
\omega_0^{2/3}\wedge \zeta_0 = \frac{\omega_0}{\omega_0^{1/3}\wedge \zeta_0^{1/2}} \wedge \zeta_0 \leq \frac{\omega_0}{\|\mu^*\|}\wedge \zeta_0.\label{Eq:Qbound1}
\end{align}

We now turn to the case where $\|\mu^*\| > \omega_0^{1/3}\wedge \zeta_0^{1/2}$. Let $\mathcal{M}_0$ be the set of $m\in[M]$ such that $|\langle \hat\mu^{(0)}_{[m]}, \mu^*\rangle| / (\|\mu^*\|\|\hat\mu^{(0)}_{[m]}\|) \geq \sqrt{1/(d\log n_{\mathrm{U}})}$ and let $M_0 := |\mathcal{M}_0|$. By definition of the EM initializers, the random variables $\bigl\{\langle \hat\mu^{(0)}_{[m]}, \mu^*\rangle / (\|\mu^*\|\|\hat\mu^{(0)}_{[m]}\|):m \in [M]\bigr\}$ are independent, and moreover, by Lemma~\ref{Lemma:BetterThanWZLemma6} we have 
\[
\mathbb{P}\biggl(\frac{|\langle \hat\mu^{(0)}_{[m]}, \mu^*\rangle|}{\|\mu^*\|\|\hat\mu^{(0)}_{[m]}\|}\geq \sqrt\frac{1}{d\log n_{\mathrm{U}}}\biggr)\geq 1-\sqrt\frac{2}{\pi\log n_{\mathrm{U}}} > \frac{3}{5}.
\]
Defining $\Omega_3:=\{M_0 > M/2\}$, by Hoeffding's inequality, we have 
\[
\mathbb{P}(\Omega_3^{\mathrm{c}}) \leq e^{-M/50}.
\]
Let 
\[
\mathcal{M}_1 := \bigl\{m \in [M] \setminus \{\hat{m}\}: \|\hat Q^{[m]} - \hat Q^{[\hat m]}\|_{\mathrm{op}} \leq \mathrm{median}(\|\hat Q^{[m']} - \hat Q^{[\hat m]}\|_{\mathrm{op}}: m'\in [M]\setminus\{\hat m\})\bigr\}.
\]
Since $|\mathcal{M}_1\cup\{\hat m\}|\geq \lceil (M-1)/2\rceil + 1 > M/2$, we have on $\Omega_3$ that $\mathcal{M}_0\cap(\mathcal{M}_1 \cup \{\hat m\})\neq \emptyset$. Thus, on the event $\Omega_1(\omega)\cap\Omega_2(\phi)\cap \Omega_3$, we can let $\tilde m := \min\bigl(\mathcal{M}_0\cap(\mathcal{M}_1 \cup \{\hat m\})\bigr)$, so by definition of $\hat m$, we have 
\begin{align*}
\|\hat Q^{[\hat m]} - \mu^*\mu^{*\top}\|_{\mathrm{op}} &\leq \|\hat Q^{[\hat m]} - \hat Q^{[\tilde m]}\|_{\mathrm{op}} + \|\hat Q^{[\tilde m]} - \mu^*\mu^{*\top}\|_{\mathrm{op}}\\
&\leq \mathrm{median}(\|\hat Q^{[m']} - \hat Q^{[\hat m]}\|_{\mathrm{op}}: m'\in [M]\setminus\{\hat m\}) + \|\hat Q^{[\tilde m]} - \mu^*\mu^{*\top}\|_{\mathrm{op}}\\
&\leq  \mathrm{median}(\|\hat Q^{[m']} - \hat Q^{[\tilde m]}\|_{\mathrm{op}}: m'\in [M]\setminus\{\tilde m\}) + \|\hat Q^{[\tilde m]} - \mu^*\mu^{*\top}\|_{\mathrm{op}}\\
&\leq \max_{m, m'\in\mathcal{M}_0} \|\hat Q^{[m]} - \hat Q^{[m']}\|_{\mathrm{op}} + \|\hat Q^{[\tilde m]} - \mu^*\mu^{*\top}\|_{\mathrm{op}}\\
&\leq 3\max_{m\in\mathcal{M}_0} \|\hat Q^{[m]} - \mu^*\mu^{*\top}\|_{\mathrm{op}}.
\end{align*}
Since $\omega_0 \leq (d\log n)^{-3}$, by discussing cases of $\gamma < \omega$, $\omega\leq \gamma\leq \omega^{2/3}$ and $\gamma>\omega^{2/3}$, we see that $\omega_0^{1/3}\wedge \zeta_0^{1/2} \geq \zeta_0\sqrt{d\log n}$. 
From the proof of Proposition~\ref{Prop:OneGoodInitialization}\emph{(ii)}, we have on the event $\Omega_1(\omega)\cap \Omega_2(\phi)$ that $\limsup_{T\to\infty}\max_{m\in\mathcal{M}_0}L(\hat\mu_{[m]}^{(T)}, \mu^*)\lesssim_r \frac{\omega_0}{\|\mu^*\|}\wedge (\omega_0\gamma^{-1/2})$. Let $\Omega_4$ be the event on which the conclusion of  Proposition~\ref{Prop:LossTransfer} holds.   Then on $\Omega_1(\omega)\cap\Omega_2(\phi)\cap\Omega_3\cap \Omega_4$, we therefore have
\begin{align}
\limsup_{T\to\infty} \|\hat Q^{[\hat m]} - \mu^*\mu^{*\top}\|_{\mathrm{op}}&\lesssim_r \limsup_{T\to\infty} \max_{m\in\mathcal{M}_0} L(\hat\mu_{[m]}^{(T)}, \mu^*) + \frac{d\log n+\log(1/\delta)}{n}\nonumber\\
&\lesssim_r \biggl(\frac{\omega_0}{\|\mu^*\|} \wedge \frac{\omega_0}{\gamma^{1/2}}\biggr) + \frac{d\log n+\log(1/\delta)}{n}\nonumber\\
&\lesssim_r \frac{\omega_0}{\|\mu^*\|} \wedge \zeta_0.\label{Eq:Qbound2}
\end{align}
The desired result follows by combining~\eqref{Eq:Qbound1} and~\eqref{Eq:Qbound2}, and the fact that $\mathbb{P}\bigl(\Omega_1(\omega)\cap\Omega_2(\phi)\cap\Omega_3\cap\Omega_4\bigr) \geq 1 - 3\delta - e^{-M/50}$.
\end{proof}

\subsection{Proof of Corollary~\ref{Cor:MultipleInitialisaiton}}
\label{pfCor:MultipleInitialisaiton}
\begin{proof}[Proof of Corollary~\ref{Cor:MultipleInitialisaiton}]
Fix $P\in\mathcal{P}_d$, define $Z_i := PX_i$ for $i \in [n]$, $\mu^*:=P\nu^*$, $\delta := \epsilon/\bigl\{4\binom{p}{d}\bigr\}$, $\omega_0 := \sqrt\frac{d \log n + \log(1/\delta)}{n_{\mathrm{U}}}$ and $\zeta_0 :=\omega_0^{1/2}\wedge \omega_0\gamma^{-1/2}$.  Then, provided $C_1 > 2$, we have $\delta \geq 2e^{-n/2} / p^d > 2e^{-n}$, and $\|\mu^*\|\leq \|\nu^*\|\leq r$.
Let $c>0$ be chosen, depending only on $r$, to satisfy Theorem~\ref{Thm:LowDimEM}.  By increasing $C_1 > 0$, depending only on $r$, if necessary, we may assume that $\omega_0 \leq \min\{c,(d \log n)^{-3}\}$. Hence, since $(P\Sigma_{\mathrm{w}}P^\top)^{-1}P\Sigma_{\mathrm{b}}P^\top = \mu^*\mu^{*\top}$, we can apply Theorem~\ref{Thm:LowDimEM} to obtain that for some $C_2'>0$ depending only on $r$, with probability at least $1-3\delta-e^{-M/50}$ we have that 
\begin{align*}
\limsup_{T \rightarrow \infty}
&\bigl\|\psi^{(T)}\bigl((PX_i, Y_i)_{i\in[n]}\bigr) - (P\Sigma_{\mathrm{w}}P^\top)^{-1}P\Sigma_{\mathrm{b}}P^\top\bigr\|_{\mathrm{op}} \leq C_2' \zeta_0 \\
&\leq C_2\min\biggl[\biggl\{\frac{d\log (p\vee n)+\log(1/\epsilon)}{n}\biggr\}^{1/4}, \sqrt\frac{d\log (p\vee n)+\log(1/\epsilon)}{n_{\mathrm{L}}}\biggr]\leq \frac{(\nu^*_{\min})^2}{4}.
\end{align*}
Since $\psi^{(T)}$ is permutation equivariant for  each $T\ge 0$,
by Fatou's lemma and a union bound, we have that
\begin{align*}
\limsup_{T \rightarrow \infty} {}&{}\mathbb{P}\biggl( \max_{P \in \mathcal{P}_d} \bigl\|\psi^{(T)}\bigl((PX_i, Y_i)_{i\in[n]}\bigr) - (P\Sigma_{\mathrm{w}}P^\top)^{-1}P\Sigma_{\mathrm{b}}P^\top\bigr\|_{\mathrm{op}} > \frac{(\nu^*_{\min})^2}{4}\biggr) \\
&\leq \mathbb{P}\biggl( \limsup_{T \rightarrow \infty} \max_{P \in \mathcal{P}_d} \bigl\|\psi^{(T)}\bigl((PX_i, Y_i)_{i\in[n]}\bigr) - (P\Sigma_{\mathrm{w}}P^\top)^{-1}P\Sigma_{\mathrm{b}}P^\top\bigr\|_{\mathrm{op}} > \frac{(\nu^*_{\min})^2}{4}\biggr) \\
&\leq \sum_{P\in\mathcal{P}_d} \mathbb{P}\biggl( \limsup_{T \rightarrow \infty} \bigl\|\psi^{(T)}\bigl((PX_i, Y_i)_{i\in[n]}\bigr) - (P\Sigma_{\mathrm{w}}P^\top)^{-1}P\Sigma_{\mathrm{b}}P^\top\bigr\|_{\mathrm{op}} > \frac{(\nu^*_{\min})^2}{4}\biggr) \\
&\leq \binom{p}{d}(3\delta+e^{-M/50})\leq \frac{3}{4}\epsilon + e^{-M/50 + d\log p} \leq \epsilon.
\end{align*}
The result now follows from Theorem~\ref{Thm:Meta}, noting that $\gamma_{\min} = (\nu_{\min}^*)^2$ and $\gamma_{\max} = (\nu^*_{\max})^2$.
\end{proof}

\section{Auxiliary lemmas}

\begin{lemma}
\label{Lemma:EMderiv}
Suppose that $K = 2$ and $\mathcal{C}$ is defined as in~\eqref{Eq:ConstraintSet}. Let $(-\hat \mu^{(t)}, \hat \mu^{(t)}, I_d) \in\mathcal{C}$ be the $t$th iterate of the EM iteration described in~\eqref{Eq:EStep} and~\eqref{Eq:MStep} with data $(Z_1,Y_1),\ldots,(Z_n,Y_n)$, starting from $(-\hat\mu^{(0)}, \hat\mu^{(0)}, I_d)$. Then for all $t\geq 1$, we have 
\[
\hat\mu^{(t)} = \frac{1}{n}\biggl\{\sum_{i:Y_i\neq 0} (-1)^{Y_i} Z_i + \sum_{i:Y_i=0}Z_i\tanh\langle Z_i, \hat\mu^{(t-1)}\rangle\biggr\}.
\]
\end{lemma}
\begin{proof}
    At step $t\geq 1$, in the E-step, by~\eqref{Eq:EStep}, we have for $k\in\{1,2\}$ that $L_{i,k} = \mathbbm{1}_{\{Y_i = k\}}$ if $Y_i\neq 0$ and 
    \[
    L_{i,k} = \frac{e^{-\|Z_{i}-(-1)^{k}\hat{\mu}^{(t-1)}\|^2/2} }{e^{-\|Z_{i}-\hat{\mu}^{(t-1)}\|^2/2} + e^{-\|Z_{i}+\hat{\mu}^{(t-1)}\|^2/2}}
    \]
    otherwise. In the M-step, defining
    \[
    Q(\mu\mid \hat\mu^{(t-1)}) := \frac{1}{n}\sum_{i=1}^n (L_{i,1}\|Z_i + \mu\|^2 + L_{i,2}\|Z_i-\mu\|^2),
    \]
    we have $\hat\mu^{(t)} = \argmin_{\mu\in\mathbb{R}^d} Q(\mu\mid\hat\mu^{(t-1)})$. Differentiating $Q(\mu\mid \hat\mu^{(t-1)})$ with respect to $\mu$, we obtain 
    \[
    \hat\mu^{(t)} = \frac{1}{n}\sum_{i=1}^n (L_{i,2}-L_{i,1})Z_i.
    \]
    The desired result follows since $L_{i,2}-L_{i,1} = (-1)^{Y_i}$ if $Y_i \in\{1,2\}$, and
    \[
    L_{i,2} - L_{i,1} = \frac{e^{-\|Z_{i}-\hat{\mu}^{(t-1)}\|^2/2} - e^{-\|Z_{i}+\hat{\mu}^{(t-1)}\|^2/2} }{e^{-\|Z_{i}-\hat{\mu}^{(t-1)}\|^2/2} + e^{-\|Z_{i}+\hat{\mu}^{(t-1)}\|^2/2}} = \frac{e^{\langle Z_i,\hat\mu^{(t-1)}\rangle} - e^{-\langle Z_i,\hat\mu^{(t-1)}\rangle}}{e^{\langle Z_i,\hat\mu^{(t-1)}\rangle}+e^{-\langle Z_i,\hat\mu^{(t-1)}\rangle}} = \tanh\langle Z_i, \hat\mu^{(t-1)}\rangle
    \]
    if $Y_i = 0$.
\end{proof}

\begin{lemma}
\label{Lemma:TestingRadius}
Let $X_1,\ldots,X_n\stackrel{\mathrm{iid}}{\sim} P$ for some distribution $P$ on $\mathbb{R}^d$. If $\|\mu^*\|\leq n^{-1/4}$, then for any Borel measurable function $\psi: (\mathbb{R}^d)^n \to \{0,1\}$ of the null hypothesis $H_0: P = \mathcal{N}_d(0,I_d)$ against the alternative $H_1: P = \frac{1}{2}\mathcal{N}_d(\mu^*,I_d) + \frac{1}{2}\mathcal{N}_d(-\mu^*,I_d)$, we have
\[
\mathbb{P}_{H_0}\bigl(\psi(X_1,\ldots,X_n) = 1\bigr) + \mathbb{P}_{H_1}\bigl(\psi(X_1,\ldots,X_n) = 0\bigr) > 1/2.
\]
\end{lemma}
\begin{proof}
Write $X = (X_1,\ldots,X_n)^\top$.  Observe that, writing $d_{\mathrm{TV}}$ for the total variation distance between probability measures,
\begin{align}
\mathbb{P}_{H_0}\bigl(\psi(X) = 1\bigr) &+ \mathbb{P}_{H_1}\bigl(\psi(X) = 0\bigr) \geq 1 - d_{\mathrm{TV}}(\mathbb{P}_{H_0}, \mathbb{P}_{H_1}) \nonumber\\
& = 1 - \frac{1}{2}\int \biggl|\frac{d\mathbb{P}_{H_1}}{d\mathbb{P}_{H_0}} - 1\biggr|\,d\mathbb{P}_{H_0} \geq 1 - \frac{1}{2}\biggl\{\int \biggl(\frac{d\mathbb{P}_{H_1}}{d\mathbb{P}_{H_0}} - 1\biggr)^2\,d\mathbb{P}_{H_0}\biggr\}^{1/2}\nonumber\\
 & = 1 - \frac{1}{2}\biggl\{\int \biggl(\frac{d\mathbb{P}_{H_1}}{d\mathbb{P}_{H_0}}\biggr)^2 \,d\mathbb{P}_{H_0} - 1\biggr\}^{1/2}.\label{Eq:ChiSquaredDivergence}
\end{align}
To control the chi-squared divergence in the right-hand side of~\eqref{Eq:ChiSquaredDivergence} above, we let $\xi = (\xi_1,\ldots,\xi_n)^\top$ have independent Rademacher components and $W = (W_{i,j})_{i\in[n],j\in[d]}$ be a random matrix with independent $N(0,1)$ entries, independent of $\xi$. Then $X\stackrel{\mathrm{d}}{=} W$ under $H_0$ and $X\stackrel{\mathrm{d}}{=} \xi \mu^{*\top} + W$ under $H_1$. Let $\tilde\xi$ be an independent copy of $\xi$. Using the Ingster--Suslina device, see, e.g.,
\cite{ingster2012nonparametric}, 
\cite{liu2021minimax}, Lemma~21, we have that 
\[
\int \biggl(\frac{d\mathbb{P}_{H_1}}{d\mathbb{P}_{H_0}}\biggr)^2 \,d\mathbb{P}_{H_0} = \mathbb{E}\exp\langle \xi \mu^{*\top}, \tilde\xi \mu^{*\top}\rangle = \cosh^n(\|\mu^*\|^2) \leq e^{n\|\mu^*\|^4/2}\leq e^{1/2},
\]
where we used the fact that $\cosh x \leq e^{x^2/2}$ for all $x\in\mathbb{R}$ in the penultimate step. The desired result follows from substituting the above bound into~\eqref{Eq:ChiSquaredDivergence} and the fact that $1-(e^{1/2}-1)^{1/2}/2 > 1/2$.
\end{proof}
We prove a generalization of Cochran's theorem for quadratic forms of independent Gaussian random vectors with a common covariance matrix, which result in independent noncentral Wishart distributions. Recall that if $X$ is a matrix, then $\mathrm{vec}(X)$ is the vectorization of $X$, obtained by stacking its columns on top of each other. The Kronecker product between matrices $A = (A_{i,j})_{i\in[m],j\in[n]}$ and $B$ is defined as 
\[
A\otimes B :=\begin{pmatrix}
A_{1,1}B& \cdots& A_{1,n}B\\
\vdots & \ddots & \vdots\\
A_{m,1}B & \cdots & A_{m,n}B
\end{pmatrix}.
\]
Recall also that when $X_1,\ldots,X_n\stackrel{\mathrm{iid}}{\sim} \mathcal{N}_d(0,\Sigma)$, the matrix $\sum_{i=1}^n X_iX_i^\top$ has a $d$-dimensional Wishart distribution with $n$ degrees of freedom and covariance matrix $\Sigma \in \mathbb{S}^{d \times d}$, denoted $\mathcal{W}_d(n,\Sigma)$.  More generally, $\sum_{i=1}^n (X_i+\mu_i)(X_i+\mu_i)^\top$ has a non-central Wishart distribution with $n$ degrees of freedom, covariance matrix $\Sigma$ and non-centrality matrix $\Omega = \sum_{i=1}^n\mu_i\mu_i^\top$, written $\mathcal{W}_d(n,\Sigma;\Omega)$.  Thus $\mathcal{W}_d(n,\Sigma;0) \stackrel{d}{=} \mathcal{W}_d(n,\Sigma)$.
\begin{lemma}
\label{Lem:CochranVariant}
Let $Z_1,\ldots,Z_n$ be independent with $Z_i\sim \mathcal{N}_d(\mu_i, \Sigma)$ for $i \in [n]$, and write $Z := (Z_1,\ldots,Z_n)^\top \in \mathbb{R}^{n \times d}$ and $M:=\mathbb{E}(Z)$. If $P_1,\ldots,P_k \in \mathbb{R}^{n \times n}$ are positive semidefinite matrices such that $P_1+\cdots+P_k=I_n$ and $\mathrm{rank}(P_1)+\cdots+\mathrm{rank}(P_k)=n$, then $Z^\top P_1 Z,\ldots,Z^\top P_kZ$ are independent with $Z^\top P_r Z\sim \mathcal{W}_d\bigl(\mathrm{rank}(P_r), \Sigma; M^\top P_rM\bigr)$. 
\end{lemma}
\begin{proof}
As in the proof of the classical Cochran's theorem \citep{cochran1934distribution}, we first note that $P_1,\ldots,P_k$ can be simultaneously diagonalised such that
\[
P_r = Q D_r Q^{\top},
\]
for some $Q \in \mathbb{O}^{n \times n}$ and $D_r = \mathrm{diag}\bigl((\mathbbm{1}_{\{j \in S_r\}})_{j \in [n]}\bigr)$, where $S_r \subseteq [n]$, $|S_r|=\mathrm{rank}(P_r)$ and $S_r\cap S_{r'} = \emptyset$ for all $r\neq r'$. In particular, $P_1,\ldots,P_k$ satisfy $P_r^2=P_r$ for $r \in [k]$ and $P_r P_{r'}=0$  for all $r\neq r'$. Since $P_1Z,\ldots,P_kZ$ are jointly Gaussian, with
\begin{align*}
  \mathrm{Cov}\bigl(\mathrm{vec}(P_rZ), \mathrm{vec}(P_{r'}Z)\bigr) &= \mathrm{Cov}\bigl((I_d\otimes P_r) \mathrm{vec}(Z), (I_d\otimes P_{r'} ) \mathrm{vec}(Z)\bigr) \\
  &= (I_d \otimes P_r)(\Sigma \otimes I_n)(I_d \otimes P_{r'})^\top = 0,
\end{align*}
we have that $P_1Z,\ldots,P_kZ$ are independent.  But $Z^\top P_r Z = (P_rZ)^\top P_rZ$, so it follows that $Z^\top P_1 Z,\ldots,Z^\top P_kZ$ are independent. Moreover, writing $W = (W_1,\ldots,W_n)^\top :=Q^{\top}Z$, we have $\mathrm{vec}(Z) \sim \mathcal{N}_{nd}\bigl(\mathrm{vec}(M),\Sigma \otimes I_n\bigr)$, so
\[
\mathrm{vec}(W) = (I_d \otimes Q^{\top})\mathrm{vec}(Z) \sim \mathcal{N}_{nd}\bigl((I_d \otimes Q^{\top})\mathrm{vec}(M),  \Sigma \otimes I_n\bigr) \stackrel{d}{=} \mathcal{N}_{nd}\bigl(\mathrm{vec}(Q^{\top}M), \Sigma \otimes I_n\bigr).
\]
Therefore,
\begin{align*}
  Z^\top P_r Z = W^\top  D_r W = \sum_{i\in S_r} W_iW_i^\top &\sim \mathcal{W}_d\biggl(|S_r|, \Sigma; \sum_{i\in S_r}\mathbb{E}(W_i)\mathbb{E}(W_i)^\top\biggr) \\
  &\stackrel{d}{=} \mathcal{W}_d\bigl(\mathrm{rank}(P_r), \Sigma; M^\top P_r M\bigr),
\end{align*}
as desired.
\end{proof}
\begin{lemma}\label{tanh-lem}
For any $a, b \in \mathbb{R}$, we have
\[
\frac{1}{2} \{\tanh(a+b) - \tanh(a-b)\} \leq |b|
\]
and
\[
\frac{a}{2} \{\tanh(a+b) + \tanh(a-b)\} \geq a^2 - \frac{a^4}{3} - a^2b^2
\]
\end{lemma}
\begin{proof}
For the first inequality, since the left-hand side is an increasing function of $b$, and an even function of $a$, we may assume that $a \geq 0$ and $b\geq 0$. Notice that $\frac{\partial}{\partial a}\bigl(\tanh(a+b) - \tanh(a-b)\bigr)= 1/\cosh^2(a+b) - 1/\cosh^2(|a-b|) \leq 0$, since $x\mapsto\cosh(x)$ is an increasing function on $[0,\infty)$. Hence
\[
\frac{1}{2} \{\tanh(a+b) - \tanh(a-b)\} \leq \tanh b \leq b.
\]
as desired. 

For the second inequality, since both sides are even functions of both $a$ and $b$, we may again assume without loss of generality that $a > 0$ and $b\geq 0$. We may also assume that $b\leq 1$ since otherwise, the right-hand side is negative and the inequality holds trivially.  But then 
\begin{align*}
\frac{1}{2} \{\tanh(a+b) + \tanh(a-b)\} &= \frac{1}{2}\biggl(\frac{\tanh a + \tanh b}{1 +\tanh a \tanh b} + \frac{\tanh a - \tanh b}{1 -\tanh a \tanh b}\biggr) \\
&= \frac{\tanh a}{(1 - \tanh^2 a \tanh^2 b)\cosh^2 b} \geq \frac{\tanh a}{\cosh^2 b} \\
&\geq (1-b^2)\tanh a \geq (1-b^2)\biggl(a - \frac{a^3}{3}\biggr) \geq a - \frac{a^3}{3} - ab^2,
\end{align*}
as desired. Here, the second inequality holds because $(1-b^2)\cosh^2 b \leq (1 - b^2)e^{b^2} \leq 1$.
\end{proof}

\begin{lemma}
\label{Lemma:H}
Let $H:[0,\infty) \rightarrow [0,\infty)$ be an increasing, concave function with $H'(x_0) < 1$ for some $x_0 \geq 0$ and either $H(0) > 0$ or both $H(0)=0$ and $H'(0) > 1$.  Then there exists a unique $\alpha^*>0$ such that
\[
H(\alpha) - \alpha \begin{cases} >0 & \alpha \in (0, \alpha^*)\\ =0 &\alpha = \alpha^* \\ < 0 & \alpha \in (\alpha^*, \infty).\end{cases}
\] 
Moreover, if $\alpha_0 > 0$, then the sequence $(\alpha_t)_{t \geq 0}$ given by $\alpha_t := H(\alpha_{t-1})$ monotonically converges to $\alpha^*$. 
\end{lemma}
\begin{proof}
For the first claim, consider the concave function $\tilde{H}(x) := H(x) - x$, which satisfies $\tilde{H}(x) > 0$ for sufficiently small $x > 0$, and for $x \geq x_0$, we have that any supergradient $v_x \in \mathbb{R}$ of $\tilde{H}$ at $x$ satisfies $v_x \leq -\bigl\{1 - H'(x_0)\bigr\} < 0$.  It follows that $\tilde{H}(x) \rightarrow - \infty$ as $x \rightarrow \infty$, so by the intermediate value theorem, there exists $\alpha^* \in (0,\infty)$ such that $\tilde{H}(\alpha^*) = 0$, i.e.~$H(\alpha^*) = \alpha^*$.  Again using the facts that $\tilde{H}(x) > 0$ for sufficiently small $x > 0$, and $\tilde{H}(x) \rightarrow - \infty$ as $x \rightarrow \infty$, we see that the concave function $\tilde{H}$ can only cross the $x$-axis at one positive value $\alpha^*$, and $\tilde H(\alpha) >0$ for $\alpha \in (0,\alpha^*)$ and $\tilde H(\alpha) < 0$ for $\alpha\in (\alpha^*,\infty)$. 

Next, note that if $\alpha \in (0, \alpha^*)$, then $\alpha < H(\alpha) < H(\alpha^*) = \alpha^*$.  Thus, if $\alpha_0 < \alpha^*$, then $(\alpha_t)_{t \geq 0}$ is an increasing sequence, bounded above by $\alpha^*$, so it converges to a limit.  But then, taking limits on both sides of the recursion $\alpha_t := H(\alpha_{t-1})$, we deduce that this limit must be $\alpha^*$.  A similar argument can be used to show that if $\alpha_0 \in (\alpha^*,\infty)$ then $(\alpha_t)_{t \geq 0}$ decreases down to the limit $\alpha^*$, while if $\alpha_0 = \alpha^*$, then $\alpha_t = \alpha^*$ for all $t$.
\end{proof}

\begin{lemma}
\label{Lemma:q}
Let $\mu^*$ be a non-zero vector in $\mathbb{R}^d$, let $Z \sim \frac{1}{2}\mathcal{N}_d(-\mu^*,I_d) + \frac{1}{2}\mathcal{N}_d(\mu^*,I_d)$, let $\eta := \mu^*/\|\mu^*\|$, and define $q:[0,\infty) \rightarrow [0,\infty)$ by
\[
q(\alpha) := \left\{ \begin{array}{ll} \alpha^{-1}\eta^\top \mathbb{E}(Z\tanh \langle \alpha \eta,Z\rangle) & \mbox{if $\alpha > 0$} \\
1+\|\mu^*\|^2 & \mbox{if $\alpha = 0$.} \end{array} \right.
\]
Then $q$ is a differentiable function with $q(\|\mu^*\|) = 1$ and for any $h \geq \|\mu^*\|$, we have
\[
\sup_{\alpha \in [0,h]} \frac{q'(\alpha)}{\alpha} \leq -\frac{e^{-h^2/2}}{3 \cdot 2^{11}\sqrt{2\pi}(h^5 \vee 1)}. 
\]
\end{lemma}
\begin{proof}
Write $s:=\|\mu^*\|$. The fact that $q(s) = 1$ follows from \citet[Theorem~1]{xu2016global}. By \citet[Lemma~3(4)]{WuZhou2019}, $q$ is differentiable with  $q'(\alpha)\leq -(2\alpha/3) \cdot \mathbb{E}\bigl(\tilde{Z}^4/\cosh^2(\alpha \tilde{Z})\bigr)$ for $\alpha \in [0,\infty)$, where $\tilde Z \sim \frac{1}{2}\mathcal{N}(-s,1)+ \frac{1}{2}\mathcal{N}(s,1)$. We can now compute that for $\alpha,s \in [0,h]$,
\begin{align*}
\mathbb{E}\biggl(\frac{\tilde{Z}^4}{\cosh^2(\alpha \tilde{Z})}\biggr) &\geq \mathbb{E}(\tilde{Z}^4 e^{-2\alpha |\tilde{Z}|}) \geq \frac{1}{2\sqrt{2\pi}} \int_0^\infty y^4e^{-2\alpha y}e^{-(y-s)^2/2} \, dy \\
&= \frac{1}{2\sqrt{2\pi}} \int_0^\infty y^4e^{-(y-s + 2\alpha)^2/2 -2\alpha s + 2\alpha^2} \, dy \\
&\geq \left\{ \begin{array}{ll} \frac{e^{-2\alpha s + 2\alpha^2}}{2\sqrt{2\pi}} \int_0^{\frac{1}{2(2\alpha-s)}} \frac{1}{2}y^4 e^{-(2\alpha-s)^2/2} \, dy & \mbox{if $s + 1 \leq 2\alpha$} \\
 \frac{e^{-2\alpha s + 2\alpha^2}}{2\sqrt{2\pi}} \int_{s - 2\alpha + 1}^{s - 2\alpha + 2} y^4 e^{-2} \, dy & \mbox{if $s + 1 > 2\alpha$} \end{array} \right. \\
&\geq \left\{ \begin{array}{ll} \frac{e^{-s^2/2}}{2^7\cdot 5\sqrt{2\pi}}\Bigl(\frac{1}{2\alpha-s}\Bigr)^5  & \mbox{if $s + 1 \leq 2\alpha$} \\
 \frac{e^{-2\alpha s + 2\alpha^2-2}}{10\sqrt{2\pi}} & \mbox{if $s + 1 > 2\alpha$} \end{array} \right. \\
&\geq \frac{e^{-h^2/2}}{2^{12}\cdot 5\sqrt{2\pi}(h^5\vee 1)},
\end{align*}
which establishes the desired bound. 
\end{proof}

\begin{lemma}
\label{Lemma:BetterThanWZLemma6}
Let $d \geq 2$, and let $\eta = (\eta_1,\ldots,\eta_d)^\top \sim \mathrm{Unif}(\mathbb{S}^{d-1})$. Then for any $a>0$, we have
\[
\mathbb{P}\biggl(|\eta_1| \leq \frac{a}{\sqrt{d}}\biggr) \leq \sqrt\frac{2}{\pi}a.
\]
\end{lemma}
\begin{proof}
Leting $Z = (Z_1,\ldots,Z_d)^\top \sim N_d(0, I_d)$, we have $\eta \stackrel{\mathrm{d}}{=} Z / \|Z\|$ and in particular $\eta_1^2 \stackrel{\mathrm{d}}{=} Z_1^2 / (Z_1^2 + \cdots + Z_d^2) \sim \mathrm{Beta}\bigl(1/2, (d-1)/2\bigr)$. Thus, 
\begin{align*}
    \mathbb{P}\biggl(|\eta_1| \leq \frac{a}{\sqrt{d}}\biggr) &= \mathbb{P}\biggl(\frac{Z_1^2}{\|Z\|^2} \leq \frac{a^2}{d}\biggr) = \frac{\Gamma(d/2)}{\Gamma(1/2)\Gamma\bigl((d-1)/2\bigr)} \int_0^{a^2/d} t^{-1/2}(1-t)^{(d-3)/2}\,dt\\
    &\leq \frac{2\Gamma(d/2)a}{\sqrt{d}\Gamma(1/2)\Gamma((d-1)/2)} \leq \sqrt\frac{2}{\pi}a,
\end{align*}
where the final bound uses, e.g., \citet[][Corollary~11]{dumbgen2021bounding}.
\end{proof}
{
\bibliographystyle{custom}
\bibliography{references}
}

\end{document}